\documentclass{lmcs}

\usepackage{bbding}
\usepackage{graphicx}   
\makeatletter 
\RequirePackage[bookmarks,unicode,colorlinks=true]{hyperref}   \def\@citecolor{blue}   \def\@urlcolor{blue}
   \def\@linkcolor{blue}

\makeatother 

\usepackage{bussproofs}
\usepackage{proof}
\usepackage{amsmath}
\usepackage{amssymb}
\usepackage{latexsym}  
\usepackage{epsf}
\usepackage{bbold}
\usepackage{stmaryrd}

\usepackage{amscd}

\usepackage{url} 
\usepackage{hyperref}

\newcommand{\NN}{\mathbf{N}}

\newcommand{\D}{\mathbf{D}}
\newcommand{\E}{\mathbf{E}}

\newcommand{\myC}{\mathbf{S}}
\newcommand{\myG}{\mathbf{G}}

\newcommand{\gtos}{\mathsf{gtos}}
\newcommand{\nh}{\mathsf{nh}}
\newcommand{\nnot}{\mathsf{not}}

\newcommand{\onedigit}{\mathsf{onedigit}}

 \newcommand{\gscomp}{\mathsf{gscomp}}
\newcommand{\Amb}{\mathbf{Amb}}

\newcommand{\Am}{\mathbf{A}}

\newcommand{\amb}{\mathbf{amb}}

\newcommand{\Fun}{\mathbf{Fun}}
\newcommand{\FV}{\mathrm{FV}} 
\newcommand{\Left}{\mathbf{Left}}

\newcommand{\Nil}{\mathbf{Nil}}
\newcommand{\Pair}{\mathbf{Pair}} 
\newcommand{\Right}{\mathbf{Right}}

\newcommand{\SD}{\mathbf{SD}}
\newcommand{\Set}{{\mathbf{\downdownarrows}}}

\newcommand{\False}{\mathbf{False}}

\newcommand{\tent}{\mathbf{t}}

\newcommand{\caseof}[1]{\mathbf{case}\, #1\, \mathbf{of}\, }

\newcommand{\botexp}{\mathbf{\bot}}

\newcommand{\leftright}{\mathsf{leftright}}
\newcommand{\mapamb}{\mathsf{mapamb}}
\newcommand{\aup}{\mathsf{up}}
\newcommand{\adown}{\mathsf{down}}

\newcommand{\re}{\mathbf{r}}

\newcommand{\defined}[1]{{#1}\ne \bot}     \newcommand{\eqdef}{\stackrel{\mathrm{Def}}{=}}
\newcommand{\eqmu}{\stackrel{\mu}{=}}
\newcommand{\eqnu}{\stackrel{\nu}{=}}
\newcommand{\eqmunu}{\stackrel{\munu}{=}}
\newcommand{\eqrec}{\stackrel{\mathrm{rec}}{=}}

\newcommand{\strictapp}[2]{#1{\downarrow}#2}

\newcommand{\ire}[2]{#1\,\mathbf{r}\,#2}

\newcommand{\inl}[1]{\Left(#1)}
\newcommand{\inr}[1]{\Right(#1)}

\newcommand{\all}[1]{\forall #1\,}
\newcommand{\ex}[1]{\exists #1\,}

\newcommand{\IFP}{\mathrm{IFP}}  \newcommand{\CFP}{\mathrm{CFP}}  \newcommand{\RIFP}{\mathrm{RIFP}}                                     \newcommand{\RCFP}{\mathrm{RCFP}}                                     
\newcommand{\rec}{\mathbf{rec}}

\newcommand{\tfix}[2]{\mathbf{fix}\,#1\,.\,#2}

\newcommand{\ssp}{\rightsquigarrow}
\newcommand{\newprintp}{\overset{\mathrm{p}}{\ssp}}
\newcommand{\newprintptr}{\mathbin{\stackrel{\mathrm{p}}{\ssp}\kern-.25em{}^*}}
\newcommand{\newprintc}{\overset{\mathrm{c}}{\ssp}}

\newcommand{\newprintpstard}[1]{\underset{#1}{\overset{\mathrm{p}}{\ssp}}^*}

\newcommand{\valu}[2]{\llbracket{#1}\rrbracket#2}
\newcommand{\val}[1]{\llbracket{#1}\rrbracket}   \newcommand{\rk}{\mathbf{rk}}
\newcommand{\cl}[1]{\mathrm{Pr}(#1)}

\newcommand{\ddata}{\mathrm{data}}

\newcommand{\Data}{\mathrm{Data}}

\newcommand{\rea}{\mathbf{R}}
\newcommand{\reah}{\mathbf{H}}

\newcommand{\dle}{\sqsubseteq}

\newcommand{\reali}[1]{\tilde{#1}}

\newcommand{\seq}{\mathbf{seq}}

\newcommand{\COIND}{\mathbf{COIND}}

\newcommand{\mon}{\mathsf{mon}}

\newcommand{\mycomment}[1]{}

\newcommand{\allex}{\Diamond}
\newcommand{\munu}{\Box}

\newcommand{\pcv}[1]{\hat{#1}}

\newcommand{\Mon}{\mathbf{Mon}}

\newcommand{\one}{\mathbf{1}}

\newcommand{\sgb}{\mathrm{sgb}}

\newcommand{\ax}{\mathcal{A}}

\newcommand{\indu}[1]{\mathbf{Ind}(#1)}
\newcommand{\coindu}[1]{\mathbf{CoInd}(#1)}
\newcommand{\induprime}[1]{\mathbf{Ind'}(#1)}

\newcommand{\resti}[2]{\mathbf{Rest}^+(#1,#2)}
\newcommand{\restb}[2]{\mathbf{RestBd}(#1,#2)}
\newcommand{\restret}[2]{\mathbf{RestRet}_{#2}(#1)}
\newcommand{\resta}[2]{\mathbf{RestAMon}(#1,#2)}
\newcommand{\restmp}[2]{\mathbf{RestMP}(#1,#2)}
\newcommand{\restefq}[1]{\mathbf{RestEFQ}_{#1}}
\newcommand{\rests}[1]{\mathbf{RestStab}(#1)}
\newcommand{\concl}[2]{\mathbf{ConcLEM}(#1,#2)}
\newcommand{\concret}[1]{\mathbf{ConcRet}(#1)}
\newcommand{\concmp}[2]{\mathbf{ConcMP}(#1,#2)}

\newcommand{\subd}{\delta^{\lhd}}            \newcommand{\subdom}[1]{#1 \lhd D}             
\newcommand{\tval}[2]{D^{#2}_{#1}}
\newcommand{\ftyp}[2]{#1 \Rightarrow #2}

\newcommand{\ep}[1]{\mathbf{ep}(#1)}
\newcommand{\epp}[1]{\mathbf{ep'}(#1)}
 \newcommand{\pt}[1]{\mathbf{pt}(#1)}  \newcommand{\monproof}[1]{\mathbf{Mon}_{#1}}

\newcommand{\monprop}[2]{\mathrm{Mon}_{#1}(#2)}

\newcommand{\tri}{\mathbf{3}}
\newcommand{\bool}{\mathbf{2}}
\newcommand{\nat}{\mathbf{nat}}

\newcommand{\stream}[1]{#1^\omega}

\newcommand{\toot}{\ \leftrightarrow \ }

\newcommand{\ConSD}{\mathbf{ConSD}}
\newcommand{\conSD}{\mathsf{conSD}}

\newcommand{\rt}[2]{#2|_{#1}}
\newcommand{\rtp}[2]{#2|'_{#1}}

\newcommand{\regD}{\mathrm{Reg}}
\newcommand{\cd}[1]{\mathrm{cd}_{#1}}
\newcommand{\variant}[2]{#1[#2]}

\newcommand{\datacons}{\mathrm{C_{data}}}

\newcommand{\adummy}[1]{\Delta(#1)}

\newcommand{\Dnh}{Definitely non-Harrop}
\newcommand{\dnh}{definitely non-Harrop}

\newcommand{\typeq}{\simeq}

\newcommand{\restintro}{Rest$^+$}

\usepackage{ifthen}

\theoremstyle{definition}\newtheorem{example}{Example} \date{}

\begin{document}

\title{Extracting total Amb programs from proofs}

\author{Ulrich Berger\lmcsorcid{0000-0002-7677-3582}}[a]
\author{Hideki Tsuiki\lmcsorcid{0000-0003-0854-948X}}[b]

\address{Swansea University, Swansea, UK}
\email{u.berger@swansea.ac.uk} 

\address{Kyoto University, Kyoto, Japan}
\email{tsuiki@i.h.kyoto-u.ac.jp}

\begin{abstract}

We present a logical system $\CFP$ (Concurrent Fixed Point Logic) supporting the extraction of nondeterministic and concurrent programs that are provably total and correct. $\CFP$ is an intuitionistic first-order logic with inductive and coinductive definitions extended by two propositional operators, $\rt{A}{B}$ (restriction, a strengthening of implication) and $\Set(B)$ (total concurrency). The source of the extraction is formal $\CFP$ proofs, the target is a lambda calculus with constructors and recursion extended by a constructor Amb (for McCarthy's amb) which is interpreted operationally as globally angelic choice and is used to implement nondeterminism and concurrency. The correctness of extracted programs is proven via an intermediate domain-theoretic denotational semantics. We demonstrate the usefulness of our system by extracting a nondeterministic program that translates infinite Gray code into the signed digit representation. A noteworthy feature of $\CFP$ is the fact that the proof rules for restriction and concurrency involve variants of the classical law of excluded middle that would not be interpretable computationally without Amb.
This is a revised and extended version of the conference paper~\cite{CFPesop}
that contains full proofs of all major results.

\end{abstract}

\maketitle              

\section{Introduction}
\label{sec-introduction}

Nondeterministic bottom-avoiding choice is an important and useful idea. 
With the wide-spread use of hardware supporting parallel computation,
it 
can speed up practical computation and, at the same time,
relates 
to computation over mathematical structures like real 
numbers~\cite{Escardo96,Tsuiki02}.
On the other hand, it is not easy to apply 
theoretical tools like denotational semantics to nondeterministic bottom-avoiding choice~\cite{HughesO89,Levy07}, and guaranteeing 
correctness and totality of such programs 
is a difficult task.

To explain the subtleties of the problem, let us start with an example.
Suppose that $M$ and $N$ are partial programs that,  
under the (in general undecidable) conditions $A$ and $\neg A$, respectively, 
are guaranteed to terminate and produce values satisfying a specification $B$. 
Then, by executing $M$ and $N$ in parallel 
and taking the result 
returned
first, we should always obtain a result satisfying $B$. 
This kind of bottom-avoiding nondeterministic program
is known as \emph{McCarthy's amb (ambiguous) operator} \cite{McCarthy1963}, and 
we denote such a program by  $\Amb(M, N)$.
$\Amb$ is called the angelic choice operator and
is usually studied as one of the three 
nondeterministic choice operators (the other two are erratic choice and demonic choice).

If one tries to formalize this idea naively, one will face some 
obstacles.  
Let $\ire{M}{B}$ (``$M$ realizes $B$'') 
express 
that program $M$ satisfies specification $B$, and 
let $\Set(B)$ be the specification that 
is satisfied by concurrent programs 
of the form $\Amb(M, N)$ that  always 
terminate and produce 
a value satisfying $B$.  
Then, the above inference could be written as
\[
  \infer[\hbox{}]{
  \ire{\Amb(M, N)} \Set(B)
}{
A \to (\ire{M}{B})  \ \ \ \     \neg A \to (\ire{N}{B}) 
}
\]
However, this inference is not sound for the following reason.
Suppose that $A$ does not hold, that is, $\neg A$ holds. 
Then, the execution of $N$ will produce a
value satisfying $B$. But the execution of $M$ may terminate as well, and
produce data that does not satisfy $B$ since there is no condition on $M$
if $A$ does not hold.
Therefore,  if $M$ terminates first in the execution of $\Amb(M, N)$, 
we obtain a result that may not satisfy $B$.

To amend this problem, we add a new operator $\rt{A}{B}$ 
(pronounced ``$B$ restricted to $A$'')
and consider the rule
\begin{equation}\label{eq0}
  \infer[\hbox{}]{
  \ire{\Amb(M, N)} \Set(B)
}{
\ire{M}{(\rt{A}{B})}\ \ \ \  \ire{N}{(\rt{\neg A}{B})}}
\end{equation}

Intuitively, $\ire{M}{(\rt{A}{B})}$ means two things:
(1) If $A$ holds, then $M$ terminates, and
(2) if 
$M$ terminates, then the result satisfies $B$,
even when $A$ does not hold.
As we will see in Section~\ref{sub-conc}, 
the above rule is derivable in classical logic
and can therefore be used to prove total correctness 
of Amb programs (Lemma~\ref{lem-restrict}).

In this paper, we go 
a step further and
introduce a logical system $\CFP$ whose formulas
can be interpreted as specifications of 
nondeterministic programs
although they do not explicitly refer to programs.
$\CFP$ is defined by adding the two logical operators $\rt{A}{B}$ and 
$\Set(B)$ to the system $\IFP$, 
a logic for program extraction~\cite{IFP} 
(see also \cite{Berger11,SeisenBerger12,BergerPetrovska18}). 
$\IFP$ supports the extraction of lazy functional programs 
from inductive/coinductive proofs in intuitionistic first-order logic.
It has a prototype implementation in Haskell,
called Prawf \cite{DBLP:conf/cie/0001PT20}.
A related approach
has been developed in the
 proof system Minlog~\cite{SchwichtenbergMinlog06,BergerMiyamotoSchwichtenbergSeisenberger11,SchwichtenbergWainer12}.
 
We show that from a $\CFP$ proof of a formula, both
a program and a proof that the program
satisfies the specification can be extracted
(Soundness theorem, Theorem \ref{thm-soundnessI}). 
For example, in $\CFP$ we have the rule 
 \begin{equation}\label{eq00}
   \infer[\hbox{(Conc-lem)}]{
   \Set(B)
 }{
 \rt{A}{B}\ \ \ \    \rt{\neg A}{B}}
  \end{equation}
which is realized by the program   $\lambda a. \lambda b. \Amb(a, b)$,
  and whose correctness is expressed by the rule (\ref{eq0}).
Programs extracted from $\CFP$ proofs can be  
executed in Haskell,
by implementing $\Amb$
with the Concurrent Haskell package.

As an application,
we extract a nondeterministic 
program that converts 
infinite Gray code into the signed digit representation, where 
infinite Gray code is a 
coding of real numbers by partial digit streams
that are allowed to contain a $\bot$, that is, a digit
whose computation does not terminate~\cite{Gianantonio99,Tsuiki02}. 
Partiality and multi-valuedness are common phenomena in computable analysis 
and exact real number computation~\cite{Weihrauch00,LUCKHARDT1977321}.
This case study connects these two aspects through a nondeterministic and
concurrent program whose correctness is guaranteed by a CFP proof. The extracted Haskell programs are listed in the appendix, and are also 
available in the repository~\cite{githubUB}.

In this paper, we focus on the specific aspect of concurrent computation and its resulting nondeterminism discussed above: the situation where two processes simultaneously attempt to access the same memory cell. Other aspects of concurrency, such as communication, are not covered here.

Organization of the paper: 
In Sections~\ref{sec-ang} and~\ref{sec-ops}
we present the denotational and operational semantics of a functional 
language with $\Amb$ and prove that they match 
(Theorems~\ref{thm:data} and~\ref{thm:dataconv}).
Sections~\ref{sec-cfp} and~\ref{sec-pe} describe the formal system $\CFP$ 
and its realizability interpretation 
on which our program extraction method is based 
(Theorems~\ref{thm-soundnessI},~\ref{thm-soundnessII}, and~\ref{thm-pe}).
In Section~\ref{sec-gray} we extract
a concurrent program that converts 
representations
 of real numbers 
and study its behaviour in Section~\ref{sec-experiments}.
The concluding Section~\ref{sec-conclusion} summarizes the results,
discusses related work, and
outlines further research directions.

\section{Denotational semantics of globally angelic choice }
\label{sec-ang}

In \cite{McCarthy1963}, 
McCarthy defined 
the ambiguity operator $\amb$ as
\[
\amb(x, y) = \left\{ \begin{array}{ll} x  &(x \ne \bot)\\
                       y & (y \ne \bot) \\
                       \bot & (x =  y = \bot)
                              \end{array}\right.
\]                              
where $\bot$ means `undefined' and the values $x$ and $y$ are 
taken nondeterministically 
if they
are not $\bot$.
This is called \emph{locally} angelic nondeterministic choice  
since convergence is chosen over divergence for each local call for 
the computation of $\amb(x, y)$.
It can be implemented by executing both of the arguments
in parallel and taking the result obtained first.  
Despite being a simple construction,
$\amb$ is known to have a lot of expressive power, 
and many constructions of nondeterministic and parallel computation 
such as erratic choice, countable choice (random assignment), 
and `parallel or' can be encoded through it \cite{LassenMoran99}.
These multifarious aspects of the operator $\amb$ are reflected by 
the difficulty of its mathematical treatment in denotational semantics.
For example, 
$\amb$ is not monotonic when interpreted over powerdomains
with the Egli-Milner order~\cite{Broy1986}.

Alternatively, 
one can consider an interpretation of $\amb$ as  
\emph{globally} angelic choice, where  an argument of $\amb$ 
is chosen so that the whole ambient computation converges 
if convergence 
is possible at all~\cite{ClingerHalpern85,SondergardSestoft92}.
It can be implemented by 
running the {\em whole} computation for both
of the arguments of $\amb$ in parallel and taking the 
result obtained first.
The difference between 
the locally and the globally angelic interpretation
of $\amb$ 
is highlighted
by the fact that the 
former 
does not commute with function application.
For example, if 
$f(0) = 0$ but $f(1)$ diverges,
then,
with the local interpretation,  
$\amb(f(0), f(1))$ will always terminate with the value $0$, whereas
$f(\amb(0, 1))$ may return 0 or diverge.
On the other hand, 
the latter term will always return $0$ 
if $\amb$ is implemented with a globally angelic semantics.
Though globally angelic choice is not defined compositionally,  
one can, as suggested in \cite{ClingerHalpern85},  
integrate it into the design of a programming 
language by using this commutation property.

Denotationally, globally angelic choice 
can be modelled by
the Hoare powerdomain construction. 
However, this would not be suitable for analyzing total correctness
  because the ordering of the Hoare powerdomain does 
not discriminate $X$ and $X \cup \{\bot\}$~\cite{HughesM92,HughesO89}.
Instead, we consider a two-staged approach (see~Section~\ref{sub-denot})
and %provide a simple denotational semantics for a language with globally angelic choice.
provide a simple denotational semantics for a language with an $\Amb$ operator which is operationally interpreted as globally angelic choice.

\subsection{Programs and types}
\label{Sec:2.1}
Our target language for program extraction
is an untyped lambda calculus with recursion operator and
constructors as in~\cite{IFP}, but extended by an additional 
constructor $\Amb$ that corresponds to a globally angelic version
of McCarthy's $\amb$.
This could be easily generalized to 
$\Amb$ operators 
of any arity $\ge 2$.
\label{sub-prog}
\begin{align*}
&  \mathit{Programs} \owns M,N,L,P, Q,R :: = 
  a, b, \ldots, f, g \ \ \text{(program variables)}\\
&\quad  |\  \lambda a.\,M 
\ | \ M\,N 
\ | \ \strictapp{M}{N} 
\ | \ \rec\,M \ | \  \botexp \\
&\quad |\    \Nil\ | \ \Left(M)\ | \ \Right(M)\ | \ \Pair(M,N) |\ \Amb(M,N)\\
&\quad |\  \caseof{M}\{\Left(a) \to L; \Right(b) \to R\} \\ 
&\quad |\  \caseof{M}\{\Pair(a,b) \to N\} \\
&\quad |\  \caseof{M}\{\Amb(a,b) \to N\}  
\end{align*}
Denotationally, $\Amb$ is just another pairing operator. 
Its interpretation as globally angelic choice 
will come into effect only through its operational semantics 
(see Section~\ref{sec-ops}).
Though essentially a call-by-name language, 
it also has strict application $\strictapp{M}{N}$. 

We use $a, \ldots, g$ for program variables to distinguish them from the variables
$x, y, z$ of 
the logical system 
CFP (Section~\ref{sec-cfp}). 
$\Nil, \Left, \Right, \Pair, \Amb$ are called \emph{constructors}.
Constructors different from $\Amb$ are called \emph{data constructors} and
$\mathrm{\datacons}$ denotes the set of data constructors.

$\strictapp{\Left}{M}$ stands for $\strictapp{(\lambda a.\Left(a))}{M}$, etc.,
and we sometimes write $\Left$ and $\Right$ for $\Left(\Nil)$ and 
$\Right(\Nil)$.  
Natural numbers are
 encoded as $0 \eqdef \Left$,  $1 \eqdef \Right(\Left)$, 
$2 \eqdef \Right(\Right(\Left))$, etc.

Although programs are untyped, programs extracted from
proofs will be typable by the following system of 
simple recursive types:
\[
Types \ni \rho, \sigma ::=  \alpha\ (\hbox{type variables})
                         \mid \one  
                         \mid \rho \times \sigma
                         \mid \rho + \sigma
                         \mid \ftyp{\rho}{\sigma}
                         \mid \tfix{\alpha}{\rho}\ \hbox{(*)}
                         \mid \Am(\rho) 
\]
where (*) is the restriction that $\tfix{\alpha}{\rho}$ must be 
\emph{contractive}, i.e.\ $\rho$ must not be of the form 
$\tfix{\beta_1}{\ldots\tfix{\beta_n}{\alpha}}$ ($n\ge 0$)~\cite{Pierce:2002}. 
Types $\rho, \sigma$ are considered equal, written $\rho\typeq\sigma$, 
if they yield the same (potentially infinite) type in the limit when 
fixed point types are repeatedly unfolded.\footnote{Unfolding means to replace $\tfix{\alpha}{\tau}$ by $\tau[\tfix{\alpha}{\tau}/\alpha]$.} 
Due to contractiveness, this limit is well-defined. 
It is well-known that
$\typeq$ is a decidable congruence on types~\cite{Pierce:2002}.
The intention of $\Am(\rho)$ is to be the type of
programs which, if they terminate (see Section~\ref{sec-ops}), 
reduce to a form
$\Amb(M, N)$ with $M,N\!:\!\rho$.

To obtain a well-behaved denotational semantics of types 
(see Section~\ref{sub-denot}) 
we impose a further restriction on fixed points, and also a restriction
on types of the form $\Am(\rho)$:
A type is \emph{regular}\footnote{Not to be confused with the notion of a regular tree (tree with finitely many different subtrees).}
if 
(1) for every subtype $\tfix{\alpha}{\rho}$, 
every free occurrence of $\alpha$ in $\rho$ is at
a strictly positive (s.p.) position in $\rho$, i.e.\ not in
the left part of a function type in $\rho$.
(2) for every subtype $\Am(\rho)$, $\rho$ is {\em determined}, that is,  %
$\rho$ is of the form
$\tfix{\alpha_1}{\ldots\tfix{\alpha_n}{\rho'}}$ ($n\ge 0$) where $\rho'$ 
is of the form $\one$, 
or $\sigma\times\tau$, 
or $\sigma + \tau$, 
or $\ftyp{\sigma}{\tau}$.
%is neither a type variable nor of the form $\Am(\rho'')$.
%

Restriction (1) (together with contractiveness) ensures
that the semantics of the type transformer $\alpha \mapsto \rho$
has a unique fixed point, which is taken as the semantics of $\tfix{\alpha}{\rho}$
(see Section~\ref{sub-denot}).
Restriction (2)
enables the interpretation of $\Amb$ as a bottom-avoiding choice operator
(see the explanation below Theorem~\ref{cor:ddatabot}).

An example of a regular and determined type is
\[\nat \eqdef \tfix{\alpha}{\one+\alpha},\]
the type of lazy partial natural numbers.
Besides the `total' natural numbers $0,1,2,\ldots$
described as programs above,
this type contains 
(in the sense of the typing relation introduced below)
the `partial natural numbers'
obtained by replacing in the total natural numbers
the subterm $\Left$ by $\bot$ or $\Left(\bot)$.
In particular, $\nat$ contains the terms
$\bot$, $\Right(\bot)$, $\Right(\Right(\bot))$, \ldots,
which, in the domain-theoretic denotational semantics of types
introduced in Section~\ref{sub-denot},
denote an increasing chain
whose supremum lies in (the denotation of) $\nat$, too,
and which can be seen as an infinite number denoted by
the infinite term $\Right(\Right(\ldots))$.

\begin{table}
\fbox{\begin{minipage}{\textwidth}
\begin{center}
\AxiomC{$\Gamma,\rho$ regular}
\UnaryInfC{$\Gamma,a:\rho \vdash a:\rho$}
\DisplayProof
\hspace{3em} 
\AxiomC{$\Gamma$ regular}
\UnaryInfC{$\Gamma \vdash \Nil:\one$}
\DisplayProof
\hspace{3em}
\AxiomC{$\Gamma,\rho$ regular}
\UnaryInfC{$\Gamma \vdash \bot:\rho$}
\DisplayProof
\\[0.5em]
\AxiomC{$\Gamma\vdash M:\rho$}
\AxiomC{$\sigma$ regular}
             \BinaryInfC{$\Gamma \vdash \Left(M) : \rho + \sigma$}
            \DisplayProof 
\hspace{3em} 
\AxiomC{$\rho$ regular}
\AxiomC{$\Gamma\vdash M:\sigma$}
             \BinaryInfC{$\Gamma \vdash \Right(M) : \rho + \sigma$}
            \DisplayProof \ \ \ \ 
\\[0.5em]
\AxiomC{$\Gamma\vdash M:\rho$}
\AxiomC{$\Gamma\vdash N:\sigma$}
             \BinaryInfC{$\Gamma \vdash \Pair(M,N) : \rho\times\sigma$}
            \DisplayProof 
\ \ \ \ \ \ \ \ 
\AxiomC{$\Gamma\vdash M:\rho$}
\AxiomC{$\Gamma\vdash N:\rho$}
\AxiomC{$\rho$ determined}
             \TrinaryInfC{$\Gamma \vdash \Amb(M,N) : \Am(\rho)$}
            \DisplayProof 
\\[0.5em]
\AxiomC{$\Gamma, a:\rho\vdash M:\sigma$}
             \UnaryInfC{$\Gamma \vdash \lambda a.\,M : \ftyp{\rho}{\sigma}$}
            \DisplayProof 
\ \ \ \ \ \ \ \ \ \ 
\AxiomC{$\Gamma, a:\rho\vdash M\,a:\rho$}
             \UnaryInfC{$\Gamma \vdash \rec\,M : \rho$}
            \DisplayProof 
{($a$ not free in $M$)}
\\[0.5em]
\AxiomC{$\Gamma\vdash M:\ftyp{\rho}{\sigma}$}
\AxiomC{$\Gamma\vdash N:\rho$}
             \BinaryInfC{$\Gamma \vdash M\,N : \sigma$}
            \DisplayProof \ \ \ \ 
\hspace{3em} 
\AxiomC{$\Gamma\vdash M:\ftyp{\rho}{\sigma}$}
\AxiomC{$\Gamma\vdash N:\rho$}
             \BinaryInfC{$\Gamma \vdash \strictapp{M}{N} : \sigma$}
            \DisplayProof \ \ \ \ 
\\[0.5em]
\AxiomC{$\Gamma\vdash M:\rho$}
\AxiomC{$\rho\typeq\sigma$}
\AxiomC{$\sigma$ regular}
             \TrinaryInfC{$\Gamma \vdash M : \sigma$}
            \DisplayProof \ \ \ \ 
\\[1em]
\AxiomC{$\Gamma \vdash M:\rho + \sigma$\ \ \ \ 
$\Gamma, a:\rho \vdash L:\tau$\ \ \ \ 
$\Gamma, b:\sigma \vdash R:\tau$}
\UnaryInfC{$\Gamma\vdash \caseof{M} \{\Left(a) \to L; \Right(b) \to R \} :\tau$}
            \DisplayProof \ \ \ \ 
\\[1em]
\AxiomC{$\Gamma \vdash M:\rho\times\sigma$\ \ \ \ 
$\Gamma, a:\rho, b:\sigma \vdash N:\tau$}
\UnaryInfC{$\Gamma\vdash \caseof{M} \{\Pair(a, b) \to N\} :\tau$}
            \DisplayProof 
\ \ 
\AxiomC{$\Gamma \vdash M:\Am(\rho)$\ \ \ 
$\Gamma, a,b:\rho \vdash N:\tau$}
\UnaryInfC{$\Gamma\vdash \caseof{M} \{\Amb(a, b) \to N\} :\tau$}
            \DisplayProof 
\end{center}
\end{minipage}
}
\caption{Typing rules}
\label{fig-typing}
\end{table}
The typing rules, listed in Table~\ref{fig-typing},
derive sequents of the form $\Gamma\vdash M:\rho$ where
$\Gamma = a_1:\rho_1,\ldots,a_n:\rho_n$ (considered as a set) is
a \emph{typing context} assigning to each variable $a_i$ a unique 
type $\rho_i$. 
$\Gamma$ is called regular if all $\rho_i$ are regular.
The rules (which extend those given in~\cite{IFP}) 
are valid w.r.t.~the denotational semantics.
The fact that in the rule for $\typeq$ the program $M$ appears unchanged 
in the premise and the conclusion (without the application of an isomorphism) 
means that we are working with 
equi-recursive as opposed to iso-recursive types~\cite{Pierce:2002}.
The third premise of that rule (requiring $\sigma$ to be regular)
could be omitted since, 
by Lemma~\ref{lem-typing-subst}~(\ref{lem-typing-subst-reg}), $\rho$
is regular and the regularity of $\rho$ implies that of $\sigma$.
However, the latter implication has a rather involved proof which
we decided not to include.
\begin{lem}
\label{lem-typing-subst}
\begin{enumerate}
\item\label{lem-typing-subst-tsubst}
If $\rho$ and $\sigma$ are regular, then so is $\rho[\sigma/\alpha]$; conversely, if $\rho[\sigma/\alpha]$ is regular and $\alpha$ occurs free in $\rho$, then $\sigma$ is regular.

%If $\alpha$ occurs free in $\rho$, then $\rho[\sigma/\alpha]$ is regular iff $\rho$ and $\sigma$ are both regular.
%
%\item\label{lem-typing-typeq}
%If $\rho$ is regular and $\rho \typeq \sigma$, then $\sigma$ is regular.
%
\item\label{lem-typing-subst-reg}
If $\Gamma\vdash M:\rho$, 
then $\Gamma$ and $\rho$ are regular.
\item\label{lem-typing-subst-poly}
(Polymorphism) If 
$\Gamma \vdash M:\rho$, 
then 
$\Gamma[\sigma/\alpha]\vdash M:\rho[\sigma/\alpha]$ 
for every regular type $\sigma$.
\item\label{lem-typing-subst-cut}
(Cut)
If $\Gamma,a:\rho\vdash M:\sigma$ and $\Gamma\vdash N:\rho$, then 
$\Gamma\vdash M[N/a] :\sigma$.
\end{enumerate}
\end{lem}
\begin{proof}
(\ref{lem-typing-subst-tsubst}) %and (\ref{lem-typing-typeq}) 
is proved by induction on $\rho$.
The other items are proved by induction on typing derivations.
The easy proofs are omitted.
\end{proof}

As an example of a program consider
\begin{equation}
\label{eq-f}
f\eqdef\lambda a.\,\caseof{a} \{\Left(\_)\to\Left;
                                    \Right(\_)\to \bot\}
\end{equation}
which implements the function $f$ discussed earlier, i.e.,
$f\,0 = 0$ and $f\,1 = \bot$.
$f$ has type $\ftyp{\nat}{\nat}$. 
Since 
$\Amb(0,1)$ has type $\Am(\nat)$,
the application $f\,\Amb(0,1)$ 
is not well-typed. 
Instead, we consider
$\mapamb\,f\,\Amb(0, 1)$
where 
\begin{eqnarray*}
\mapamb\eqdef \ \lambda f.\, \lambda c.\ \caseof{c}\,
\{\Amb(a,b) \to 
 \Amb(\strictapp{f}{a}, \strictapp{f}{b})\}.
\end{eqnarray*}
Clearly, $\vdash \mapamb : \ftyp{(\ftyp{\rho}{\sigma})}{\ftyp{\Am(\rho)}{\Am(\sigma)}}$ 
for all types $\rho$ and $\sigma$ that are regular and determined. 
This operator realizes the globally angelic semantics:
Using the operational semantics in Section \ref{sec-ops},
$\mapamb\ f\ \Amb(0, 1)$ is reduced to $\Amb(\strictapp{f}{0}, \strictapp{f}{1})$,
 and then $\strictapp{f}{0}$ and $\strictapp{f}{1}$ 
are computed concurrently and 
the whole expression
is reduced to $0$.
In Section~\ref{sec-cfp}, we will introduce a concurrent (or nondeterministic) 
version of Modus Ponens, (Conc-mp), 
which will automatically generate an application of $\mapamb$.

\subsection{Denotational semantics}
\label{sub-denot}
The denotational semantics has two phases: 
\emph{Phase~I} interprets programs in a Scott domain $D$
defined by the 
following
recursive domain equation
\[
D = (\Nil + \Left(D) + \Right(D) + \Pair(D\times D) 
    + \Amb(D\times D) + \Fun(D\to D))_\bot \,.
\]  
where $+$ and $\times$ denote separated sum and cartesian product,
$D \to D$ is the continuous function space, 
and the operation $\cdot_\bot$ adds a least element $\bot$.
We recommend~\cite{GierzHofmannKeimelLawsonMisloveScott03}
as a reference for domain theory and the solution of domain equations.
In the following, we will use the fact that $D$ carries the structure
of a partial combinatory algebra
(and hence can interpret the $\lambda$-calculus), and that
continuous functions on $D$ have least fixed points.
We also use that  $D$ is a bounded-complete algebraic
directed-complete partial order with least element $\bot$ and
the set $D_0$ of compact elements of $D$ is countable.
Furthermore, there exists a \emph{rank function}, $\rk$, 
which assigns to each element of $D_0$ a natural number capturing
its finite nature. 
One way of defining $\rk(a)$ is as the first stage in the
iterative construction of $D$ (see e.g.~\cite{GierzHofmannKeimelLawsonMisloveScott03})
where $a$ appears. All we need of the rank is that the following two
properties hold for $a\in D_0$ \cite{IFP}:
\begin{enumerate}
\item[$\rk 1$]
  If $a$ has the form $C(a_1,\ldots,a_k)$ for a constructor $C$,
  then $a_1,\ldots,a_k$ are compact and $\rk(a) > \rk(a_i)$ $(i \leq k)$.
\item[$\rk 2$] If $a$ has the form $\Fun(f)$, then 
for every $b\in D$, $f(b)$ is compact with $\rk(a) > \rk(f(b))$ and 
there exists a compact $b_0 \dle b$ such that $\rk(a) > \rk(b_0)$ and 
$f(b_0)=f(b)$.
Moreover, there are finitely many %
$b_1,\ldots b_n\in D_0$ with $\rk(b_i) < \rk(a)$ such that 
for all $b \in D$, $f(b) = \bigsqcup \{f(b_i) | i=1,\ldots n, b_i\dle b\}$.
\end{enumerate}
Many facts about compact elements will be proven by induction on their ranks.

Table~\ref{fig-semantics-programs} defines the denotational semantics
of programs 
which interprets 
a program $M$ as an element $\valu{M}{\eta}\in D$
for every environment $\eta$ assigning domain elements to the free variables of $M$. 
\begin{table}
\fbox{\small
\begin{minipage}{\textwidth}
\vspace*{-0.2cm}
\begin{eqnarray*}
\valu{a}{\eta} &=& \eta(a)\\
\valu{\lambda a.\,M}{\eta} &=& \Fun(f)\quad  
  \hbox{where $f(d) = \valu{M}{\eta[a\mapsto d]}$}\\
\valu{M\,N}{\eta} &=& f(\valu{N}{\eta})\quad 
                        \hbox{if $\valu{M}{\eta}= \Fun(f)$}\\
\valu{\strictapp{M}{N}}{\eta} &=& f(\valu{N}{\eta})\quad 
                        \hbox{if $\valu{M}{\eta}= \Fun(f)$ and $\valu{N}{\eta}\neq\bot$}\\
\valu{\rec\,M}{\eta} &=& \hbox{the least fixed point of $f$ if $\valu{M}{\eta}=\Fun(f)$}\\
\valu{C(M_1,\ldots,M_k)}{\eta} &=& C(\valu{M_1}{\eta},\ldots,
                                       \valu{M_k}{\eta})
\quad \hbox{($C$ a constructor (including $\Amb$))}\\
\valu{\caseof{M}{\vec{Cl}}}{\eta} &=& 
       \valu{K}{\eta[\vec a \mapsto \vec d]} 
       \quad \hbox{if $\valu{M}{\eta} = C(\vec d)$ 
and $C(\vec a) \to K\in\vec{Cl}$}\\
\valu{M}{\eta} &=& \bot\ \ 
\hbox{in all other cases, in particular $\valu{\botexp}{\eta} = \bot$}
\end{eqnarray*}

\hspace{1cm}${\eta}$ is an environment that assigns elements of $D$ to variables.
\end{minipage}
}
\caption{Denotational semantics of programs (Phase I)}
\label{fig-semantics-programs}
\end{table}
For $\valu{\lambda a.\,M}{\eta}$ to be well-defined, the function
$f$ mapping $d$ to $\valu{M}{\eta[a\mapsto d]}$ must be continuous.
This follows easily from the fact that Scott domains form a cartesian 
closed category~\cite{GierzHofmannKeimelLawsonMisloveScott03}.
Note that $\Amb$ is interpreted (like $\Pair$) as a simple pairing operator.

Types are interpreted as subdomains, i.e.\ as subsets of $D$
that are downward closed and closed under suprema of bounded subsets.
To interpret the type constructors we use the following operations
on subsets of $D$: 
\begin{eqnarray*}%
X_\bot &\eqdef& X \cup\{\bot\}\\
X+Y &\eqdef& \{\Left(a)\mid a\in X\} \cup 
                           \{\Right(b) \mid b \in Y\}\\ 
X\times Y &\eqdef& \{\Pair(a,b) \mid a\in X, b\in Y \}\\
\ftyp{X}{Y} &\eqdef& \{\Fun(f) \mid f:D\to D\hbox{ continuous, }
                           \forall a\in X (f(a) \in Y) \}\\
\Amb(X,Y) &\eqdef& \{\Amb(a,b) \mid a\in X, b\in Y \}
\end{eqnarray*}
Clearly, if $X$ and $Y$ are subdomains of $D$, so are 
$(X+Y)_{\bot}$, $(X\times Y)_{\bot}$, $\Amb(X,Y)_{\bot}$, $(\ftyp{X}{Y})_{\bot}$.
The denotational semantics of types, shown in Table~\ref{fig-semantics-types},
maps every type $\rho$ to a subdomain $\tval{\rho}{\zeta}$ for every environment 
$\zeta$ that assigns subdomains to the free type variables of $\rho$.
\begin{table}
  \fbox{\small
  \begin{minipage}{\textwidth}
\vspace*{-0.2cm}
\begin{eqnarray*}
\tval{\alpha}{\zeta} &=& \zeta(\alpha)\qquad\\
\tval{\one}{\zeta} \ &=&  \ \{\Nil\}_\bot\\
\tval{\tfix{\alpha}{\rho}}{\zeta} &=& 
\bigcap\{\subdom{X} \mid \tval{\rho}{\zeta[\alpha \mapsto X]} \subseteq X \} 
\text{\qquad ($\subdom{X}$ means $X$ is a subdomain of $D$)}\\
\tval{\Am(\rho)}{\zeta} &=&
  \Amb(\tval{\rho}{\zeta},\tval{\rho}{\zeta})_\bot \\
\tval{\rho\diamond\sigma}{\zeta} &=& 
(\tval{\rho}{\zeta}\diamond\tval{\sigma}{\zeta})_{\bot}\quad 
(\diamond\in\{+,\times,\ftyp{}{}\})   \end{eqnarray*}                           

\hspace*{1cm}
${\zeta}$ is a type environment that assigns subdomains of $D$ to type variables.
\end{minipage}
}
\caption{Denotational semantics of types}
\label{fig-semantics-types}
\end{table}
It is easy to see that the typing rules in Table~\ref{fig-typing} are sound w.r.t.\ 
this semantics, 
that is, 
whenever $\Gamma\vdash M :\rho$ is derivable by the typing rules, 
and $\eta\in\tval{\Gamma}{\zeta}$, i.e.\ 
$\eta(a)\in\tval{\sigma}{\zeta}$ for all type assignments $a:\sigma\in\Gamma$,
then $\valu{M}{\eta} \in \tval{\rho}{\zeta}$.
The rules for fixed point types are sound since 
$\tval{\tfix{\alpha}{\rho}}{\zeta} = \tval{\rho[\tfix{\alpha}{\rho}/\alpha]}{\zeta}$
(see~\cite{IFP}, Lemma~7 for a proof).

For closed terms and types the environments are redundant and we write 
just $\val{M}$ and $\tval{\rho}{}$.
We then also write ``$M:\rho$'' or ``$M$ has type $\rho$''
for ``$\valu{M}\in\tval{\rho}{}$'',
and we say ``$M$ is regular typed'' or ``$M$ has a regular type''  
if $M$ is a closed program such that $M:\rho$ for some closed regular type $\rho$. 
In contrast, we write $\vdash M:\rho$ to express that the type 
assignment is provable by the rules in Table~\ref{fig-typing}.\footnote{Even if $\rho$ is regular, for computability-theoretic reasons,
$M:\rho$ does, in general, not imply that $\vdash M:\rho$ is derivable.}

In \emph{Phase~II} of the denotational semantics we assign to every $a\in D$
a set $\ddata(a)\subseteq D$ 
that reveals the role of $\Amb$ as a choice operator. 
The relation $\ddata(a,d)$, also written 
`$d\in\ddata(a)$',  
is defined (coinductively) as the largest relation satisfying
\begin{align}
d \in \ddata(a) \quad \eqnu \quad 
&\exists a',b'\,(a = \Amb(a', b') \land a' \ne \bot \land d \in \ddata(a')) \ \lor \notag\\
&\exists a',b'\,(a = \Amb(a', b') \land b' \ne \bot \land  d \in \ddata(b')) \ \lor \notag\\
&(a = \Amb(\bot, \bot) \land d = \bot) \ \lor \notag\\  
&\bigvee_{C \in \mathrm{\datacons}} 
\exists \vec{a'},\vec{d'}
\left(a =  C(\vec{a'}) \land
d = C(\vec{d'}) \land \bigwedge_i d'_i \in \ddata(a'_i)\right) \ \lor\notag\\
&\exists f\,(a = \Fun(f)  \land d = a )\  \lor\notag\\ 
&(a = d = \bot)\,. \label{e-data} 
\end{align} 
Now, every closed program $M$ denotes the set 
$\ddata(\val{M}) \subseteq D$ containing all possible globally 
angelic choices derived from its denotation in $D$.
For example, $\ddata(\val{\Amb(0,1)}) = \{0,1\}$ and, for $f$ as
defined in~(\ref{eq-f}), we have, as expected,
$\ddata(\val{\mapamb\ f\ \Amb(0, 1)}) = \ddata(\val{\Amb(0,\bot)}) = \{0\}$. 
In Section~\ref{sec-ops} we will define an operational semantics whose fair
execution sequences starting with a regular-typed program 
$M$ compute exactly the elements of $\ddata(\val{M})$.
In Section~\ref{sec-pe} (Lemma~\ref{lem-data-nonempty}) 
we will show that $\ddata(a)$ is nonempty 
for all $a\in D$.

\begin{rem}
In~\cite{IFP} the meaning of the word `data' is different.
There, it means that a domain element has no function component.
Therefore, in $\IFP$, `data' means `no $\Fun$', 
whereas here, it means `no $\Amb$'.
\end{rem}

\begin{example}\label{ex:random}
Let $M = \rec\ \lambda a. \Amb(\Left(\Nil), \Right(a))$. 
$M$ is a closed program
of regular type $\tfix \alpha \Am(\one + \alpha)$.
We have $\ddata(\val{M}) = \{0, 1, 2, \ldots,\infty\}$
where $\infty = \Right(\Right(\ldots))$.
Thus, we can express countable choice (aka random assignment \cite{LassenMoran99})
with $\Amb$.
\end{example}

\begin{lem}\label{lem:ddatabot}
If $a\in\tval{\rho}{}$  %$a\in\tval{\rho}{\zeta}$ 
where $\rho$ is a closed regular 
type, then the following are equivalent:
  (i) $\bot \in \ddata(a)$;
  (ii) $\{\bot\} = \ddata(a)$;
  (iii) $a \in\{\bot,\Amb(\bot, \bot)\}$.   
\end{lem}
\begin{proof}
Obviously (iii) implies (ii), and (ii) implies (i).
Therefore, it suffices to show that (i) implies (iii).
By the definition of $\ddata$, 
$\bot \in \ddata(a)$ if and only if $a = \bot$ or ($a = \Amb(a', b')$
  and (($a' \ne \bot$ and $\bot \in \ddata(a')$) or 
      ($b' \ne \bot$ and $\bot \in \ddata(b')$) or $a' = b' = \bot$)).
If $a=\bot$, then (iii) holds.
If $a=\Amb(a',b')$, then, since $\rho$ is 
%not a variable (but regular),
closed and regular,
$\rho$ must be of the form $\tfix{\alpha_1}{\ldots\tfix{\alpha_n}{\Am(\sigma)}}$
where $\sigma$ is determined.  
Therefore, %$a',b'\in\tval{\sigma}{}$
$a',b'\in\tval{\sigma}{\zeta'}$ (for some type environment $\zeta'$) 
and hence neither $a'$ nor $b'$ has the form $\Amb(a'', b'')$.  
It follows that
  $\bot \in \ddata(a')$ if and only if $a' = \bot$, and
  $\bot \in \ddata(b')$ if and only if $b' = \bot$.
  Therefore, (iii) holds.
\end{proof}

% \begin{lem}\label{lem:ddatabot}
% If $a\in\tval{\rho}{\zeta}$ 
% where $\rho$ is a regular
% type but not a variable, then the following are equivalent:
%   (i) $\bot \in \ddata(a)$;
%   (ii) $\{\bot\} = \ddata(a)$;
%   (iii) $a \in\{\bot,\Amb(\bot, \bot)\}$.   
% \end{lem}
% \begin{proof}
% %
% Obviously (iii) implies (ii), and (ii) implies (i).
% Therefore, it suffices to show that (i) implies (iii).
% %
% By the definition of $\ddata$, 
% $\bot \in \ddata(a)$ if and only if $a = \bot$ or ($a = \Amb(a', b')$
%   and (($a' \ne \bot$ and $\bot \in \ddata(a')$) or 
%       ($b' \ne \bot$ and $\bot \in \ddata(b')$) or $a' = b' = \bot$)).
% %
% If $a=\bot$, then (iii) holds.
% %
% If $a=\Amb(a',b')$, then, since $\rho$ is not a variable (but regular),
% $\rho$ must be of the form $\tfix{\alpha_1}{\ldots\tfix{\alpha_n}{\Am(\sigma)}}$
% where $\sigma$ is determined.  
% %
% Therefore, $a',b'\in\tval{\sigma}{\zeta'}$ (for some type environment $\zeta'$) and hence 
% neither $a'$ nor $b'$ have the form $\Amb(a'', b'')$.  
% It follows that,
%   $\bot \in \ddata(a')$ if and only if $a' = \bot$, and
%   $\bot \in \ddata(b')$ if and only if $b' = \bot$.
%   Therefore, (iii) holds.
% %
% \end{proof}

\section{Operational semantics}
\label{sec-ops}
We define a small-step operational semantics 
that, in the limit, reduces each closed program $M$ nondeterministically
to an element in $\ddata(\val{M})$ (Theorem~\ref{thm:data}). 
If $M$ has a regular type, the converse holds as well: 
For every $d\in\ddata(\val{M})$ there exists
a reduction sequence for $M$ computing $d$ in the limit 
(Theorem~\ref{thm:dataconv}).
If $d$ is compact,
%If $M$ denotes a compact data, 
then the limit is obtained 
after finitely many reductions (Lemma~\ref{lem-claim}).
In this section, all programs and types are assumed to be closed.

\subsection{Reduction to weak head normal form}
A program is called a \emph{weak head normal form (w.h.n.f.)} if it
 begins with a constructor 
(including $\Amb$) or has the form $\lambda a. M$, and
called a \emph{deterministic weak head normal form (deterministic w.h.n.f.)} if it is 
a w.h.n.f.\ that does not begin with $\Amb$.
Note that a w.h.n.f.\ (deterministic or not) always has a non-bottom value.

In the first part of Table \ref{table:reduction},
we inductively define a 
small-step leftmost-outermost reduction 
relation $\ssp$ on 
programs.
We say a program $M$ \emph{has a w.h.n.f.} if there exists 
a w.h.n.f.\ $N$  such that $M\ssp^*N$.
\begin{table}
\fbox{
\begin{minipage}{\textwidth}
\noindent Rules of $\ssp$:  
\medbreak
\begin{quote}
\begin{enumerate}\setlength{\itemsep}{0.2cm}
\setlength{\itemindent}{0.4cm}
  \item[(s-i)] $(\lambda a.\,M)\ N \ssp M[N/a]$
\item[(s-ii)] \AxiomC{$M \ssp M'$}
            \UnaryInfC{$M\,N \ssp M'\,N$}
            \DisplayProof 
\item[(s-iii)]  $\strictapp{(\lambda a.\,M)}{N} \ssp M[N/a]$ \quad 
if $N$ is a w.h.n.f.

\item[(s-iv)] \AxiomC{$M \ssp M'$}
            \UnaryInfC{$\strictapp{M}{N} \ssp \strictapp{M'}{N}$}   
            \DisplayProof  \quad  
if $N$ is a w.h.n.f.
\item[(s-v)]  \AxiomC{$N \ssp N'$}
            \UnaryInfC{$\strictapp{M}{N} \ssp \strictapp{M}{N'}$}
            \DisplayProof 
\item[(s-vi)] $\rec\,M \ssp M\,(\rec\,M)$
\item[(s-vii)] $\caseof{C(\vec M)} \{\ldots;C(\vec b)\to N;\ldots\}\ssp  N[\vec M/\vec b]$\\
 \hspace*{4cm}($C$ ranges over constructors including $\Amb$)

\item[(s-viii)] \AxiomC{$M \ssp M'$}
             \UnaryInfC{$\caseof{M}\{\vec{Cl}\}\ssp 
               \caseof{M'}\{\vec{Cl}\}$}
            \DisplayProof 
\item[(s-ix)] $M \ssp \bot$ \quad if $M$ is $\bot$-like (see below)
\end{enumerate}
\end{quote}

\bigskip
\noindent Rules of $\newprintc$:
\begin{quote}
\begin{enumerate}
  \setlength{\itemsep}{0.2cm}
  \setlength{\itemindent}{0.4cm}
  \item[(c-i)] \AxiomC{$M \ssp M' $} 
  \UnaryInfC{$M \newprintc M'$}
  \DisplayProof 
  
  \item[(c-ii)] 
  \AxiomC{$M_1 \ssp M_1'$}  
  \UnaryInfC{$\Amb(M_1,M_2) \newprintc \Amb(M_1',M_2)$}
  \DisplayProof 
  \item[(c-ii')] 
  \AxiomC{$M_2 \ssp M_2'$}  
  \UnaryInfC{$\Amb(M_1,M_2) \newprintc \Amb(M_1,M_2')$}
  \DisplayProof 
  \item[(c-iii)] 
    $\Amb(M_1, M_2) \newprintc M_1$  \ 
  if $M_1$ is a w.h.n.f.
    
  \item[(c-iii')] 
    $\Amb(M_1, M_2) \newprintc M_2$  \ 
  if $M_2$ is a w.h.n.f.
  \end{enumerate}
\end{quote}

\bigskip
\noindent Rules of $\newprintp$:
  \begin{quote}
  \begin{enumerate}    \setlength{\itemsep}{0.2cm}
    \setlength{\itemindent}{0.4cm}
    \item[(p-i)]
    \AxiomC{$M \newprintc M' $} 
    \UnaryInfC{$M \newprintp M'$}
    \DisplayProof 
    \item[(p-ii)]
    \AxiomC{$M_i \newprintp M_i'$ $(i = 1,\ldots, k)$} 
    \UnaryInfC{$C(M_1,\ldots,M_k) \newprintp C(M_1',\ldots,M_k')$}
    \DisplayProof 
        ($C\in\datacons$)      \\    %
    \item[(p-iii)]
          $\lambda a.\,M \newprintp \lambda a.\,M$ 
  \end{enumerate}
\end{quote}      

\bigskip
$\bot$-like programs are such that their syntactic forms immediately imply that
they denote $\bot$, more precisely they are of the form  $\bot$,
$C(\vec M)\,N$, $\strictapp{C(\vec M)}{N}$, 
and $\caseof{M}\,\{\ldots\}$ where $M$ is a lambda-abstraction or
of the form $C(\vec M)$ such that there is no clause in $\{\ldots\}$ 
which is of the form $C(\vec a) \to N$. 
W.h.n.f.s 
are never $\bot$-like, and the only typable $\bot$-like
program is $\bot$.

\end{minipage}
}
\caption{
Operational semantics of programs
\label{table:reduction}}
\end{table}

The following lemma establishes fundamental properties of the reduction relation $\ssp$.
Part~(5) provides the crucial link to the denotational semantics and lies at the core
of the Adequacy Theorems.

\begin{lem}\label{lem:ssp}
\begin{enumerate}
\item\label{lem:ssp:sred}
$\ssp$ preserves types, syntactically and semantically (subject reduction):\\ 
If $M \ssp M'$ and $(\vdash)\, M:\rho$, then $(\vdash)\, M':\rho$.
\item\label{lem:ssp:pres}
$\ssp$ preserves the denotational semantics: 
If $M \ssp M'$, then $\val{M} = \val{M'}$. 
\item\label{lem:ssp:det}
$\ssp$ is deterministic: 
For every program $M$, $M \ssp M'$ for at most one $M'$.
\item\label{lem:ssp:norm}
$M$ is a $\ssp$-normal form iff $M$ is a  
w.h.n.f.  
\item\label{lem:ssp:ade}
Adequacy Lemma\label{lem:ade}:
  If $\val{M} \ne \bot$, then 
$M$ has a w.h.n.f.\ (which, by~(\ref{lem:ssp:det}) and (\ref{lem:ssp:pres})
is uniquely determined and has the same denotational semantics as $M$).
%there is a w.h.n.f.~$V$ s.t.\ $M \ssp^* V$. 
%
\end{enumerate}
\begin{proof}
Item (\ref{lem:ssp:sred}) is standard for syntactic typing.
For semantic typing it follows from~(\ref{lem:ssp:pres}). 

(\ref{lem:ssp:pres}), (\ref{lem:ssp:det}), (\ref{lem:ssp:norm}) are easy.

We sketch the proof of~(\ref{lem:ssp:ade}), which is similar
to the proof of Theorem~11 in \cite{Berger10}.

To every $a\in D_0$ one assigns a set $\cl{a}$ of closed programs,  
by recursion on $\rk(a)$: 
\begin{eqnarray*}
\cl{\bot} &=& \hbox{the set of all closed programs}\\
\cl{C(a_1,\ldots,a_k)} &=& 
 \{M \mid \exists M_1,\ldots,M_k,\,M\ssp^* C(M_1,\ldots,M_k) \land\\
 &&\quad\quad\quad\quad\bigwedge_{i \leq k}M_i\in\cl{a_i} \}\\
\cl{\Fun(f)} &=& \{M\mid \exists x, M',\,(M\ssp^*\lambda x.\,M' \land\\
&&\quad\quad\quad\quad\forall b\in D_0\,(\rk(b) <\rk(\Fun(f))\to\\
&& \quad\quad\quad\quad\quad\forall N\in\cl{b}\,(M'[N/x]\in\cl{f(b)})))\} 
\end{eqnarray*}

One can show:
\begin{enumerate}[(a)]
\item\label{ade-lem-mon}
If $a\dle b$, then $\cl{a}\supseteq\cl{b}$, for all $a,b\in D_0$. \\
(induction on the maximum of $\rk(a)$ and $\rk(b)$)
\item\label{ade-lem-whnf}
If $a\in D_0$ is different from $\bot$,
then for every closed program $M$:

$M\in\cl{a}$ iff  $M\ssp^*M'$ for some
w.h.n.f.\ $M'$ with $M'\in\cl{a}$.\\
(easy induction on $\rk(a)$ using (\ref{lem:ssp:norm}))
\item\label{ade-lem-denot}
If $a \dle \val{M}$, then $M\in\cl{a}$,
for all closed programs $M$ and $a\in D_0$.\\
(one shows by induction on programs $M$, 
and using \ref{ade-lem-mon} and \ref{ade-lem-whnf},
the more general statement:
If $a \dle \valu{M}{\eta}$ where $a\in D_0$, then
$M\theta\in\cl{a}$, provided $\eta(x)\in D_0$ and $\theta(x)$ is closed 
with $\theta(x)\in\cl{\eta(x)}$ for all free variables $x$ of $M$)
\end{enumerate} 
Now, assume $\val{M}\neq\bot$ where $M$ is a closed program. 
Then $\bot\neq a\dle\val{M}$ for some $a\in D_0$.
By \ref{ade-lem-denot}, $M\in\cl{a}$. 
By \ref{ade-lem-whnf}, $M\ssp^* M'$ for some w.h.n.f.\ $M'$.
\end{proof}
\end{lem}

\subsection{Making choices}
Next, we define the reduction relation $\newprintc$ (`c' for `choice') 
that asynchronously reduces the arguments of $\Amb$, or chooses one
that is in %deterministic 
w.h.n.f.
\ (middle part of Table~\ref{table:reduction}). 
This reduction does not preserve the denotational semantics and is not deterministic. Therefore, 
properties corresponding to
Lemma~\ref{lem:ssp}~(\ref{lem:ssp:pres}), (\ref{lem:ssp:det})
do not hold. As for (\ref{lem:ssp:norm}), we have
\begin{lem}\label{lem:newprintc}
$M$ is a $\newprintc$-normal form iff $M$ is a  
deterministic w.h.n.f.
\end{lem}
\begin{proof}
Immediate from 
the definition and Lemma \ref{lem:ssp}(\ref{lem:ssp:norm}).
\end{proof}
Finally, we define the reduction relation $\newprintp$ 
(third part of Table~\ref{table:reduction})
which reduces the arguments of data constructors 
%$\Pair$ 
in parallel 
and ensures the reduction of every (closed) program 
(a property that can be easily proven by structural induction).
For example, $\Nil\newprintp \Nil$ follows from (p-ii).
This organization of the operational semantics secures adequacy w.r.t.\ 
the denotational semantics in Section \ref{sec:adequacy}.
For a practical implementation, one might prefer a sequential version 
of (p-ii) and remove (p-iii), as done in the definition of 
the function \verb|ed| (`extract data') in Appendix~\ref{sub-program-data}.
\begin{lem}\label{lem-newprintp}
If $M \newprintp M'$ and $(\vdash)\,M:\rho$, then $(\vdash)\,M':\rho'$ 
for some type $\rho'$ obtained from $\rho$ by deleting some occurrences 
of the type constructor $\Am$.
Hence, if $M$ has a regular type, then so does $M'$.
\end{lem}
\begin{proof}
We prove the first statement by induction on the definition of $M \newprintp M'$.
Since the proof is essentially the same for syntactical and semantic typing,
we only consider the former. 
Rule (p-i) imports the relation $\newprintc$, so we first prove 
a corresponding statement for $\newprintc$, by induction on the definition of
$M \newprintc M'$: For the rules (c-i), (c-ii), and (c-ii')
the statement holds by Lemma~\ref{lem:ssp}~(\ref{lem:ssp:sred}). 
The rules (c-iii) and (c-iii') delete  
an $\Amb$ from a program which corresponds to a deletion of an
$\Am$ from its type.
This completes the proof for the rule (p-i).
The other rules for $\newprintp$ just apply $\newprintp$ in parallel to subprograms,
therefore the induction hypothesis directly applies.

The second statement follows 
from the first, since, clearly, deleting $\Am$ preserves regularity.
\end{proof}
In Theorem~\ref{cor:ddatabot} we will prove a property corresponding 
to Lemma~\ref{lem:ssp}~(\ref{lem:ssp:ade}) for $\newprintp$. 
It requires fairness of computation in the following sense:

We call a $\newprintp$-reduction sequence 
\emph{fair} if, intuitively, an infinite sequence of (c-ii) or (c-ii')
reductions occurs only if this is unavoidable, i.e.\ all members of this
sequence are $\Amb$ programs where both arguments have no w.h.n.f.

More precisely, we coinductively define the 
set of \emph{fair reduction sequences}
as the largest set of infinite $\newprintp$-reduction sequences such that 
\begin{enumerate}
\item[(a)] if $M_0\newprintp M_1 \newprintp \ldots$ is fair
and $m$ is such that all $M_{m+i}$ are of the form
$\Amb(L_i,R_i)$ and all steps $M_{m+i}\newprintp M_{m+i+1}$ are
instances of (c-ii) or (c-ii'), then $L_0$ and $R_0$ have no w.h.n.f.
(and hence no $L_i$ and no $R_i$ has a w.h.n.f.);
%
%% if $M_0\newprintp M_1 \newprintp \ldots$ is fair 
%% and $m$ is such that all $M_{m+i}$ are reduced by (c-ii) or (c-ii')
%% then $M_m$ is of the form 
%% $\Amb(L,R)$ where neither $L$ nor $R$ has a w.h.n.f. 
%% (and hence all $M_{m+i}$ are of such a form).
%
\item[(b)] if $M_0\newprintp M_1 \newprintp \ldots$ is fair 
and $m$ is such that all $M_{m+i}$ are of the form
$C(N_{i,1},\ldots,N_{i,n})$, with a fixed data constructor $C$, then all the 
reduction sequences $N_{0,k} \newprintp N_{1,k}\newprintp \ldots$ 
($k = 1,\ldots n$) are fair.
\end{enumerate}
Since the relation $\ssp$ is deterministic 
(Lemma~\ref{lem:ssp}~(\ref{lem:ssp:det})),
fairness does not depend on initial segments, that is,
for $M_0\newprintp M_1 \newprintp \ldots$ and any $k\in\NN$,
$M_0\newprintp M_1 \newprintp \ldots$ is fair iff 
$M_k\newprintp M_{k+1} \newprintp \ldots$ is fair.

%% We call a $\newprintp$-reduction sequence 
%% \emph{unfair} if, intuitively, from some point on, one side of an 
%% $\Amb$ term is permanently reduced but not the other. 
%% More precisely, we inductively  define a reduction sequence 
%% $M_0\newprintp M_1 \newprintp \ldots$ to be unfair if 
%% \UB{for some $m$
%% \begin{enumerate}
%% %
%% \item[-] either all $M_{m+i}$ are of the form 
%% $\Amb(L_i,R_i)$ and there are no $i,j$ such that 
%% $L_i\ssp L_{i+1}$ and $R_j \ssp R_{j+1}$, 
%% %
%% \item[-] or else all $M_{m+i}$ are of the form 
%% $C(N_{i,1},\ldots,N_{i,n})$, with a fixed data constructor $C$, and
%% (at least) one of the reduction sequences 
%% $N_{0,k} \newprintp N_{1,k}\newprintp \ldots$
%% is unfair.
%% %
%% \end{enumerate}
%% }
%% %
%% A $\newprintp$-reduction sequence is \emph{fair} if it is not unfair.

A {\em computation} of $M$ is an infinite fair sequence 
$M = M_0 \newprintp M_1 \newprintp \ldots$  
Intuitively,  
computation proceeds as follows:
A program $M$ is head reduced by $\ssp$ to a 
w.h.n.f.\ $M'$, 
and
if $M'$ is a data constructor term, all arguments are reduced in parallel by (p-ii).
If $M'$ has the form $\Amb(L, R)$, 
two concurrent threads 
are invoked for the reductions of $L$ and $R$ in parallel, 
and the one reduced to a 
w.h.n.f.\ first is used. 
Fairness corresponds to the requirement 
that the `speed' of each thread is positive
which means, in particular, that no thread can block another.
Note that the premises of (s-ii), (s-iv), (s-v) and (s-viii)
use $\ssp$, not $\newprintc$.
This means that  $\newprintc$ is applied only to the outermost redex.
Also, (c-ii) is 
defined through $\ssp$, not $\newprintc$,  and thus no thread creates new threads.
This ability to control the number of threads was not available
in an earlier version of this language~\cite{BergerCSL16}
(see also the discussion in Section~\ref{sub-related}).

\subsection{Computational adequacy:
  Matching denotational and operational semantics}\label{sec:adequacy}
We define 
${M}_{{D}} \in {D}$ by structural induction on programs $M$:
\begin{align*}
{C(M_1,\ldots,M_k)}_{D} &= C({M_1}_{D},\ldots, {M_k}_{D})
&  \hbox{($C\in\datacons$)}\\
{(\lambda a. M) }_{D} &= \val{\lambda a. M}   \\
{M}_{D} &= \bot  & \mbox{otherwise}
\end{align*}
Intuitively, $M_{D}$ is the part of $M$ that has been fully evaluated to a data.
Note that $M_D \neq\bot$ iff $M$ is a deterministic w.h.n.f.

\begin{lem}\label{lem-MD}
\begin{enumerate}[(a)]
\item\label{lem-MD-val} If $d \in\ddata(\val{M})$, then $M_{D} \sqsubseteq d$.
\item\label{lem-MD-mon} If $M \newprintp M'$, then $M_D \sqsubseteq M'_D$.
\end{enumerate}
\end{lem}
\begin{proof}
\ref{lem-MD-val} Easy induction on $M$. 
%% If $M = C(M_1,\ldots,M_n)$ for a data constructor $C$, then
%% $d= C(d_1,\ldots,d_n)$ with $d_i\in\ddata(M_i)$ and 
%% $M_D = C((M_1)_D,\ldots,(M_n)_D)$. By the i.h.\ $(M_i)_D\sqsubseteq d_i$.
%% Hence $M_D \sqsubseteq d$.
%% %
%% If $M = \lambda a.N$, then $d=M_D=\val{M}$.
%% %
%% In all other cases $M_D=\bot\sqsubseteq d$.
%
\ref{lem-MD-mon} Easy induction on the definition of $M\newprintp M'$.
\end{proof}

%Since clearly $M \newprintp N$ implies $M_{D} \sqsubseteq_{D} N_{D}$,
By Lemma~\ref{lem-MD}~\ref{lem-MD-mon},
for every computation $ M_0 \newprintp M_1 \newprintp  \ldots$,  
  the sequence $((M_i)_{D})_{i \in \NN}$ is increasing and therefore
  has a least upper bound in $D$. %

\begin{thm}[Computational Adequacy: Soundness]
\label{thm:data}
For every computation  $M =  M_0 \newprintp M_1 \newprintp  \ldots$, 
$\sqcup_{i \in \NN} (M_i)_{D} \in \ddata(\val{M})$. 
\begin{proof}
In the proof we will use the principle of coinduction for the binary predicate 
$\ddata$ 
which was defined coinductively in (\ref{e-data}). 
In general, for a coinductive predicate
$Q(\vec x) \eqnu \Phi(Q)(\vec x)$ 
and any predicate $P$ of the same arity, one can prove
$P\subseteq Q$ by coinduction
by proving 
$P\subseteq \Phi(P)$.
This proof principle is formalized in IFP/CFP
(last rule in Table~\ref{table-proof-ifp} in Section~\ref{sec-cfp}).
Here, we apply it on the meta-level. 

  Set $P(a, d) \eqdef d = \sqcup_{i \in \NN} (M_i)_D$ for some
  computation
    $M_0 \newprintp M_1 \newprintp\ldots$ with $a = \val{M_0}$. 
     We show $P\subseteq\ddata$, that is, 
$P(a, d) \to d \in \ddata(a)$, 
by coinduction.
    Therefore, we have to show $P(a, d) \to \Phi(P)(a, d)$ where
   $\Phi(P)(a,d)$ is obtained by replacing in the right-hand side of 
   (\ref{e-data}), the predicate $\ddata$ by $P$.
     Assume $P(a,d)$, witnessed by the computation, that is, 
  fair reduction sequence 
  $M = M_0 \newprintp M_1 \newprintp \ldots$ with $a = \val{M}$
  and $d = \sqcup_{i \in \NN} (M_i)_{D}$.
   We have to show that at least one of the following six conditions holds:
  \begin{enumerate}
  \item[(1)] $a = \Amb(a', b') \land a' \ne \bot \land P(a',d)$
  \item[(2)] $a = \Amb(a', b') \land b' \ne \bot \land P(b',d)$
  \item[(3)] $a = \Amb(\bot, \bot) \land d = \bot$
  \item[(4)] $ a =  C(\vec{a'}) \land d = C(\vec{d'})\land\bigwedge_i P(a'_i,d'_i)$
 for some data constructor $C$.     \item[(5)] $a = \Fun(f)  \land d = a$
  \item[(6)] $a = d = \bot$
  \end{enumerate}
  Any computation $M = M_0 \newprintp M_1 \newprintp \ldots$ belongs to one 
of the following categories:
  
  Case a: All reductions $M_i \newprintp M_{i+1}$ are (p-i) derived from (c-i):
  That is, $M_0 \ssp M_1 \ssp \ldots$ In this case,  $d = \bot$ and (6) holds by
  Lemma~\ref{lem:ssp} (\ref{lem:ssp:det}) and (\ref{lem:ssp:ade}).
 
  Case b: For some $n$, $M_i \ssp M_{i+1}$  ($i< n$) and
  $M_i \newprintc M_{i+1}$  ($n \leq i$) by (c-ii) and (c-ii'):
  In this case, $M_i = \Amb(L_i, R_i)$ for $i \geq n$.
  We have $L_i \ssp L_{i+1}$ and $R_i = R_{i+1}$ or 
  $L_i = L_{i+1}$ and $R_i \ssp R_{i+1}$.  
  By fairness, the $L_i$ and $R_i$ have no w.h.n.f.s, 
  hence (3) holds by Lemma~\ref{lem:ssp}~(\ref{lem:ssp:det}), 
  (\ref{lem:ssp:pres}) and~(\ref{lem:ssp:ade}).

  Case c:
  $M_i \ssp M_{i+1}$ for $i< n$  and $M_i \newprintp M_{i+1}$ for $i \geq n$ by (p-iii):
  (5) holds by Lemma \ref{lem:ssp}(\ref{lem:ssp:pres}).

  Case d:
  $M_i \ssp M_{i+1}$ for $i< n$ and $M_i \newprintp M_{i+1}$ for $i \geq n$ by (p-ii):
  $M_{i}$ has the form 
$C(N_{i,1},\ldots,N_{i,k})$ 
for $i \geq n$ 
  and $N_{n,j} \newprintp N_{n+1,j} \newprintp \ldots$ are fair reductions.  In addition, 
  %$a = C(\val{N_{i,1}},\ldots,\val{N_{i,k}})$  
  $a = C(\val{N_{n,1}},\ldots,\val{N_{n,k}})$ and 
$d = C(\sqcup_{i \ge n} ({N_{i,1}})_D,\ldots, \sqcup_{i \ge n} (N_{i,k})_D)$.  
Therefore, (4) holds.
  
  Case e:
  $M_i \ssp M_{i+1}$ for $i< n$, $M_n = \Amb(L,R)$, 
  $M_{i} \newprintc M_{i+1}$ for $n \leq i< m$ by (c-ii) and (c-ii'),
  $M_{m} \newprintc M_{m+1}$ by (c-iii) or (c-iii'):
  $a = \Amb(a',b')$ with $a' = \val{L}$ and 
  $b' = \val{R}$. $M_m = \Amb(L',R')$.
  If (c-iii) is used, $M_{m+1} = L'$ with 
  $a' = \val{L} = \val{L'} \neq \bot$.
  Since the reduction sequence 
  $M_{m+1} \newprintp M_{m+2} \newprintp \ldots$ is fair again,
  $P(a',d)$ and hence (1) holds.  
  Similarly, (2) holds for the case (c-iii') is used.
    
\end{proof}
    \end{thm}

The converse of Theorem~\ref{thm:data} does not hold in general,
i.e.\ $d \in \ddata(\val{M})$ does not necessarily imply
$d = \sqcup_{i \in \NN} ((M_i)_D)$  
for some computation $M=M_0\newprintp M_1\newprintp \ldots$
For example, for $M \eqdef \rec\,\lambda a.\,\Amb(a,\bot)$ 
(for which $\val{M} = \Amb(\val{M},\bot)$)
one sees that $d \in \ddata(\val{M})$ 
for every $d \in D$ (by coinduction).
However,
$M \newprintp^* N$ iff $N$ 
is one of 
$M$, 
$(\lambda a.\,\Amb(a,\bot))\,M$, 
$\Amb(M,\bot)$, 
$\Amb((\lambda a.\,\Amb(a,\bot))\,M,\bot)$ and 
$\Amb(\Amb(M,\bot),\bot)$.
%where $F\eqdef \lambda a.\,\Amb(a,\bot)$.
%However, clearly $M \newprintp^* N$ iff
%$N\in\{M,\Amb(M,\bot),\Amb(\Amb(M,\bot),\bot)\}$.
%
%\HT{
%$N$ is $\ssp$-reducible to one of 
%$\{M,\Amb(M,\bot),\Amb(\Amb(M,\bot),\bot)\}$.}
Hence, $N_D=\bot$ for all such $N$.
But $M$ has the type $\tfix{\alpha}{\Am(\alpha)}$ which is not
regular (see Section~\ref{sub-prog}). 
For programs of a regular type, 
the converse of Theorem~\ref{thm:data} does hold.

\begin{thm}[Computational Adequacy: Completeness]
\label{thm:dataconv}
If $M$ has a regular type, 
then for every $d \in \ddata(\val{M})$, there is a computation 
$M = M_0 \newprintp M_1 \newprintp \ldots$ with
$d = \sqcup_{i \in \NN} ((M_i)_{D})$.
\end{thm}

To prove Theorem~\ref{thm:dataconv}, we first note two auxiliary facts
that will be used in Lemma~\ref{lem-claim} 
which is the essence of the proof.
%The following lemmas is the essence of the proof:
\begin{lem}\label{lem-lem-claim}
If $M$ is a program of regular type and $d \in \ddata(\val{M})$,
then $d$ is not of the form $\Amb(d', d'')$.
\end{lem}
\begin{proof}
Let $a = \val{M}$. By the definition of $\ddata$ in (\ref{e-data}),
$d \in \ddata(a)$ implies that
$d = \bot$, or
$d = C(\vec{d'})$ for some data constructor $C$, or
$d = a = \Fun(f)$, or
$a = \Amb(a', b')$ and
($a' \ne \bot$ and $d \in \ddata(a')$, or $b' \ne \bot$ and $d \in \ddata(b')$).
In the first three cases the claim is obvious.
In the last case, neither $a'$ nor $b'$ has the form $\Amb(a'', b'')$,
as shown in the proof of Lemma~\ref{lem:ddatabot}.
Hence, if $a' \ne \bot$ and $d \in \ddata(a')$, then,
by the definition of $\ddata$ again, $d = C(\vec{d'})$ for some
data constructor $C$ or $d = a' = \Fun(f)$.
The same holds for $b'$.
\end{proof}

\begin{lem}\label{lem:step-step}
If $d \in \ddata(\val{M})$,
then there exists $N$ such that
$M \newprintp N$ and $d \in \ddata(\val{N})$.
\end{lem}
\begin{proof}
Induction on $M$.

If $M$ is not a w.h.n.f., then $M \ssp N$ for some $N$
(Lemma~\ref{lem:ssp}~(\ref{lem:ssp:norm})), hence $M\newprintp N$,
and $\val{N}=\val{M}$ by Lemma~\ref{lem:ssp}~(\ref{lem:ssp:pres}).

If $M$ is a $\lambda$-abstraction, then $M\newprintp M$, 
by (p-iii).

If $M = C(M_1,\ldots,M_k)$ with $C\in\datacons$, then, by the definition of
$\ddata$, $d = C(d_1,\ldots,d_k)$ with $d_i\in\ddata(\val{M_i})$.
By the induction hypothesis, there are $N_i$ with
$M_i\newprintp N_i$ and $d_i\in\ddata(\val{N_i})$.
Hence $M \newprintp C(N_1,\ldots,N_k)$ by (p-ii), and
$d\in\ddata(\val{C(N_1,\ldots,N_k)})$.

If $M=\Amb(M_1,M_2)$, then, by the definition of $\ddata$,
one of the following holds.

(1) $\val{M_1}\neq\bot$ and $d\in\ddata(\val{M_1})$.
If $M_1$ is a w.h.n.f., then $M\newprintp M_1$, 
by (c-iii), and we set $N = M_1$.
Otherwise, $M_1\ssp M_1'$ for some $M_1'$, and we set $N=\Amb(M_1',M_2)$;
then $M\newprintp N$, by (c-ii) and 
$\ddata(\val{N})=\ddata(\val{M})$.

(2) $\val{M_2}\neq\bot$ and $d\in\ddata(\val{M_2})$. 
This case is symmetric to (1).

(3) $\val{M_1}=\val{M_2}=\bot$ and $d=\bot$.
Then neither $M_1$ nor $M_2$ is a w.h.n.f., hence $M_1\ssp M_1'$ for some $M_1'$.
Therefore, $M \newprintp \Amb(M_1',M_2)$, by (c-ii), 
and we set $N = \Amb(M_1',M_2)$
which has the same value as $M$.
\end{proof}

In the following we write $M\newprintpstard{d}M'$
as an abbreviation for `there exist programs $M_1,\ldots,M_n$ such that 
$M=M_0\newprintp\ldots\newprintp M_n=M'$ and $d\in\ddata(\val{M_i})$ 
for all $i=1,\ldots,n$'.
\begin{lem}\label{lem-claim}
    Let $e$ be a 
    compact
  element of $D$.
  If $M$ is a 
    program of regular type $\rho$,
  $d \in \ddata(\val{M})$ and $e \sqsubseteq d$, 
  then there exists $M'$ such that $M\newprintpstard{d}M'$
  %% there exist $M_0,\ldots,M_n$ such that
  %% $M=M_0\newprintp \ldots \newprintp M_n$,
  %% $d \in \ddata(\val{M_i})$ for all $i=1,\ldots,n$, 
  and 
%  $M'$ such that 
%  $M \newprintp^* M'$, $d \in \ddata(\val{M'})$, and 
  $e \sqsubseteq (M')_{D}$.
%  $e \sqsubseteq M'_{D} \sqsubseteq d$.
\end{lem}
  
\begin{proof}
Induction on $\rk(e)$.
  It is easy to see that the regularity of $\rho$ entails that
  $\tval{\rho}{} = \tval{\Am^k(\sigma)}{}$ 
where $k \in\{0, 1\}$ and  $\sigma$ is neither a fixed point type nor of the form $\Am(\sigma')$.
  
  \emph{Case $e = \bot$}. Then the assertion holds with $n=0$, 
%$M' = M$, 
  since $M_{D} \sqsubseteq d$, by Lemma~\ref{lem-MD}~\ref{lem-MD-val}.
%  for all $d \in\ddata(\val{M})$ (induction on $M$). 
  
  \emph{Case $e = C(\vec{e'})$.}  Note that
$C \ne \Amb$ because $e \sqsubseteq d$ 
and $d$ does not have the form $\Amb(d', d'')$, 
by Lemma~\ref{lem-lem-claim}.
%with $d \in \ddata(\val{M})$. 
We have 
$d = C(\vec{d'})$ with 
$e_i' \sqsubseteq d_i'$.
%$d_i' \sqsubseteq e_i'$.
  
  If $k=0$, then $\rho$ %
is semantically equal to a type of the form 
  $\one$,  $\rho_1+ \rho_2$ or $\rho_1\times\rho_2$,
  and therefore $\val{M}$ has the form $C(\vec{a'})$. 
  By the Adequacy Lemma (Lemma~\ref{lem:ssp}~(\ref{lem:ssp:ade})), 
  $M \ssp^* C(\vec{M'})$ for some $\vec{M'}$
  and $d \in \ddata(\val{M}) = \ddata(C(\vec{\val{M'}}))$, 
in fact $d\in\ddata(\val{K})$ 
    for all programs $K$ in this reduction sequence,
    since the denotation does not change.
  Therefore, by the definition of $\ddata$, $d_i' \in \ddata(\val{M'_i})$.  
  Furthermore, by Lemma~\ref{lem:ssp}~(\ref{lem:ssp:pres}), $C(\vec{M'})$ has 
  the same denotation as $M$ and therefore, each $M_i'$ has a denotation
  in some regular type (for example, if $\rho=\rho_1\times\rho_2$, then 
  $M_1':\rho_1$). 
  Since the ranks of the $e_i'$ are smaller than that of $e$, 
  by   induction hypothesis, there exists $\vec{M''}$ such that
  $M'_i \newprintpstard{d_i'} M''_i$, 
%  $M'_i \newprintp^* M''_i$,
%\UB{$d_i'\in\ddata(\val{K})$ for all programs $K$ in this finite reduction sequence,} 
% $d_i' \in \ddata(\val{M''_i})$ 
  and $e_i' \sqsubseteq ({M_i}'')_{D}$.
The lengths of these reduction sequences may differ.
However, by Lemma~\ref{lem:step-step}, each of them can be extended by an
arbitrary number of steps 
$M''_i\newprintpstard{d_i'} M'''_i$ 
%$M''_i\newprintp^* M'''_i$ 
%such that
%\UB{$d_i'\in\ddata(\val{K})$ for all programs $K$ in this finite reduction sequence,} 
%$d_i'\in\ddata(\val{M'''_i})$, 
and then still
$e_i' \sqsubseteq (M_i'')_{D} \sqsubseteq (M_i''')_{D}$, 
by Lemma~\ref{lem-MD}~\ref{lem-MD-mon}.
%since $N\newprintp N'$ implies $N_D\sqsubseteq N'_D$, and
%$N_D\sqsubseteq d'$ for all $d'\in\ddata(\val{N})$.
Hence we may assume that all the reduction sequences
$M'_i \newprintp^* M''_i$ have the same length.
  Therefore, 
$C(\vec{M'}) \newprintpstard{d} C(\vec{M''})$ 
%$C(\vec{M'}) \newprintp^* C(\vec{M''})$ 
by (p-ii),  
%\UB{$d\in\ddata(\val{K})$ for all programs $K$ in this finite reduction sequence,} 
%  $d \in \ddata(\val{C(\vec{M''})})$, 
and  $e \sqsubseteq C(\vec{M''})_{D}$, so we are done. % \sqsubseteq d$.
%  Since $M \newprintp^* C(\vec{M'})$, we are done.

  If $k = 1$,   then $\val{M}$ has the form $\Amb(a', b') $.
  Since $d \sqsupseteq e \ne \bot$, 
    $a' \ne \bot \land d \in \ddata(a')$
  or $b' \ne \bot \land d \in \ddata(b')$.
  By the Adequacy Lemma, $M \ssp^* \Amb(N_1, N_2)$.
  If   $\defined{\val{N_1}} \land d \in \ddata(\val{N_1})$ then 
  $N_1 \ssp^* N'$ for some w.h.n.f.\ $N'$ and therefore
  $M \newprintc^* N'$ by applying (c-i), (c-ii), and (c-iii),
  and thus $M \newprintp^* N'$ by (p-i).
  Note that   $N' : \Am^{0}(\sigma)$ and
  $d \in \ddata(\val{K})$ for all programs $K$ in this finite 
  $\newprintp$-reduction sequence. 
Therefore, the case $k=0$ applies (which has been proven for all
regular types that do not begin with $\Am$) and 
  there exists $N''$ such that $N' \newprintpstard{d} N''$,
%  $d \in \ddata(\val{K})$ for all programs in this sequence, 
  and $e \sqsubseteq N''_{D}$. % $\sqsubseteq d$.
  Hence $M \newprintpstard{d} N''$. 
%  and $d \in \ddata(\val{K})$ for all programs in this sequence,  
%  and we have the result.
  
  \emph{Case $e =\Fun(g)$}. 
  If $k=0$, then $d = \Fun(f)$ and therefore $\val{M} = \Fun(f)$.
  Furthermore, by the Adequacy Lemma, 
  $M\ssp^* M'$ for some w.h.n.f.\ $M'$.
  Since $\val{M'} = \val{M} = \Fun(f)$,
  $M'$ is a $\lambda$-abstraction and hence $M'_{D} = \val{M'}$.
  It follows that $M'_{D} = d$. 
  Since $\ssp$ does not change the denotation of a program, 
  we have $M \newprintpstard{d}M'$.
  If $k = 1$   the same argument as in the case $e = C(\vec{e'})$ applies. 
\end{proof}

%% \UB{
%% Next we prove a simple special case of Theorem~\ref{thm:dataconv} 
%% and then show that we do not need to worry about fairness.
%% %
%% \begin{lem}\label{lem-val-databot}
%% If $M$ has a regular type and $\bot\in\ddata(\val{M})$, 
%% then there exists a computacdon
%% $M = M_0 \newprintp M_1 \newprintp \ldots$ with
%% $(M_i)_{D}=\bot$ for all $i\in\NN$. 
%% \end{lem}
%% \begin{proof}
%% By Lemma~\ref{lem:ddatabot}, $\val{M} \in\{\bot,\Amb(\bot,\bot)\}$.
%% %
%% If $\val{M}=\bot$, then $M$ has no w.h.n.f., 
%% by Lemma~\ref{lem:ssp}~(\ref{lem:ssp:norm}), i.e.\ 
%% there is an infinite sequence $M=M_0 \ssp M_1 \ssp \ldots$.
%% Hence, $M=M_0 \newprintp M_1 \newprintp \ldots$ is a computation with 
%% $(M_i)_D=\bot$ for all $i$.
%% %
%% If $\val{M}=\Amb(\bot,\bot)$, then $M \ssp^* \Amb(L,R)$ with 
%% %$\bot\in\ddata(\val{L},\val{R})$ and   
%% \mps{¥HT{I deleted $\bot\in\ddata(\val{L},\val{R})$ and   }}
%% $\val{L}=\val{R}=\bot$
%% and consequently, $L$ and $R$ have no w.h.n.f.s, i.e.\ 
%% $L = L_1 \ssp L_2 \ssp \ldots$ and $R = R_1 \ssp R_2 \ssp \ldots$.
%% Through interleaving we obtain the computation
%% %
%% \[ M \ssp^* \Amb(L_0,R_0) \newprintp 
%%            \Amb(L_1,R_0) \newprintp 
%%            \Amb(L_1,R_1) \newprintp
%%            \Amb(L_2,R_1) \newprintp \ldots 
%% \]
%% %
%% which is as required.
%% %
%% \end{proof}
%% }

Before proceeding to the proof of Theorem~\ref{thm:dataconv}, 
we show that we do not need to worry about fairness.

\begin{lem}\label{lem-fairdata}
Let $M_0 \newprintp M_1 \newprintp \ldots$ be an infinite reduction 
sequence of regular typed programs $M_i$. Set $d = \sqcup_{i \in \NN} ((M_i)_{D})$.
Then $M_0 \newprintp M_1 \newprintp \ldots$ is fair iff 
$d\in\ddata(\val{M_i})$ for all $i\in\NN$.
\end{lem}
\begin{proof}
If the sequence $M_0\newprintp M_1 \newprintp \ldots$ is fair, 
then so is each end piece $M_k\newprintp M_{k+1} \newprintp \ldots$
for which $d = \sqcup_{i \ge k} ((M_i)_{D})$, and to which hence 
Theorem~\ref{thm:data} can be applied.

The proof of the converse implication is by coinduction. 
Define a property $P$ of infinite $\newprintp$-reduction sequences as
\begin{enumerate}
\item[] $P((M_i)_{i\in\NN}) \eqdef$\ 
         all $M_i$ have a regular type and 
         $d\in\ddata(\val{M_i})$ for all $i\in\NN$, \\
\hbox{}\hspace{25em} where $d \eqdef \sqcup_{i \in \NN} ((M_i)_{D})$. 
\end{enumerate}
We show that (a) and (b) of the definition of fairness hold
if `$(M_i)_{i\in\NN}$ is fair' is replaced by $P((M_i)_{i\in\NN})$.

(a)  Assume that $P((M_i)_{i\in\NN})$ holds
and let $m$ be such that all $M_{m+i}$ are of the form 
$\Amb(L_i,R_i)$ and all steps $M_{m+i}\newprintp M_{m+i+1}$ are
instances of (c-ii) or (c-ii').
Since $(M_{m+i})_{D} = \bot$ for all $i$, 
we have $d=\bot$ and therefore 
$\bot\in\ddata(\val{\Amb(L_i,R_i)})$, which implies 
$\val{L_i} = \val{R_i} = \bot$ for all $i$,  by Lemma~\ref{lem:ddatabot}. 
It follows that the $L_i$ and $R_i$ have no w.h.n.f.s,
since w.h.n.f.s have a non-$\bot$ value and $\ssp$ is value preserving.

%% If $M_0\newprintp M_1 \newprintp \ldots$ is fair
%% and $m$ is such that all $M_{m+i}$ are of the form
%% $\Amb(L_i,R_i)$ and all steps $M_{m+i}\newprintp M_{m+i+1}$ are
%% instances of (c-ii) or (c-ii'), then $L_0$ and $R_0$ have no w.h.n.f.
%% (and hence no $L_i$ and no $R_i$ has a w.h.n.f.).

(b) Assume that $P((M_i)_{i\in\NN})$ holds
and there exists $m$ such that all $M_{m+i}$ are of the form
$C(N_{i,1},\ldots,N_{i,n})$, with a fixed data constructor $C$.
Since $d\in\ddata(\val{C(N_{i,1},\ldots,N_{i,n})})$, 
we have $d= C(d_1,\ldots,d_n)$ with $d_k\in\ddata(\val{N_{i,k}})$ for all
$k=1,\ldots,n$ and $i\in\NN$. Furthermore, each $(N_{i,k})_{i\in\NN}$ 
($k=1,\ldots,n)$ is a $\newprintp$-reduction sequence of 
regularly typed programs and $d_k = \sqcup_{i \in \NN} ((N_{i,k})_{D})$. 
Hence $P((N_{i,k})_{i\in\NN})$ holds.
\end{proof}

\begin{proof}[Proof of Theorem~\ref{thm:dataconv}]
Let $d\in\ddata(\val{M})$.  
Let $d_0 \sqsubseteq d_1 \sqsubseteq \ldots$ 
  be an infinite sequence of 
      compact
    approximations of
    $d$ such that $d = \sqcup_i d_i$.  
    We construct a sequence $(M_i)_{i\in \NN}$ 
    such that $d \in \ddata(\val{M_i})$ and $M_i$ has a regular type as follows.
    Let $M_0 = M$.  By applying Lemma~\ref{lem-claim} to 
    $d_i$, $d$ and $M_i$, we have $M_{i+1}$ such that 
    $M_i \newprintpstard{d} M_{i+1}$ 
    (hence, by Lemma~\ref{lem-newprintp},
    $M_{i+1}$ has a regular type as well) and
    $d_i \sqsubseteq (M_{i+1})_{D}$. 
%    and furthermore
%    $d \in \ddata(\val{K})$ for all programs $K$ in the finite 
%    reduction sequence $M_i \newprintp \ldots \newprintp M_{i+1}$. 
    By Lemma~\ref{lem:step-step}, we may assume that each of 
    these finite reduction sequences consists 
    of at least one step.
    By concatenating them we obtain an infinite sequence
    $M = N_0 \newprintp N_1 \newprintp \ldots$ such that 
    $d = \sqcup_{i \in \NN} ((N_i)_{D})$ and $d\in\ddata(\val{N_i})$, for all
    $i\in\NN$. By Lemma~\ref{lem-fairdata}, this sequence is fair,
    i.e.\ a computation.
\end{proof}

We say that a closed program $M$ is {\em productive} if every computation
$M = M_0 \newprintp M_1 \newprintp \ldots$ produces 
a deterministic w.h.n.f. (i.e., some $M_i$ is a deterministic w.h.n.f).
We consider this property 
a counterpart of termination for a nondeterministic infinite computation.
Therefore, in informal discussions we sometimes call productive programs \emph{terminating}.

As a consequence of 
the first Adequacy Theorem and Lemma~\ref{lem:ddatabot} we have: 

\begin{thm}\label{cor:ddatabot}
For a program $M$ of regular type, the following 
are equivalent.
  \begin{enumerate}
   \item\label{cor:ddatabot:prod}
   $M$ is productive.
   \item\label{cor:ddatabot:det}
   Some computation $M = M_0 \newprintp M_1 \newprintp \ldots$ 
   produces a deterministic w.h.n.f.
   \item\label{cor:ddatabot:bot}
$\val{M}$ is neither $\bot$ nor $\Amb(\bot,\bot)$.
\end{enumerate}
\begin{proof}
Clearly, every program has a fair $\newprintp$-reduction sequence. 
Therefore, 
(\ref{cor:ddatabot:prod}) implies (\ref{cor:ddatabot:det}). 
  Next, assume 
(\ref{cor:ddatabot:det}). 
Then, by Theorem~\ref{thm:data}, $\ddata(\val{M})$ must contain a
  non-bottom element. By Lemma~\ref{lem:ddatabot}, (\ref{cor:ddatabot:bot}) holds.
  Finally, if (\ref{cor:ddatabot:bot}) holds, then by Lemma~\ref{lem:ddatabot}, 
  $\bot\not\in\ddata(\val{M})$.
  To show that $M$ is productive, consider a computation 
  $M = M_0 \newprintp M_1 \newprintp \ldots$
  By Theorem~\ref{thm:data}, 
  $\sqcup_{i \in \NN} ((M_i)_{D}) \in\ddata(\val{M})$,
  hence $\sqcup_{i \in \NN} ((M_i)_{D})\neq\bot$.
  It follows that for some $i$, $(M_i)_{D}\neq\bot$ which, 
  as noted earlier, is equivalent to $M_i$ being a deterministic w.h.n.f.
\end{proof}
  \end{thm}
The theorem 
does not hold without  the 
regularity
condition.
For example,
the term $M = \Amb(\Amb(\Nil,\Nil), \Amb(\bot, \bot))$,
which has the irregular type $\Am(\Am(\one))$,
can be reduced to $M_1 = \Amb(\bot, \bot)$
and then
repeats $M_1$ forever,  whereas it can also be reduced to $\Nil$.
Hence (2) and (3) hold, while (1) does not.
%Hence does (2) hold, while (1) and (3) don't.
%
Note that, by Lemma~\ref{lem-typing-subst},
the typing judgment $\vdash M:\Am(\Am(\one))$ is not derivable,
since $\Am(\Am(\one))$ is not regular.
McCarthy's $\amb$ operator is bottom-avoiding
in that when it can terminate, it always terminates.
Theorem~\ref{cor:ddatabot}  guarantees a similar property for our 
globally angelic choice operator $\Amb$.

\section{CFP (Concurrent Fixed Point Logic)} 
\label{sec-cfp}

In this section we study four formal systems, 
$\IFP$, $\RIFP$, $\CFP$, and $\RCFP$:
%(Sections~\ref{sub-IFP} - \ref{sub-RCFP}).
%
$\IFP$ (Intuitionistic Fixed Point Logic)
is an intuitionistic first-order logic with strictly positive 
inductive and coinductive definitions from the proofs of which programs 
can be extracted.
$\RIFP$ (Realizability for $\IFP$) extends $\IFP$ to enable the 
formalization of types and programs and their denotational semantics,
as well as realizability,
and hence the formal verification of the extracted programs.
Both systems were introduced in~\cite{IFP}.
$\CFP$ is obtained by adding to $\IFP$ two propositional
operators, $\rt{A}{B}$ and $\Set(B)$, that facilitate the extraction of 
nondeterministic and concurrent programs. 
$\RCFP$ is the extension of $\RIFP$ by 
type and program constructs for concurrency.
While $\RIFP$ is intuitionistic, $\RCFP$ is based on classical logic, 
which is necessary for proving realizability of a concurrent form
of the law of excluded middle.

\subsection{IFP}
\label{sub-IFP}
\paragraph{Language.}
The system $\IFP$ is defined relative to a
many-sorted first-order language. 
IFP expressions consist of formulas, predicates, and operators. 
$\IFP$ formulas 
have the form 
$A \land B$,  $A \lor B$, $A \to B$,  
$\forall x\, A$, $\exists x\, A$, 
$s = t$ ($s$, $t$ terms of the same sort), 
$P(\vec t)$ (for a predicate $P$ and terms $\vec t$ of fitting arities).
Predicates are either predicate constants (as given by the first-order language),
or predicate variables (denoted $X,Y,\ldots$), 
or comprehensions $\lambda\vec x\,A$ (where $A$ is a formula and $\vec x$ 
is a tuple of first-order variables), 
or fixed points $\mu(\Phi)$ and $\nu(\Phi)$
(least fixed point a.k.a.\ inductive predicate 
and greatest fixed point a.k.a.\ coinductive predicate) where $\Phi$ is a
strictly positive (s.p.) operator. 
Operators are of the form
$\lambda X\,Q$ where $X$ is a predicate variable
and $Q$ is a predicate
that has 
the same arity as $X$. 
$\lambda X\,Q$ is s.p.\ if 
every free occurrence of $X$ in $Q$ 
is at a strictly positive position,
that is, at a position that is not in
the left part of an implication (see the corresponding definition for types
in Section~\ref{Sec:2.1}).
Every term has a fixed sort and every predicate variable has a 
fixed arity which is a tuple of sorts. We usually 
suppress sorts and arities, notationally. 
\emph{Notation}: $P(\vec t)$ will also be written $\vec t \in P$, and
if $\Phi$ is 
$\lambda X\,Q$, then
$\Phi(P)$ stands for $Q[P/X]$.
Definitions (on the meta-level) 
of the form 
$P \eqdef \munu(\Phi)$ 
($\munu\in\{\mu,\nu\}$) 
where $\Phi = \lambda X\,\lambda \vec x\,A$,
will usually be written $P(\vec x) \eqmunu A[P/X]$. 
We write $P \subseteq Q$ for 
$\forall \vec{x}\  (P(\vec{x}) \to Q(\vec{x}))$,
$P\equiv Q$ for $(P \subseteq Q)\land(Q \subseteq P)$, 
  $\forall x \in P\ A$ for $\forall x\  (P(x) \to  A)$, and
  $\exists x \in P\ A$ for $\exists  x\  (P(x) \land  A)$.
$\neg A \eqdef A \to \False$ where $\False\eqdef\mu(\lambda X\, X)$
and $X$ has empty arity (i.e.\ $X$ is a propositional variable). 
We identify $(\lambda\vec x\,A)(\vec t)$ with $A[\vec t/\vec x]$ where
$[\vec t/\vec x]$ means capture-avoiding substitution.
Formulas are identified with predicates of empty arity. 
Hence, every statement about predicates is also a statement about formulas.
In particular, $A\equiv B$
means that $A$ and $B$ are equivalent if $A$ and $B$ are formulas.

\begin{example}[Inductively defined predicates]
\label{ex-ind}
In our examples and case study, we work with
an instance of $\IFP$ (and also its extension $\CFP$) 
that contains a sort for real numbers,
whose language includes
constants, operations and relations such as 
$0,1,+,-,*, <, |\cdot|, /$.
 In this instance, one can express the predicate 
$\NN(x)$ that $x$ is a natural number inductively as
$$\NN(x) \eqmu  x = 0 \lor \NN(x-1) $$
which  is shorthand for 
$\NN \eqdef \mu(\lambda X\, \lambda x\, (x = 0 \lor X(x-1)))$.
\end{example}

\paragraph{Proofs.}
The proof rules 
of $\IFP$ are the 
usual natural deduction rules for 
intuitionistic first-order logic with equality 
plus rules for induction and 
coinduction, as shown in Table~\ref{table-proof-ifp},
where in a sequent $\Gamma\vdash A$ it is assumed that 
$\Gamma$ is a finite set of $\IFP$ formulas and $A$ is an $\IFP$ formula.
In the last four rules (in~\cite{IFP} called closure, induction, coclosure, coinduction),
$\Phi$ is a s.p.\ operator.
The induction rule has a strong and a half strong
variant where the premise is weakened to
$\Phi(P\cap\mu(\Phi))\subseteq P$ respectively $\Phi(P)\cap\mu(\Phi)\subseteq P$.
Similarly, the coinduction rule has a strong and a half strong
variant where the premise is weakened to
$P\subseteq\Phi(P\cup\nu(\Phi))$ respectively $P\subseteq\Phi(P)\cup\nu(\Phi)$.
These variants are logically redundant, since they 
can be derived from the original versions. However, they can be given
more efficient realizers than those that would be obtained by extraction from 
their derivations (see~\cite{IFP}). 

\begin{table}
\fbox{\small
\begin{minipage}{\textwidth}
\begin{center}
$\Gamma, A \vdash A$
\hspace{3em} 
$\Gamma \vdash A$
\quad ($A\in\ax$)
\\[1em]
$\Gamma \vdash t=t$
\hspace{3em} 
\AxiomC{$\Gamma\vdash A[s/x]$}
\AxiomC{$\Gamma\vdash s=t$}
             \BinaryInfC{$\Gamma \vdash A[t/x]$}
            \DisplayProof 
\\[1em]
\AxiomC{$\Gamma\vdash A$}
\AxiomC{$\Gamma\vdash B$}
             \BinaryInfC{$\Gamma \vdash A \land B$}
            \DisplayProof 
\hspace{2em} 
\AxiomC{$\Gamma \vdash A \land B$}
       \UnaryInfC{$\Gamma \vdash A$}
            \DisplayProof 
\hspace{2em} 
\AxiomC{$\Gamma \vdash A \land B$}
       \UnaryInfC{$\Gamma \vdash B$}
            \DisplayProof 
\\[1em]
\AxiomC{$\Gamma\vdash A$}
             \UnaryInfC{$\Gamma \vdash A\lor B$}
            \DisplayProof 
\hspace{2em} 
\AxiomC{$\Gamma\vdash B$}
             \UnaryInfC{$\Gamma \vdash A\lor B$}
            \DisplayProof %
\hspace{2em} 
\AxiomC{$\Gamma\vdash A \lor B$}\AxiomC{$\Gamma, A\vdash C$}\AxiomC{$\Gamma, B\vdash C$}             \TrinaryInfC{$\Gamma \vdash C$}
            \DisplayProof \ \ \ \ 
\\[1em]
\AxiomC{$\Gamma, A\vdash B$}
             \UnaryInfC{$\Gamma \vdash A\to B$}
            \DisplayProof 
\hspace{3em} 
\AxiomC{$\Gamma\vdash A \to B$}  
\AxiomC{$\Gamma\vdash A$}
             \BinaryInfC{$\Gamma \vdash B$}
            \DisplayProof \ \ \ \ 
\\[1em]
\AxiomC{$\Gamma \vdash A$}
             \UnaryInfC{$\Gamma \vdash \forall x\,A$}
            \DisplayProof
($x$ not free in $\Gamma$)
\hspace{2em}  
\AxiomC{$\Gamma \vdash \forall x\,A$}
             \UnaryInfC{$\Gamma \vdash A[t/x]$}
            \DisplayProof 
\\[0.5em]
\AxiomC{$\Gamma \vdash A[t/x]$}
             \UnaryInfC{$\Gamma \vdash \exists x\,A$}
            \DisplayProof
\hspace{2em}  
\AxiomC{$\Gamma \vdash \exists x\,A$}
\AxiomC{$\Gamma, A \vdash B$}
             \BinaryInfC{$\Gamma \vdash B$}
            \DisplayProof 
($x$ not free in $\Gamma,B$) 
\\[1em]
$\Gamma \vdash \Phi(\mu(\Phi))\subseteq \mu(\Phi)$
\hspace{3em} 
\AxiomC{$\Gamma \vdash \Phi(P)\subseteq P$}
             \UnaryInfC{$\Gamma \vdash \mu(\Phi)\subseteq P$}
            \DisplayProof 
\\[1em]
$\Gamma\vdash\nu(\Phi) \subseteq \Phi(\nu(\Phi))$
\hspace{3em} 
\AxiomC{$\Gamma \vdash P \subseteq \Phi(P)$}
             \UnaryInfC{$\Gamma \vdash P \subseteq \nu(\Phi)$}
            \DisplayProof 
\end{center}
\end{minipage}
}
\caption{Derivation rules of IFP. \label{table-proof-ifp}}
\end{table}

\paragraph{Axioms.}
$\IFP$
is parametric in 
a set $\ax$ of \emph{axioms}, which have to be 
\emph{non-computational 
(nc)} 
formulas, i.e., 
closed
formulas
built from atomic formulas by the propositional operators $\land$, $\to$,
the quantifiers $\forall$, $\exists$, and least and greatest fixed points.
Disjunction and other logical operators introduced later are excluded.
Axioms should be chosen such that they are true in an intended 
Tarskian model.
Since Tarskian semantics admits classical logic, this means that
a fair amount of classical logic is available through axioms.
For example, for each nc-formula $A(\vec x)$, stability, 
$\forall \vec x\,(\neg\neg A(\vec x) \to A(\vec x))$
can be postulated as an axiom.
The significance of the restriction to nc-formulas is that these are 
identical to their (formalized) realizability interpretation given below.
In particular, Tarskian and realizability semantics coincide for axioms 
in $\ax$.

The systems $\RIFP$, $\CFP$ and $\RCFP$ introduced in the following 
will all be extensions of $\IFP$, in particular they contain the same
set $\ax$ of axioms.
In our examples and our case study (Section~\ref{sec-gray}) we work with an 
instance of $\ax$ consisting of axioms for the usual arithmetic operations 
and the ordering on the real numbers. Since the use of these axioms in a 
proof has no bearing on the extracted program, their exact choice does not 
matter. It suffices if they are (classically) true in the real numbers.

\begin{rem}
\label{rem-nc}
It might look strange that nc-formulas must not contain disjunction 
but may contain the existential quantifier.
The reason lies in the realizability interpretation (Table~\ref{table-realizability}): 
While for an existential formula no witness for the existential quantifier is required,
the realization of a disjunction requires a witness determining which of the disjuncts
is realized. As a consequence, the usual encoding of $A \lor B$ as
$\exists x \,.\, (x = 0 \to A) \land (x \neq 0 \to B)$ does not work,
since to prove the elimination rule,
$(A\lor B) \to (A\to C) \to (B\to C) \to C$, 
w.r.t.\  this encoding,
one needs a case analysis
on the formula $x=0$, or induction on natural numbers.
But this requires the encoding to be relativized to the natural numbers:
$\exists x \,.\, \NN(x) \land (x = 0 \to A) \land (x \neq 0 \to B)$. 
This is 
a formula containing disjunction in the definition of the inductive predicate $\NN$ 
(Example~\ref{ex-ind}).
\end{rem}

\subsection{RIFP (realizability for IFP)}
\label{sub-RIFP}
From an IFP proof of a formula $A$, one can extract a program $M$ that is a realization of the
computational content of $A$.
Realizability is formalized in an extension of $\IFP$, called $\RIFP$.
To give an informal overview,
program extraction is done by
\begin{enumerate}
\item\label{reali-type} 
defining, for each formula $A$, a subdomain $\tau(A)$  
and
a predicate $\rea(A)$ on 
the Scott domain $D$ (see Section~\ref{sub-denot} for the definition of $D$)  
specifying which elements of $\tau(A)$  
realize $A$.
We use types and programs as terms of RIFP to denote subdomains and elements of $D$, respectively.
\item\label{reali-prog} 
showing that 
from a proof of $A$ one can extract a program $M$ such that
the typing rules prove $M : \tau(A)$ and $\RIFP$ proves $\rea(A)(M)$
(and hence the program's denotation 
belongs to the subdomain $\tau(A)$ and satisfies $\rea(A)$).

\end{enumerate}

$\RIFP$ has additional sorts, constants,
and axioms.
$\RIFP$ has the new sorts $\delta$ 
for the domain $D$ 
and $\subd$ for the set of subdomains of $D$,
a binary relation symbol $:$ for the typing relation
(hence, for closed $M$ and $\rho$, $M:\rho$ means $\val{M}\in\tval{\rho}{}$,
in accordance with the notation introduced in Section~\ref{sub-denot})
as well as constants for the type and program constructs in Section~\ref{sec-ang}, 
excluding $\Am$, $\Amb$ and $\strictapp{}{}$. 
It also has axioms describing the denotational semantics of $\RIFP$ programs and types,
which can be found in~\cite{IFP}.
As a consequence, the typing rules of Table~\ref{fig-typing}, excluding those concerning
$\Am$, $\Amb$ and $\strictapp{}{}$,
%$\Amb$ and $\Am$, 
are provable in $\RIFP$.
In addition, RIFP has special predicate variables and type variables, corresponding to IFP 
predicate variables, as well as axioms connecting them, which we describe below.

To avoid `computational garbage' we distinguish between formulas with (nontrivial)
computational content and those with trivial computational content.
The latter are called \emph{Harrop formulas} and are defined as those $\IFP$ formulas
which contain at strictly positive positions neither free predicate variables nor 
disjunctions ($\lor$). 
Their
trivial computational content is represented by the
program $\Nil$.
A formula is \emph{non-Harrop} if it is not Harrop.
The definition of the Harrop property extends to predicates 
in the obvious way.

\begin{table}
\fbox{
\begin{minipage}{\textwidth}  
\begin{align*}
  \tau(P(\vec t)) &= \tau(P)\\
  \tau(A \lor B) &= \tau(A) + \tau(B)\\
  \tau(A \land B) &= \tau(A) \times \tau(B) &\hbox{($A,B$ non-Harrop)}\\
                    &= \tau(A)  &\hbox{($B$ Harrop, $A$ non-Harrop)}\\             
                    &= \tau(B)  &\hbox{($A$ Harrop, $B$ non-Harrop)}\\             
                 &= \one  &\hbox{($A,B$ Harrop)}\\             
  \tau(A \to B) &= \ftyp{\tau(A)}{\tau(B)}  &\hbox{($A,B$ non-Harrop)}\\
                &= \tau(B)  &\hbox{(otherwise)}\\
  \tau(\diamond x\,A) &= 
    \tau(A) &\hbox{($\diamond \in\{\forall,\exists\}$)}\\[.5em]
  \tau(X) &= \alpha_X & \hspace{-1cm}
           \hbox{($X$ a predicate variable, $\alpha_X$ a fresh type variable)}\\
  \tau(P) &= \one &\hbox{($P$ a predicate constant)}\\
  \tau(\lambda \vec x\,A) &= \tau(A)\\
  \tau(\diamond (\lambda X\,P)) &= \tfix{\alpha_X}{\tau(P)}
                    &\hbox{($\diamond \in\{\mu,\nu\}$, $\diamond (\lambda X\,P)$ non-Harrop)}\\
           &= \one &\hbox{($\diamond \in\{\mu,\nu\}$, $\diamond (\lambda X\,P)$ Harrop)}
  \end{align*}
\end{minipage}
}
\caption{Types of $\IFP$ expressions. \label{table-type}}
\end{table}
Table~\ref{table-type} defines the type $\tau(A)$ of an $\IFP$ formula. 
Simultaneously, a type $\tau(P)$ is defined for every $\IFP$ predicate $P$.
For a predicate variable, $\tau(X)$ is a fresh type variable 
$\alpha_X$ representing the unknown type of the unknown predicate $X$.
One easily sees that $\tau(A) = \one$ iff $A$ is a Harrop formula,
and $\tau(A)$ is a regular type for every formula $A$ (Harrop or non-Harrop).
The regularity of $\tau(A)$ will continue to hold for formulas in the extension, $\CFP$,
of $\IFP$ (Lemma~\ref{lem-strict}).

The realizability predicate
$\rea(A)$ is defined 
in Table \ref{table-realizability}, 
by structural recursion on the IFP formula $A$.
It is identical to the interpretation in \cite{IFP} 
except for implications between non-Harrop formulas where 
%we add strictness, i.e.\ 
we require realizers to be defined, i.e.\ different from $\bot$.
This change does not affect the validity of the results in \cite{IFP}\footnote{We checked that the proofs in~\cite{IFP} go through with minor modifications}.
The definedness requirement will be important for the interpretation 
of concurrency in Section~\ref{sub-RCFP}. 
We often write $\ire{a}{A}$ for $\rea(A)(a)$ (`$a$ realizes $A$')
and $\re\,A$ for $\exists\, a\ \ire{a}{A}$ (`$A$ is realizable').
Simultaneously with $\rea(A)$, we define for every IFP predicate $P$ an 
$\RIFP$ predicate $\rea(P)$ with an extra argument for (potential) realizers.
Since Harrop formulas have trivial computational content, 
it only matters whether they are  realizable or not. 
Therefore, we define for a Harrop formula $A$, an $\RIFP$ formula
 $\reah(A)$ that represents the realizability of $A$. For a predicate variable, 
$\rea(X)$ is a fresh predicate variable $\reali{X}$
representing the unknown computational content of the unknown predicate $X$.

\begin{table}
%\label{table-reali-ifp}
%
\fbox{\small
\begin{minipage}{\textwidth}
$\rea(A)$ for Harrop formulas $A$:
\begin{align*}
\rea(A) &= \lambda a\,(a = \Nil \land \reah(A))
\hspace*{17.7em}
\\
\end{align*}

$\rea(A)$ for non-Harrop
formulas $A$:
\begin{align*}
\rea(P(\vec t)) &= \lambda a\,(\rea(P)(\vec t,a))\\
                                        \rea(A\lor B)   &=\lambda c\,(\ex{a}(c=\inl{a}\land\ire{a}{A})\lor
                              \ex{b}(c=\inr{b}\land\ire{b}{B}))\\
\rea(A\land B)  &=\left\{ \begin{array}{ll}
   \lambda c\,(\exists a,b\,(c = \Pair(a,b) \land \ire{a}{A}\land \ire{b}{B}))
                            &\hbox{($A,B$ non-Harrop)}\\
              \lambda a\,(\ire{a}{A} \land \reah(B)) 
                            &\hbox{($B$ Harrop)}\\
              \lambda b\,(\reah(A) \land \ire{b}{B})
                            &\hbox{($A$ Harrop)}
                          \end{array} \right.\\ 
\rea(A\to B)    &= \left\{ \begin{array}{ll}
  \lambda c\,(c:\ftyp{\tau(A)}{\tau(B)} \land c \neq \bot \land  
          \all{a}(\ire{a}{A}\to\ire{(c\,a)}{B})) 
                     &\hbox{($A$ non-Harrop)}\\ 
          \lambda b\,(b:\tau(B) \land (\reah(A) \to \ire{b}{B}))  
                     &\hbox{($A$ Harrop)}
                           \end{array}\right.\\
\rea(\allex x\,A)  &=\lambda a\,(\allex x\,(\ire{a}{A})) 
  \qquad \hbox{($\allex\in\{\forall,\exists\}$)}\\
\end{align*} 

$\rea(P)$ for non-Harrop predicates $P$:
\begin{align*}
\rea(X) &= \reali{X} \\
\rea(\lambda \vec x\,A) &= \lambda (\vec x,a)\,(\ire{a}{A})  \\
\rea(\munu(\lambda X\,P)) &= 
\munu(\lambda\reali{X}\,\rea(P)
[\tfix{\alpha_X}{\tau(P)}/\alpha_X])
 \qquad \hbox{($\munu\in\{\mu,\nu\}$)} 
\hspace*{9em}
\\
\end{align*}
$\reah(A)$ for Harrop formulas $A$:
\begin{align*}
  \reah(P(\vec t)) &= \reah(P)(\vec t) \\ %
\reah(A\land B)  &=
      \reah(A)\land \reah(B) \\
\reah(A\to B)    &= \re\,A \to\reah(B)\\
\reah(\allex x\,A)  &=\allex x\,\reah(A)
  \quad \hbox{($\allex\in\{\forall,\exists\}$)}
\hspace*{18.6em}
\\
\end{align*}
$\reah(P)$ for Harrop predicates  $P$:
\begin{align*}
\reah(P) &= P\quad \hbox{($P$ a predicate constant)}\\
\reah(\lambda \vec x\,A) &= \lambda \vec x\,\reah(A) 
\\
\reah(\munu(\lambda X\,P)) &= \munu(\lambda X\,\reah_X(P))
  \qquad \hbox{($\munu\in\{\mu,\nu\}$)}
\hspace*{16em}
\\
\end{align*}
\begin{itemize}
\item Recall that $\ire{a}{A}$ stands for $\rea(A)(a)$ and 
$\re\,A$ stands for $\exists a\,\rea(A)(a)$.
\item To each $\IFP$ predicate variable $X$ there are 
assigned a fresh type variable $\alpha_X$ and a fresh $\RIFP$ predicate 
variable $\reali{X}$ with one extra argument for domain elements. 
\item $\reah_X(P) \eqdef\reah(P[\pcv{X}/X])[X/\pcv{X}]$
where $\pcv{X}$ is a fresh predicate constant assigned to the (non-Harrop) 
predicate variable $X$. 
This is motivated by the fact that $\munu(\lambda X\,P)$ is Harrop
iff $P[\pcv{X}/X]$ is. The idea is that $\reah_X(P)$ is the same as 
$\reah(P)$ but considering $X$ as a (Harrop) predicate constant.
\end{itemize}
\end{minipage}
}
\caption{Realizability interpretation of $\IFP$}
\label{table-realizability}
\end{table}
The main difference of our interpretation to realizability for intuitionistic 
number theory lies in the treatment of quantifiers. 
While in number theory variables range over natural numbers, 
which have concrete computationally meaningful representations, we make no assumption 
of this kind, since we extract programs from proofs in abstract mathematics.
This is the reason why we interpret quantifiers \emph{uniformly}, that is, 
a realizer of a universal statement does not depend on the quantified variable 
and a realizer of an existential statement does not contain a witness.
A similar uniformity can be found in the Minlog system.
The definition in number theory can be recovered by relativization 
to the predicate $\NN$ (Example~\ref{ex-ind}). 
%
%% The main difference of our interpretation to the 
%% usual realizability interpretation of intuitionistic number theory lies in the
%% interpretation of quantifiers. While in number theory variables range over
%% natural numbers, which have concrete computationally meaningful representations,
%% we make no general assumption of this kind,
%% since it is our goal to extract programs from proofs in abstract mathematics.
%% This is the reason why we interpret quantifiers \emph{uniformly}, that is, 
%% a realizer of a universal statement must be independent
%% of the quantified variable and a realizer of an existential statement does not contain a
%% witness.
%% %
%% A similar uniform interpretation of quantifiers can be found in the
%% Minlog system.
%% The usual definition of realizability of quantifiers in intuitionistic number theory 
%% can be recovered by relativization to the inductively defined predicate $\NN$ in 
%% Example~\ref{ex-ind}, i.e., by writing $\forall x\,(\NN(x) \to A)$.
%

\begin{example}[Natural numbers]
The type $\tau(\NN)$ assigned to the predicate $\NN$ 
(recall that $\NN(x) \eqmu  x = 0 \lor \NN(x-1)$)
is the 
type of unary lazy natural numbers,
$\nat\eqdef  \tfix{\alpha}{1+\alpha}$, introduced in Section~\ref{sec-ang}. 
Realizability for $\NN$ works out as
\[
\ire{a}{\NN(x)} \eqmu (a = \Left \land x = 0)  
\lor \exists b\,(a = \Right(b) \land \ire{b}{\NN(x-1))}\,.
\]
Therefore, the formulas $\NN(0)$, $\NN(1)$, $\NN(2)$, \ldots 
are realized by 
the domain elements $\Left$ ($= \Left(\Nil)$),
$\Right(\Left)$, $\Right(\Right(\Left))$, \ldots,
which means that if $x$ is a natural number, 
then the (unique) realizer of $\NN(x)$ is the unary 
(domain) representation of $x$ introduced in Section~\ref{sec-ang}.
Other ways of characterizing natural numbers may induce different
(e.g.~binary) representations.  
\end{example}
 
\begin{example}[Functions]
As an example of an extraction of a program with function type,
consider the formula 
expressing that the sum of two natural numbers is a natural number,
\begin{equation}
\label{eq:intro1}
\forall x, y\ (\NN(x) \to \NN(y) \to \NN(x+y)).
\end{equation}
It has type $\ftyp{\nat}{\ftyp{\nat}{\nat}}$ and
is realized by a function $f$ that, given realizers of $\NN(x)$ and $\NN(y)$, 
returns a realizer of $\NN(x+y)$, hence $f$ performs addition of unary numbers.
\end{example}

\begin{example}[Non-terminating realizer - this example will be used in Section~\ref{sec-gray}]
\label{ex-d}
Let
$$
\D(x) \eqdef  x\neq 0 \to (x\leq 0 \lor x\geq 0)\,.
$$  
Then 
$\tau(\D) = \bool$ where $\bool = \one + \one$, and $\ire{a}{\D(x)}$ 
is equivalent to
$$
a: \bool \land (x \neq 0 \to (a = \Left \land x \leq 0) \lor 
(a = \Right \land x \geq 0)).
$$
Therefore, $\D(x)$ is realized by $\Left$ if $x < 0$ and by $\Right$ if $x > 0$.
If $x=0$, any element of $\bool$ realizes $\D(x)$,
in particular $\bot$. 
Hence, nonterminating programs of type $\bool$,
which denote $\bot$ by 
Lemma \ref{lem:ssp}~(\ref{lem:ssp:ade}),
realize 
$\D(0)$.
In contrast, \emph{strict} formulas (defined in Section~\ref{sub-CFP}) 
are never realized 
by a nonterminating program, as will be shown in Lemma~\ref{lem-strict}~(\ref{lem-strict-bot})
in Section~\ref{sub-RCFP}.
\end{example}

\subsection{CFP}
\label{sub-CFP}

\paragraph{Language.}
$\CFP$ extends $\IFP$ by two propositional operators,
$\rt{A}{B}$ for restriction, and $\Set(B)$ for concurrency.
Logically, 
these operators 
are equivalent to $A \to B$ and $B$, respectively because
the logical rules in Table \ref{table-infrule} are valid in 
a Tarskian semantics provided we identify $\rt{A}{B}$ with $A \to B$ and 
$\Set(B)$ with $B$.  Their importance relies exclusively on their realizability
interpretations, which are meaningful only when the realizers of $B$ 
are neither $\bot$ nor of the form $\Amb(a,b)$.
Therefore, we require in $\rt{A}{B}$ and $\Set(B)$ 
the formula $B$ to be \emph{strict} in the following inductively defined sense\footnote{The notions ``strict'' in \cite{CFPesop} and ``productive'' in \cite{BergerSpreen23} have the same purpose and imply our notion of strictness.}: 
\begin{itemize}
  \item[-] Harrop formulas and disjunctions are strict 
(the notions of a Harrop formula and a s.p.\ position are extended to $\CFP$ below). 
\item[-] A non-Harrop conjunction is strict if 
\begin{enumerate}
\item[-] at least one of the conjuncts is strict and {\dnh} (see below),
\item[-] or both conjuncts are {\dnh},
\item[-] or it is a conjunction of a Harrop formula and a strict formula.
\end{enumerate}
\item[-] A non-Harrop implication is strict if the premise is {\dnh}.
\item[-] A formula of the form $\diamond x\,A$ ($\diamond\in\{\forall,\exists\}$) or $\munu(\lambda X\lambda\vec x\,A)$ ($\munu\in\{\mu,\nu\}$) is
 strict if $A$ is strict.
 \item[-] Formulas of other forms ($\rt{A}{B}$,  $\Set(B)$, $X(\vec{t})$) are not strict.
\end{itemize} 
The notion of a Harrop formula, referred to above, is extended to
$\CFP$ by disallowing at strictly positive positions 
not only disjunction and free predicate variables, 
but also restriction, $\rt{A}{B}$, and concurrency, $\Set(B)$.
Here, a position in a $\CFP$ expression
counts as strictly positive if for any occurrence of a subformula of the form 
$A \to B$  %
 or $\rt{A}{B}$, that position is not in $A$. 
The operators used for least and greatest fixed point constructions 
are subject to this extended notion of strict positivity as well. 

An expression is \emph{{\dnh}} if it
contains a disjunction or a restriction or concurrency at a s.p. position.
Clearly, a {\dnh} expression is non-Harrop.

A $\CFP$ expression $E$ is called
\emph{well-formed} if for all formulas of the form $\rt{A}{B}$ or
$\Set(B)$ occurring in $E$, the formula $B$ is strict and all
operators occurring in $E$ are strictly positive.
Note that every $\IFP$ expression is also a well-formed $\CFP$ expression.
\begin{lem}
\label{lem-wf}
The following holds for $\CFP$ expressions.
\begin{enumerate}
\item[(a)] Harrop expressions are closed under substitution
of arbitrary terms and predicates (i.e.\ if $E$ is Harrop, so are
$E[t/x]$ and $E[P/X]$ for every term $t$ and every predicate $P$ of fitting
type resp.\ arity).
\item[(b)]  {\Dnh} expressions are closed under substitution
of arbitrary terms and predicates.
\item[(c)]  Strict expressions are closed under substitution
of arbitrary terms and predicates.
\item[(d)] Well-formed expressions are closed under
subexpressions
 (i.e.\ if an expression is well-formed,
so are all of its subexpressions),
as well as substitution of arbitrary terms, and well-formed predicates.
\end{enumerate}
\end{lem}
\begin{proof}
For all parts, closure under substitution by terms is trivial and hence ignored
in the following.
\medskip

(a) and (b) are immediate.
\medskip

(c) is shown by structural induction:
The case of a Harrop expression is solved by (a).

Suppose $A \land B$ is non-Harrop.
If one of the conjuncts is strict and {\dnh}, we use the i.h.\ and (b).
If both conjuncts are {\dnh}, we use (b).
If, say, $A$ is Harrop and $B$ is strict, we use (a) for $A$ and the i.h.\ for $B$.

Suppose $A \to B$ is non-Harrop and $A$ is {\dnh}.
If $B[P/X]$ is Harrop, then $A[P/X]\to B[P/X]$ is Harrop and
hence strict by the first clause.
Otherwise, that formula is a non-Harrop implication and its premise, $A[P/X]$,
is {\dnh} by (b). Hence, it is strict by the third clause.

The remaining cases are straightforward.
\medskip

As for (d), closure under %
subexpressions 
is immediate.

Regarding closure under substitution of 
well-formed
predicates, only restriction and
concurrency need to be looked at. But these cases are solved by (c) since strictness
is closed under substitution of predicates.
\end{proof}
\emph{In the following (i.e.~for the rest of this paper), 
%we assume all occurring $\CFP$ expressions to be well-formed, and 
if we write ``\,$\CFP$ formula'' we always mean ``well-formed $\CFP$ formula'' etc.}

\paragraph{Proofs.}
The inference rules of $\CFP$ are those of $\IFP$ 
(Table~\ref{table-proof-ifp}), 
extended with the rules 
in Table~\ref{table-infrule}. 
%
%where all occurring formulas and contexts are assumed to be %
%$\CFP$ formulas.
%
Since the rules in Table~\ref{table-infrule} do not change the assumptions 
of a sequent, we display them with formulas instead of sequents.
Hence, each premise or conclusion $A$ stands for a sequent $\Gamma\vdash A$
with the same $\Gamma$ in the premises and the conclusion of each rule.
There are no hidden well-formedness assumptions in the rules.
To ensure that well-formedness is preserved, 
Rest-return, Rest-efq, Conc-return and Conc-mp explicitly require $B$
to be strict. Hence, every formula derived from a well-formed context 
is well-formed.
\begin{table}
\medbreak
\noindent
\fbox{\small
\begin{minipage}{\textwidth}

\[
\infer[\hbox{
($A, B_0, B_1$ Harrop) \restintro
%\begin{tabular}{l}($A, B_0, B_1$ Harrop)\\
%\restintro\end{tabular}
}
]{
        \rt{A}{(B_0 \vee B_1)}
}{
A \to (B_0 \vee B_1) \ \ \     \neg A \to B_0 \wedge B_1
}
\]

\smallskip

\[
\begin{array}{ll}
\infer[\hbox{Rest-bind}]{
      \rt{A}{B'}
}{
 \rt{A}{B}\ \ \          B \to (\rt{A}{B'})
}
\ \ \ \ \ \ \ \ & 
\infer[\hbox{($B$ strict) Rest-return}]{   \rt{A}{B}
}{
  B
}  \\\\
  \infer[\hbox{Rest-antimon}]{
    \rt{A'}{B}
    }{
      A' \to A \ \ \ \rt{A}{B}  
}&
  \infer[\hbox{Rest-mp}]{
    B
}{
\rt{A}{B} \ \ \    A
}
\end{array}
\]

\smallskip

\[
\begin{array}{ll}
  \infer[\hbox{($B$ strict) Rest-efq}]{
  \rt{\False}{B}
}{
}
\ \ \ \ \ \ \ \ &
\infer[\hbox{Rest-stab}]{
    \rt{\neg\neg A}{B}
    }{
    \rt{A}{B}
}
\end{array}
\]

\smallskip

\[
  \infer[\hbox{Conc-lem}]{
  \Set(B)
}{
\rt{A}{B}   \ \ \ \     \rt{\neg A}{B}
}
\qquad
  \infer[\hbox{($B$ strict) Conc-return}]{\ 
  \Set(B)
}{
B
}
\]

\smallskip

\[
  \infer[\hbox{($B$ strict) Conc-mp}]{\
\Set(B)
}{
  A\to B\ \ \  \Set(A) 
}\ \ \ \ \ 
\]
\begin{itemize}
\item[] Assumption contexts are omitted since they do not change.
\end{itemize}
\end{minipage}
}
\medbreak
\caption{Inference rules for $\rt{A}{B}$ and $\Set(B)$.\label{table-infrule}}
\end{table}

\subsection{RCFP (realizability for CFP)}
\label{sub-RCFP}
Realizability for $\CFP$ is expressed in $\RCFP$, an extension of $\RIFP$ by the 
type constructor $\Am$, the amb operator $\Amb$
 and the strict application operator $\strictapp{}{}$.
Note that $\RCFP$ is not an extension of $\CFP$ since %
the propositional operators 
$\rt{A}{B}$ and $\Set(B)$ are not included.
Typing and realizability interpretation of $\CFP$ formulas are the 
extension of Table \ref{table-type} 
and Table \ref{table-realizability} with those in Table \ref{table-realizability-cfp}.
\begin{table}
\medbreak
\noindent
\fbox{\small
\begin{minipage}{\textwidth}
\begin{align*}
\tau(\rt{A}{B}) &= \tau(B) \\ \tau(\Set(B)) &= \Am(\tau(B))\ \end{align*}
\begin{align*}
  \rea(\rt{A}{B}) &= \lambda b\,( b\! :\! \tau(B) \land  
                                (\re\, A \to \defined{b}) \land
                               (\defined{b} \to b\,\re\,B)) 
                                         \\
\rea(\Set(B)) &= \lambda c\, \ex{a,b}\, 
      (c = \Amb(a, b) \land a,b:\tau(B) \land (\defined{a} \lor \defined{b})\ \land \\
              &\hspace{5em} (\defined{a} \to a\, \re\, B) \land 
                             (\defined{b} \to b\, \re\, B))
\end{align*}
\end{minipage}
}
\medbreak
\caption{Typing and realizability for $\rt{A}{B}$ and $\Set(B)$\label{table-realizability-cfp}.}
\end{table}
$\RCFP$ is assumed to have a rich enough first-order language
to express all the required domain-theoretic concepts and
enough axioms to prove their properties.
This includes
\begin{enumerate}
\item[-] the axioms of $\RIFP$ (\cite{IFP}, Section~3.4),
\item[-] axioms for strict application
($b\neq\bot \to\strictapp{a}{b}=a\,b$, $\strictapp{a}{\bot}=\bot$) 
\item[-] axioms for the type operator 
$\Am$ ($c:\Am(\rho) \toot (\exists a,b:\rho\,(c=\Amb(a,b)) \lor c=\bot)$),
\item[-] a predicate characterising compactness,
\item[-]the function $\rk$ with axioms for
its properties (Section~\ref{sub-denot}).
\end{enumerate}
Note that, by Axiom (v) of $\RIFP$ (extensionality), $\eta$-conversion 
for functions is
provable (in addition to $\beta$-conversion).
Also, permutative conversion laws for nested case-expressions are provable,
according to~\cite{IFP}, Lemma~12.
We do not need to be too specific about the precise axioms of $\RCFP$,
since $\RCFP$ proofs will not be subject to proof-theoretic manipulation
or analysis.
All that matters is that formulas provable in $\RCFP$
are valid in the domain-theoretic model.

In $\RCFP$ all typing rules of Table~\ref{fig-typing} are provable.
The proof rules of $\RCFP$ are those of $\RIFP$ (extended to $\RCFP$ formulas)
plus the law of excluded middle,
$A \lor \neg A$.

Note that in the name `$\RCFP$' the `C' should be interpreted as `classical' 
(rather than `concurrent') since $\RCFP$ is based on classical logic
but does not have nonstandard constructs for concurrency
(recall that, denotationally, the constructor $\Amb$ is just a pairing 
operator and $\Am$ constructs a type containing such pairs).

In the following, all lemmas and theorems concerning the denotational semantics
of programs and types, including realizability, can be formalized in $\RCFP$.
If the statement is schematic %
in a CFP expression 
(for example, Lemma~\ref{lem-strict} below), the proof is formalizable
for every instance.
\begin{lem}
  \label{lem-strict}
  For every 
$\CFP$ formula $A$:
\begin{enumerate}
   \item\label{lem-strict-bot}
If $A$ is strict, then $\bot$ and domain elements of the form 
     $\Amb(a,b)$ do not realize $A$. 
     Furthermore, $\tau(A)$ is 
    determined.   
   \item\label{lem-strict-reg}
$\tau(A)$ 
is a regular type.
%
%% \item\label{lem-strict-amb}
%% $\Amb(\bot, \bot)$ is not a realizer of $A$,
%% and if $\Amb(a, b)$ realizes $A$, then neither $a$ nor $b$ is
%% of the form $\Amb(u,v)$. 
%
%% \UB{\textbf{$\Amb(\bot,\bot)$ realizes $\bot\to\Set(B)$}}
\end{enumerate}
\end{lem}
  \begin{proof} 
(\ref{lem-strict-bot})
is easily proved by structural induction on formulas.

(\ref{lem-strict-reg})
is easily proved by structural induction on formulas, using the
    fact, proven in 
(\ref{lem-strict-bot}),
that the type of a strict formula is determined.
For least and greatest fixedpoints one 
also needs the following facts (a) and (b), 
both of which can be easily 
proven by induction on $A$:
(a) if $A$ is s.p.\ in $X$, %
then $\tau(A)$ is s.p.\ in $\alpha_X$, %
(b) if $\tau(A)=\alpha_X$ and $\pcv{X}$ is a predicate constant, 
then $A[\pcv{X}/X]$ is Harrop.
%
%% To prove (\ref {lem-strict-amb}), assume that $\Amb(a,b)$ realizes $A$.
%% Then, by the definition of realizability, $A$ must be of the form $\Set(B)$
%% and $a$, $b$ cannot both be $\bot$. Furthermore, by well-formedness, 
%% $B$ is strict and therefore neither $a$ nor $b$ can be of the form 
%% $\Amb(u,v)$, by (\ref{lem-strict-bot}).
%
\end{proof}

\begin{rem}
\label{rem-harrop-char}
The characterization of Harrop formulas as those formulas whose type
equals $\one$,
which is valid for $\IFP$, no longer holds for $\CFP$ since
if $B$ is a Harrop formula, then $\tau(\rt{A}{B}) = \tau(B)=\one$,
but $\rt{A}{B}$ is never Harrop.  However, it is still the case that
Harrop formulas have type $\one$.
\end{rem}

The following lemma says that the result from \cite{IFP} regarding 
the typability
of realizers carries over to $\CFP$.
We use for $E\subseteq D$ the notation
%
%\[\adummy{\rho} \eqdef \lambda (\vec x,a)\,(a:\rho),\]
\[\adummy{E} \eqdef \lambda (\vec x,a)\,E(a),\]
so that $Q\subseteq\adummy{\rho}$   unfolds to 
$\forall (\vec x,a)\,(Q(\vec x,a) \to a:\rho)$ 
(identifying $\rho$ with the predicate $\lambda a \,.\, a:\rho$).
\begin{lem}
\label{lem-realizers-typed}
If $P$ is a $\CFP$ predicate, then 
$\RCFP$ proves $\rea(P)\subseteq\adummy{\tau(P)}$ from the 
assumptions
$\reali{X}\subseteq\adummy{\alpha_X}$ for every free predicate variable 
$X$ in $P$. 
In particular, if $P$ is a formula $A$, then $\ire{a}{A}$ implies $a:\tau(A)$
under these assumptions.
\end{lem}
\begin{proof}
The proof is by structural induction on $P$ and,
in the largest parts, carries over from \cite{IFP}.
For the case that $P$ is a predicate variable $X$, the assumption
$\reali{X}\subseteq\adummy{\alpha_X}$ is used. 
Looking at the definitions of realizability for $\rt{A}{B}$ and $\Set(B)$, 
one sees that they preserve the type correctness of realizers.
\end{proof}

\begin{lem}
\label{lem-prod-notbot}
Every closed program that realizes a closed strict $\CFP$ formula 
% $B$, 
%$M$ 
is productive. 
%iff $\val{M} \neq \bot$.
\end{lem}
\begin{proof}
This follows from 
Lemma~\ref{lem-strict}, %~(\ref{lem-strict-bot}), 
%
%Lemma~\ref{lem-strict}~(\ref{lem-strict-reg}),~(\ref{lem-strict-amb}), 
Lemma~\ref{lem-realizers-typed}, 
%together with 
and Theorem~\ref{cor:ddatabot}.  
\end{proof}

  \begin{rem}
    \label{rem-partial}
    For a closed program $M$ 
and a closed $\CFP$ formula $\rt{A}{B}$,
we explain the operational meaning of $\rea(\rt{A}{B})(M)$:
By combining the definition of $\rea(\rt{A}{B})$ with 
Lemma~\ref{lem-prod-notbot}, 
we see that 
  $\rea(\rt{A}{B})(M)$ says two things: 
(1) if $A$ is realizable, then $M$ is productive
(since if $A$ is realizable, then 
$M$ realizes the strict formula $B$,
hence $M$ is productive, by Lemma~\ref{lem-prod-notbot}),
(2) if $M$ is productive then $M$ realizes $B$
(since if $M$ is productive, then $\val{M}\neq\bot$, by 
Theorem~\ref{cor:ddatabot}, hence $M$ realizes $B$).

  Recall that `productive' is the counterpart of `terminating' for concurrent 
  infinite computation. Therefore, one can consider (2) the 
 `partial correctness'  of $M$ with respect to the specification $B$.
  
  (1) and (2) together imply that if $A$ is realizable, then
 $M$ realizes $B$, a property one can call `conditional correctness' of $M$,
 and which is what $\rea(A \to B)(M)$ means if $A$ is Harrop
 (see Example~\ref{ex-d}).
 However, $\rea(\rt{A}{B})(M)$ says more than that.
 It says that even if $A$ is not realizable, 
 all the defined (i.e., non-bottom) values obtained by computing $M$ are correct.
 This is what we need for concurrent computation as we explained 
 in the introduction.

  To highlight the difference between restriction and implication
in a more concrete situation, 
consider $\rt{A}{(A\lor B)}$ vs.\ $A \to (A \lor B)$
where $A$ is Harrop. Clearly $\Left$ realizes $A \to (A \lor B)$,
but $\Left$ does not realize $\rt{A}{(A\lor B)}$ 
unless $A$ is realizable.
\end{rem}

\begin{rem}
  \label{rem-mult}
Next, let us discuss $\rea(\Set(B))(M)$ from an operational point of view:
By the 
Adequacy Lemma (\ref{lem:ssp}(\ref{lem:ssp:ade})), 
this formula
means that    
$M$ reduces with $\ssp^*$ to $\Amb(N, K)$ 
such that        
both terms $N$ and $K$ are partially correct realizers of $B$ 
and at least one of them is productive.  
    Therefore, by executing 
both 
concurrently and 
taking the one that produces some
value (i.e., is reduced to a deterministic w.h.n.f.), one has a realizer of $B$.
Hence, if $\rea(\Set(B))(M)$, 
then there is a program $M'$ such that $M\newprintp^* M'$  and 
$M'$ realizes $B$. 
This conclusion is a part of the Program Extraction Theorem (Theorem \ref{thm-pe}).
\end{rem}

\begin{rem}
\label{rem-reali-denot}
In the literature, realizability is usually \emph{operational},
i.e.\ realizers are syntactic objects, such as programs or processes 
(or codes thereof) with an operational semantics
(see e.g.\ Kleene's original realizability~\cite{Kleene45}, 
or Krivine's classical realizability~\cite{Krivine03}).
In contrast, in our setting realizability is \emph{denotational}
since realizers can be arbitrary (even uncomputable) 
elements of the Scott domain $D$. 

The connection with the operational semantics of programs arises indirectly. 
Namely, when the denotational realizer extracted from a proof is representable 
by a program term, the Computational Adequacy Theorem allows us to transfer 
denotational correctness to operational correctness. 

The Program Extraction Theorem~\ref{thm-pe} states the realizability 
of a formula by a limit of an infinite sequence of data that can, in general, 
not be represented by a single program.
Therefore, it is necessary to define realizability denotationally, 
for arbitrary domain elements. 
Moreover, the denotational approach 
makes results about realizability more abstract, 
general, and robust, since an extension of the programming language 
will not affect them. 
It also simplifies specification and correctness proof of 
extracted programs~(Section~\ref{sec-pe}), 
since these can be formalized in the system $\RCFP$ whose soundness only 
refers to the denotational semantics.
\end{rem}

\section{Extraction of concurrent programs from CFP proofs}
\label{sec-pe}
We show that the realizability interpretation 
of $\CFP$ is sound in 
the sense that from every $\CFP$ derivation one can extract 
a program realizing the proven formula (Theorem~\ref{thm-soundnessI}). 
In Section~\ref{sub-partial} we show the realizability of the new proof rules 
and hence the soundness of the extracted programs. 

In Section~\ref{subsec:soundness} we show that,
for formulas satisfying a certain syntactic admissibility condition
(defined in Section~\ref
{subsec:soundness}), 
the data of an extracted program $M$ (i.e., the elements of $\ddata(\val{M})$) 
realize the formula obtained by deleting %
$\Set$ and replacing restriction by implication.

Recall that, 
by the Computational Adequacy Theorems \ref{thm:data} and ~\ref{thm:dataconv}, 
the elements of $\ddata(\val{M})$ are exactly the data computed by
the operational semantics of $M$. 

 Most proofs in this section take place in the system $\RCFP$
(Section~\ref{sub-RCFP})
which is based on classical logic. Classical logic is needed 
since for 
the rules (\restintro{}), (Rest-stab), and (Conc-lem), 
the verification that the extracted 
programs are realizers uses the law of excluded middle.
Note that the relation `$d\in\ddata(a)$', which was defined on the meta-level 
in Section~\ref{sub-denot}, can be formalized in $\RCFP$.
This also holds for the rank function and the predicates $\Data$ and $\regD$ 
defined in Section~\ref{subsec:soundness}.
Since we always argue within $\RCFP$ we will refrain
from using semantic brackets from now on. 
For example, we will
write $\ddata(M)$ instead of $\ddata(\val{M})$,  
and the above statement about admissible formulas  
stands for the $\RCFP$ formula
$\rea(A)(M) \to \ddata(M) \subseteq \rea(A^-)$
where $A^-$ is obtained from $A$ by deleting $\Set$ and
replacing restriction by implication.
In fact, Theorem~\ref{thm-faithfulness} shows the more general 
formula $\forall a\,(\rea(A)(a) \to \ddata(a) \subseteq \rea(A^-))$,
for admissible $A$ without free predicate variables.

\subsection{Soundness of extracted programs}
\label{sub-partial}
\label{sub-conc}   

The realizers of the proof rules for restriction and concurrency
in Table~\ref{table-infrule}
are depicted in Table~\ref{table-infrule-sound}.
Proofs of the correctness of the realizers are given in
Lemma~\ref{lem-restrict}.
We use the 
following typable programs:
\begin{eqnarray*}
\leftright
&\eqdef& 
  \lambda b.\,  
\caseof{b}  \{\Left(\_) \to \Left; \Right(\_) \to \Right\}
:
\ftyp{(\rho+\sigma)}{(1+1)}\\
\mapamb
&\eqdef& 
\lambda f.\,\lambda c.\ 
\caseof{c}\,\{\Amb(a,b)\to\Amb(\strictapp{f}{a},\strictapp{f}{b})\}
:
\ftyp{(\ftyp{\rho}{\sigma})}{\ftyp{\Am(\rho)}{\Am(\sigma)}}\\
\seq 
&\eqdef& 
\lambda a.\,\lambda b.\,\strictapp{(\lambda c.\ b)}{a}
:
\ftyp{\rho}{\ftyp{\sigma}{\sigma}}
\quad \hbox{(used infix)}
\end{eqnarray*}
Here, $\rho$, $\sigma$ range over regular types, and for $\mapamb$
it is required in addition that $\rho$ and $\sigma$ are determined.

\begin{table}
\medbreak
\noindent
\fbox{
\begin{minipage}{\textwidth}
\medskip
\[
\infer[{\hbox{\begin{tabular}{l}\restintro{}
($A, B_0, B_1$ Harrop)\end{tabular}}
}]{
        \ire{(\leftright\, b)}{\rt{A}{(B_0 \vee B_1)}}
}{
\ire{b}{(A \to (B_0 \vee B_1))} \ \ \     \reah({\neg A \to B_0 \wedge B_1})
}
\]
\vspace{0.2em}
\[
\begin{array}{ll}
\infer[{ \hbox{Rest-return\ }}]{   \ire{a}{\rt{A}{B}}
}{
  \ire{a}{B}
} 
&
\infer[{
        \hbox{Rest-bind}
        \hbox{\begin{tabular}{l}
                $c=\strictapp{f}{a}$ if $B$ non-Harrop\\
                $c=a\,\seq\,f$ if $B$ Harrop
             \end{tabular}}
       }
]{\ire{c}{\rt{A}{B'}}
}{
\ire{a}{ \rt{A}{B}}\ \ \          \ire{f}{( B \to (\rt{A}{B'}))}
}
\end{array}
\]
\vspace{0.2em}
\[
\begin{array}{ll}
  \infer[{ \hbox{Rest-antimon}}]{
    \ire{a}{\rt{A'}{B}}
    }{
      \ire{}{(A' \to A)} \ \ \ \ire{a}{\rt{A}{B}  }
}\ \ \ \ \ \ \ &
  \infer[{ \hbox{Rest-mp}}]{
    \ire{b}{B}
}{
\ire{b}{\rt{A}{B}} \ \ \    \ire{}{A}
}
\end{array}
\]
\[
\begin{array}{ll}
  \infer[{ \hbox{Rest-efq}}]{
  \ire{\bot}{\rt{\False}{B}}
}{
}
\ \ \ \ \ \ \ \ &
\infer[{\hbox{Rest-stab}}]{
    \ire{b}{\rt{\neg\neg A}{B}}
    }{
    \ire{b}{\rt{A}{B}}
}
\end{array}
\]
\vspace{0.2em}
\[
  \infer[{ \hbox{Conc-lem}}]{
  \ire{\Amb(a,b)}{\Set(B)}
}{ 
\ire{a}{\rt{A}{B}}   \ \ \ \     \ire{b}{\rt{\neg A}{B}}
}
\qquad
  \infer[{ \hbox{Conc-return}}]{\ 
  \ire{\Amb(a,\bot)}{\Set(B)}
}{
\ire{a}{B}
}
\]
\vspace{0.2em}
\[
  \infer[{
            \hbox{Conc-mp}
            \hbox{\begin{tabular}{l}
               $d=\mapamb\ f\ c$ if $A$ non-Harrop\\
               $d=\Amb(f, \bot)$ if $A$ Harrop
                  \end{tabular}}
         }
]
{\
\ire{d}{\Set(B)}
}{
  \ire{f}{(A\to B)}\ \ \  \ire{c}{\Set(A) }
}\ \ \ \ \ 
\]
\smallskip
\end{minipage}
}
\medbreak
\caption{Realizers of the inference rules for $\rt{A}{B}$ and $\Set(B)$.
\label{table-infrule-sound}}
\end{table}

\begin{lem}
\label{lem-restrict}
The rules for restriction and concurrency are realizable, provably in $\RCFP$.
\end{lem}
\begin{proof}  %
Table~\ref{table-infrule-sound} gives an overview of the realizers.
The notation in the form of rules displays
the properties of the realizers in a way that highlights the connection with
the corresponding $\CFP$ rule.
For example 
\medskip

\noindent\emph{\restintro{}}:
  \ \raisebox{-0.2cm}{$\infer[\hbox{($A, B_0, B_1$ Harrop)}]{
          \ire{(\leftright\,b)}{\rt{A}{(B_0 \vee B_1)}}
  }{
  \ire{b}{(A \to (B_0 \vee B_1))} \ \ \     \reah({\neg A \to B_0 \wedge B_1}) 
  }
  $}.\medbreak

\noindent
stands for the $\RCFP$ formula
\[\forall b\,.\,\ire{b}{(A \to (B_0 \vee B_1))} \land
                 \reah({\neg A \to B_0 \wedge B_1})
          \to \ire{(\leftright\,b)}{\rt{A}{(B_0 \vee B_1)}} \]
where $A, B_0, B_1$ range over Harrop formulas.          
To prove it, we assume
\begin{center}
\begin{tabular}{lcl}
$\ire{b}{(A \to (B_0 \vee B_1))}$ &i.e.& 
  $b:\tau(B_0\lor B_1) \land(\reah(A) \to \ire{b}{(B_0 \vee B_1)})$,\\
  $\reah(\neg A \to B_0 \wedge B_1)$ &i.e.& 
   $\neg \reah(A) \to \reah(B_0) \wedge \reah(B_1)$.
\end{tabular}
\end{center}
  We claim that $\rt{A}{(B_0 \vee B_1)}$ 
  is realized by   $\leftright\,b$.
  First, note that the image of $\leftright$ is $\{\bot,\Left,\Right\}$ 
  which is a subset of 
\[\tau(B_0\lor B_1) = (\one+\one)_{\bot} = \{\bot,\Left(\bot),\Right(\bot),\Left,\Right\},\]
  so $\leftright\,b : \tau(B_0\lor B_1)$.
  Assume %
  $\reah(A)$. 
  Then $b$ realizes $B_0 \lor B_1$. 
  Hence $b\in\{\Left,\Right\}$ and therefore $\leftright\,b = b \ne \bot$.
  Now assume $\defined{\leftright\,b}$.
  We do a (classical) case analysis on whether $\reah(A)$ holds.
  If $\reah(A)$, then $\ire{b}{(B_0 \vee B_1)}$.
  Hence $b\in\{\Left,\Right\}$ and therefore $\leftright\,b = b$,
  so $\ire{(\leftright\,b)}{(B_0 \lor B_1)}$.
  If $\neg \reah(A)$, then $\reah(B_0)$ and $\reah(B_1)$. Hence,
  $\Left$ and $\Right$ both realize $B_0 \lor B_1$.
  Since $\defined{\leftright\,b}$,
  $\leftright\, b \in \{\Left,\Right\}$ and 
  therefore $\ire{(\leftright\,b)}{(B_0 \lor B_1)}$
  \footnote{Note that $\leftright\,b$ cannot be replaced by $b$ since if $\neg\reah(A)$, $\reah(B_0)$, $\reah(B_1)$ hold and $b = \Left(\bot)$, then both premises of \restintro{} in Table~\ref{table-infrule-sound} hold, but $b$ does not realize $\rt{A}{(B_0 \vee B_1)}$.}.

  \medbreak\noindent
  \emph{Rest-return}:\ \raisebox{-0.2cm}{$\infer{  
   \ire{a}{\rt{A}{B}}
  }{
    \ire{a}{B}
  }$}. \medbreak
  
  Since $B$ is strict, $\ire{a}{B}$ implies $a \ne \bot$.
  Therefore, clearly $\ire{a}{\rt{A}{B}}$.
  
  \medbreak
  \noindent
  \emph{Rest-bind}:\ \raisebox{-0.2cm}{
  \infer{
        \ire{(\strictapp{f}{a})}{\rt{A}{B'}}
  }{
  \ire{a}{ \rt{A}{B}}\ \ \          \ire{f}{( B \to (\rt{A}{B'}))}
  }
  } 
\medbreak
  
  If $B$ is non-Harrop, then 
  we have $\forall c\,(\ire{c}{B} \to \ire{(f\,c)}{\rt{A}{B'})}$.
  If $\re\,A$ then $\defined{a}$ and  $\ire{a}{B}$, and therefore $f\, a$ 
  realizes $\rt{A}{B'}$.
  Therefore $\defined{f\, a}$ because $\re\,A$.
  Note that $\strictapp{f}{a} = f\, a$ because $\defined{a}$.
  If $\defined{\strictapp{f}{a}}$, then $\defined{a}$.
  Since $\defined{a}$ and $\ire{a}{\rt{A}{B}}$,  we have $\ire{a}{B}$.
  Therefore,  $\strictapp{f}{a} = f\, a$ realizes $\rt{A}{B'}$.
  If $B$ is Harrop,  
    then
  $a\,\seq\,f$ 
 realizes $\rt{A}{B'}$ with a similar argument.
  
  \medbreak
  \noindent \emph{Rest-antimon}: \ \raisebox{-0.2cm}{
  $\infer{
    \ire{a}{\rt{A'}{B}}
      }{
        \ire{}{(A' \to A)} \ \ \ \ire{a}{\rt{A}{B}  }}$}. 
\medbreak

  Clearly, $\ire{a}{\rt{A'}{B}}$
  since $\re\,A'$ implies $\re\,A$.
  
  \medbreak
  \noindent \emph{Rest-mp}: \ \raisebox{-0.2cm}{
  $  \infer{
      \ire{b}{B}
  }{
  \ire{b}{\rt{A}{B}} \ \ \    \ire{}{A}
  }
  $}. 
\medbreak
  
  Clear from the definition of $\ire{b}{\rt{A}{B}}$.
  
  \medbreak
  \noindent \emph{Rest-efq}: \ {
  $\ire{\bot}{\rt{\False}{B}} $} for strict $B$.\medbreak
  
  Clear.
  
  \medbreak
  \noindent \emph{Rest-stab}: \ \raisebox{-0.2cm}{
  $
  \infer{
      \ire{b}{\rt{\neg\neg A}{B}}
      }{
      \ire{b}{\rt{A}{B}}}
  $}.  \medbreak
  
We use classical logic.
Clearly, $\re\,(\neg\neg A)$ is equivalent to  
$\neg\neg (\re\, A)$, and hence, by classical logic, equivalent to $\re\, A$.
Therefore, premise and conclusion 
are equivalent.

  \medbreak
  \noindent \emph{Conc-lem}: \ \raisebox{-0.2cm}{
  $
    \infer{
    \ire{\Amb(a,b)}{\Set(B)}
  }{
  \ire{a}{\rt{A}{B}}   \ \ \ \     \ire{b}{\rt{\neg A}{B}}
  }$}. 
\medbreak
  
  By classical logic
  $\re\,A$, or $\neg(\re\,A)$ i.e.\ $\re\,(\neg A)$. In the first case
  $a \ne \bot$  and in the second case $b \ne \bot$. Further, if $a \ne \bot$,
  then $a$ is a realizer of 
  $B$ since  $a$ realizes $\rt{A}{B}$. Similarly for $b$.
  
  \medbreak
  \noindent \emph{Conc-return}: \ \raisebox{-0.2cm}{
  $  \infer{\   \ire{\Amb(a,\bot)}{\Set(B)}
  }{
  \ire{a}{B}
  }$}. \medbreak
  
  Clear.
  
  \medbreak
  \noindent \emph{Conc-mp}: \ \raisebox{-0.2cm}{
  $  \infer{\
  \ire{(\mapamb\,f\,c)}{\Set(B)}
  }{
    \ire{f}{(A\to B)}\ \ \  \ire{c}{\Set(A) }
  }
  $}
$A$ non-Harrop
\medbreak
  
  We show that $\mapamb\,f\,c$ realizes $\Set(B)$.
  If $\ire{\Amb(a, b)}{\Set(A)}$, then 
  $\defined{a} \lor \defined{b}$.  If $\defined{a}$, then
  $\strictapp{f}{a} = f\ a$ and $\ire{a}{A}$, therefore
  $\ire{(\strictapp{f}{a})}{B}$.  Since $B$ is strict, we have
  $\defined{\strictapp{f}{a}}$.   In the same way, 
  if $\defined{b}$ then $\defined{\strictapp{f}{b}}$.
  Therefore, we have $\defined{\strictapp{f}{a}} \lor \defined{\strictapp{f}{b}}$.  If $\defined{\strictapp{f}{a}}$, then $\defined{a}$ and thus
  $\ire{(\strictapp{f}{a})}{B}$ as we have observed.
  If  $A$ is Harrop,  then, since clearly $\re{A}$,
  it is realized by $\Amb(f,\bot)$.
  \end{proof}
%

%In the proof of the Faithfulness Theorem below, 
The Soundness Theorem~\ref{thm-soundnessI} and the subsequent 
Faithfulness Theorem~\ref{thm-faithfulness}
rely on
%we will frequently use 
the fact that $\RCFP$ derivations are closed under 
substitution of predicates and 
respect extensional equality of predicates (recall from Section~\ref{sub-IFP}
that $P \equiv Q$ stands for $P \subseteq Q \land Q \subseteq P$):
\begin{lem}
\label{lem-respect}
\begin{enumerate}
\item\label{lem-respect-subst}
If $\RCFP$ proves $\Gamma\vdash A$, then $\RCFP$ proves $\Gamma[P/X] \vdash A[P/X]$.
\item\label{lem-respect-equiv}
$\RCFP$ proves $P\equiv Q \vdash A[P/X] \equiv A[Q/X]$.
\item\label{lem-respect-hyp}
$\RCFP$ proves $\Gamma\vdash A[P/X]$ if and only 
if $\RCFP$ proves $\Gamma, P\equiv X\vdash A$, provided $X$ is not free in $P,\Gamma$.
\end{enumerate}
\end{lem}
\begin{proof}
(\ref{lem-respect-subst})
and 
(\ref{lem-respect-equiv})
 can be easily proven by induction on $A$.
(\ref{lem-respect-hyp})
is an immediate consequence of (\ref{lem-respect-subst}) and (\ref{lem-respect-equiv}).
\end{proof}

\begin{thm}[Soundness Theorem I]
\label{thm-soundnessI}
From a 
$\CFP$ proof 
of a 
formula $A$ 
%%%%%%%%%%%%%
without free predicate variables,
%%%%%%%%%%%%%%%
one can extract a 
program $M$ 
such that 
$\vdash M : \tau(A)$
and
$\RCFP$ 
proves $\ire{M}{A}$. 

More generally, let $\Gamma$ be a set of Harrop formulas and
$\Delta$ a set of non-Harrop formulas.
Then, from a 
$\CFP$ 
proof of 
a formula $A$ 
from the assumptions $\Gamma,\Delta$,
one can extract a program $M$ with $\FV(M) \subseteq \vec u$ 
such that $\vec u : \tau(\Delta) \vdash M: \tau(A)$ 
and $\ire{M}{A}$ is 
provable in 
$\RCFP$
from the assumptions $\reah(\Gamma)$,
$\ire{\vec u}{\Delta}$ and $\reali{X}\subseteq\adummy{\alpha_X}$, for all 
free predicate variables $X$ of $\Gamma,\Delta,A$. 
\end{thm}
\begin{proof} %
Since the proof has the same architecture as the proof of the
corresponding Theorem~2 in~\cite{IFP},
we only give a sketch, highlighting the main changes and additions:
One first proves soundness for a modified system $\CFP'$
which is identical to $\CFP$ except that the induction and
coinduction rules have an additional monotonicity premise, 
\begin{center}
\AxiomC{$\Mon(\Phi)$}
\AxiomC{$\Gamma \vdash \Phi(P)\subseteq P$}
             \BinaryInfC{$\Gamma \vdash \mu(\Phi)\subseteq P$}
            \DisplayProof 
\hspace{1em}
\AxiomC{$\Mon(\Phi)$}
\AxiomC{$\Gamma \vdash P \subseteq \Phi(P)$}
             \BinaryInfC{$\Gamma \vdash P \subseteq \nu(\Phi)$}
            \DisplayProof 
\end{center}
where
$\Mon(\Phi)\eqdef X\subseteq Y\to\Phi(X)\subseteq\Phi(Y)$
with fresh predicate variables $X, Y$.
The proof of soundness for $\CFP'$ (corresponding to Theorem~3
in~\cite{IFP}), is by induction on derivations.
%The realizability of the rules of $\IFP$, including the modified
%induction and coinduction rules, is shown exactly as in~\cite{IFP}.
%
The realizability of the rules of $\IFP$ is shown as in~\cite{IFP}, 
with small modifications due to the added 
definedness
%strictness 
requirement 
for realizers of implications between non-Harrop formulas.
%
%% The realizability of the rules of IFP is shown as in [BT21]. The only addition concerns the conjunct IsFun(c) in the realizability of implications between non-Harrop formulas: the realizers extracted for implication introduction, closure and coclosure are λ-abstractions, and those for (strong and half-strong) induction and coinduction are of the form rec(λf. M ∘ N), hence equal to a λ-abstraction. In particular, they may be used as arguments of the realizer of Mon(Φ), as the proof requires.
%
The rules for restriction and concurrency are taken care of
in Lemma~\ref{lem-restrict}. 
For example, regarding the rule \restintro{}, if (by induction hypothesis),
$\RCFP$ proves $\ire{M}{(A \to (B_0 \vee B_1))}$ and 
$\reah({\neg A \to B_0 \wedge B_1})$, then, by Lemma~\ref{lem-restrict},
$\RCFP$ proves $\ire{(\leftright\,M)}{\rt{A}{(B_0 \vee B_1)}}$.

To deal with rules whose premises may contain 
free predicate variables
that do not occur in the conclusion 
(for example, implication elimination), 
these extra predicate variables are replaced by closed predicates (for example, 
$\lambda \vec x\,.\,\False$), which is possible 
by Lemma~\ref{lem-respect}~(\ref{lem-respect-subst})
(see \cite{IFP}~Theorem~3 for details).

The second step is to transform, by a straightforward structural recursion,
every $\CFP$ proof into a $\CFP'$ proof. To transform induction and coinduction,
one needs to provide a $\CFP'$ proof of $\Mon(\Phi)$ for every
s.p.\ operator $\Phi$, which can be done by structural recursion
on $\CFP$ expressions, as in~\cite{IFP}, {Lemma~19}. The concurrency operator
and restriction are dealt with by Rest-bind, Rest-return, and Conc-mp,
the other cases are as in~\cite{IFP}.

The type correctness of the extracted programs for
\restintro{}, Rest-bind and Conc-mp follows from the typings
of $\leftright$, $\mapamb$, and $\seq$; for the $\IFP$ rules
it follows as in~\cite{IFP} (the typing rules in~\cite{IFP} are formulated
without the regularity restriction, but they are only needed for
regular contexts).

We conclude the sketch by making explicit the program extraction 
process contained in this proof by defining a program extraction
function $\ep{\cdot}$ operating on proofs, 
using an extension to $\CFP$ 
of the (hopefully self-explanatory) term
notation for derivations in Section~4.2 of~\cite{IFP}:
\[\ep{d^A} \eqdef \epp{\pt{d^A}}\]
where $\pt{\cdot}$ is the transformation
of $\CFP$ proofs into $\CFP'$ proofs sketched in the
proof of Theorem~\ref{thm-soundnessI}, 
and $\epp{\cdot}$ is the program extraction procedure
for $\CFP'$ proofs.
The transformation $\pt{d^A}$ replaces recursively every 
subderivation of the form $\indu{e^{\Phi(P)\subseteq P}}$ by 
$\induprime{\pt{e^{\Phi(P)\subseteq P}},\monproof{\Phi}^{\monprop{}{\Phi}}}$
where $\monproof{\Phi}^{\monprop{}{\Phi}}$ is the $\CFP'$ proof of monotonicity 
of $\Phi$, and similarly for $\coindu{e^{P \subseteq \Phi(P)}}$.
The function $\epp{\cdot}$ is defined by structural recursion on $d$.
Since, for proofs ending with an $\IFP$ rule, its defining equations 
are exactly as in~\cite{IFP}, we only show the equations for
the new $\CFP$ rules
(we use abbreviations of the rule names of Table~\ref{table-infrule} 
as constructors of proof terms):
\begin{eqnarray*}
\epp{\resti{d^{A\to(B_0\lor B_1)}}{e^{\neg A\to(B_0\land B_1)}}^{\rt{A}{(B_0\lor B_1)}}}
              &=& \leftright\,\epp{d} \\
\epp{\restret{d^{B}}{A}^{\rt{A}{B}}} &=& \epp{d} \\
\epp{\restb{d^{\rt{A}{B}}}{e^{B \to \rt{A}{B'}}}^{\rt{A}{B'}}} &=& 
\left\{
\begin{array}{ll}
\epp{d}\,\seq\,\epp{e} &\hbox{$B$ Harrop} \\
\strictapp{\epp{e}}{\epp{d}} &\hbox{otherwise}
\end{array}
\right. \\
\epp{\resta{d^{A'\to A}}{e^{\rt{A}{B}}}^{\rt{A'}{B}}} &=& \epp{e} \\
\epp{\restmp{d^{\rt{A}{B}}}{e^A}^{B}} &=& \epp{d} \\
\epp{\restefq{B}^{\rt{\False}{B}}} &=& \bot \\
\epp{\rests{d^{\rt{A}{B}}}^{\rt{\neg\neg A}{B}}} &=& \epp{d} \\
\epp{\concl{d^{\rt{A}{B}}}{e^{\rt{\neg A}{B}}}^{\Set(B)}}
      &=& \Amb(\epp{d},\epp{e}) \\
\epp{\concret{d^{B}}^{\Set(B)}} &=& \Amb(\epp{d},\bot) \\
\epp{\concmp{d^{A\to B}}{e^{\Set(A)}}^{\Set(B)}} &=& 
\left\{
\begin{array}{ll}
\Amb(\epp{d},\bot) &\hbox{$A$ Harrop} \\
\mapamb\,\epp{d}\,\epp{e} &\hbox{otherwise}
\end{array}
\right. 
\end{eqnarray*}
\end{proof}
\begin{rem}
\label{rem-harrop-essential}
From the proof of the soundness of the rule \restintro{} in 
Lemma~\ref{lem-restrict} one can see that the restriction of that rule to
Harrop formulas is essential. Therefore, the concept of Harropness 
is not only a matter of optimizing program extraction, 
but also necessary for obtaining soundness in the first place.
\end{rem}

\begin{lem}\label{class-orelim-sleep}
In $\CFP$, one can prove all instances of the following $\CFP$ formulas:
\begin{enumerate}
  \item[(1)] $(B\to B') \to \rt{A}{B} \to \rt{A}{B'}$  (Rest-mon)\medbreak
  \item[(2)] ${\rt{A_0}{B} \to \rt{A_1}{B} \to \neg\neg(A_0\lor A_1) \to \Set(B)}$
  (Class-orelim)\medbreak
  \item[(3)]
${\rt{A_0}{B_0} \to \rt{A_1}{B_1} \to \neg\neg(A_0\lor A_1) \to \Set(B_0\lor B_1)}$\medbreak
\item[(4)]
$A  \to {\rt{\neg A}{B}} $\medbreak
\item[(5)]
$(A \lor B) \to {\rt{\neg A}{B}} $\medbreak
\item[(6)]
$\rt{A_0}{(A \lor B)} \to {\rt{A_0\land\neg A}{B}} $\medbreak
\end{enumerate}

The extracted programs and their types, in the case where
$B,B',B_0,B_1,$ and $A$ are non-Harrop, are given below. For readability, we present the extracted programs in equivalent simplified form, obtained by applying the axioms of $\RCFP$.
In what follows, we refer to these simplified forms simply as the extracted programs.

\begin{enumerate}
  \item[(1)]
  $  \lambda f.\,\lambda b.\, \strictapp{f}{b}\ : \ 
  \ftyp{(\ftyp{\tau(B)}{\tau(B')})}{\ftyp{\tau(B)}{\tau(B')}}$, \smallskip

  \item[(2)]
  $\lambda a.\, \lambda b.\, \Amb(a,b)\ : \ 
  \ftyp{\tau(B)}{\ftyp{\tau(B)}{\Am(\tau(B))}}$, \smallskip
  
 \item[(3)]
 $  \lambda a.\, \lambda b.\, \Amb(\strictapp{\Left}{a},\strictapp{\Right}{b})\  : \ 
 \ftyp{\tau(B_0)}{\ftyp{\tau(B_1)}{\Am(\tau(B_0) + \tau(B_1))}}$,\smallskip
   
  \item[(4)]
  $\lambda a.\, \bot\ : \ 
  \ftyp{\tau(A)}{\tau(B)}$, \smallskip
    
  \item[(5)]
  $\lambda a.\, \caseof{a} \{\Left(\_) \to \bot;\Right(b) \to b\}\ : \ 
  \ftyp{(\tau(A) + \tau(B))}{\tau(B)}$, \smallskip
    
    \item[(6)] The same as (5).
  \end{enumerate}
\end{lem}

\begin{proof}
(1) This is an immediate consequence of (Rest-bind) and (Rest-return).
 The extracted program would, literally, be
$\lambda f.\,\lambda b.\, \strictapp{(\lambda a.\,f\,a)}{b}$, but,
applying $\eta$-reduction, it can be simplified as shown.

(2)  Assume $\rt{A_0}{B}$, $\rt{A_1}{B}$, and $\neg\neg(A_0\lor A_1)$.
By the second assumption and the rule (Rest-stab) we have $\rt{\neg\neg A_1}{B}$.
Since, by the third assumption, we have $\neg A_0 \to \neg\neg A_1$, 
$\rt{\neg A_0}{B}$ follows by (Rest-antimon). 
Together with the first assumption and
the rule (Conc-lem), we get $\Set(B)$.

  (3) Since  $B_i \to B_0\lor B_1$ ($i=0,1$), %
   this is an immediate consequence of
 (Rest-mon) and (Class-orelim), i.e.\ (1) and (2).

(4) Immediate, from (Rest-efq) and (Rest-antimon).

(5)
Suppose that $A \lor B$.  If $A$, then $\rt{\neg A}{B}$ by (4).
 If $B$, then $\rt{\neg A}{B}$ by (Rest-return).

(6)
Suppose $\rt{A_0}{(A \lor B)}$. 
By (Rest-antimon), $\rt{A_0 \land \neg A}{(A \lor B)}$.
By (Rest-antimon) and $(A \lor B) \to \rt{\neg A}{B}$, we have 
$(A \lor B) \to \rt{A_0 \land \neg A}{B}$.
Therefore, by (Rest-bind), $\rt{A_0 \land \neg A}{B}$.
\end{proof}

\begin{example}
  \label{ex-dprime}
Continuing Example~\ref{ex-d}, we modify $\D(x)$ to 
$$\D'(x) \eqdef  \rt{x\neq 0}{(x\leq 0 \lor x\geq 0)}\,.$$
A realizer of $\D'(x)$, which has type $\bool$, may or may not terminate
(non-termination only occurs when $x = 0$).
However, in case of termination,  the result is guaranteed to 
realize $x\le 0 \lor x\ge 0$.
Note that, a realizer of $\D(x)$ also has type $\bool$ and may or may not terminate,
but there is no guarantee that
it realizes $x\le 0 \lor x\ge 0$ when it does terminate.
Nevertheless, 
$\D \subseteq \D'$ follows from (\restintro{}) 
(since $\D(x)$ is $x\ne 0\to x\le 0 \lor x\ge 0$, and 
$\neg x\ne 0\to x\le 0 \land x\ge 0$ is provable using stability of equality) 
and is realized by $\leftright$.
$\D' \subseteq \D$ holds by the rule Rest-mp.
\end{example}
\begin{example}
\label{example-ConsSD}
This 
builds on Examples~\ref{ex-d} and~\ref{ex-dprime}
and will be used in 
Section~\ref{sec-gray}.
Let $\tent(x) = 1-2|x|$ and consider the predicates
$\E(x) \eqdef  \D(x) \land \D(\tent(x))$ and
\[
\ConSD(x)  \eqdef  \Set((x\leq 0 \lor x\geq 0) \lor |x| \leq 1/2).
\]
We show $\E \subseteq \ConSD$:
From $\E(x)$ and Example~\ref{ex-dprime} we get $\D'(x)$ and $\D'(\tent(x))$
which unfolds to $\rt{x\neq 0}{(x\le 0  \lor x \ge 0 )}$ and 
$\rt{|x|\neq 1/2}{(|x| \geq 1/2  \lor |x| \leq 1/2 )}$. 
By Lemma \ref{class-orelim-sleep}~(6),
$\rt{|x| < 1/2}{(|x| \leq 1/2)}$. 
Since $\neg\neg((x \ne 0) \lor |x| < 1/2)$, we have $\ConSD(x)$ 
by  Lemma \ref{class-orelim-sleep}~(3).
 Moreover, $\tau(\E) = \bool\times \bool$ and 
$\tau(\ConSD) = \Am(\tri)$  where $\tri \eqdef \bool + \one$.
The extracted realizer of $\E \subseteq \ConSD$ is
\begin{align*}
&\conSD \eqdef \lambda c.\, 
   \caseof{c} \{ \\
& \hspace*{7em} \Pair(a, b) \to  
                  \Amb(\strictapp{\Left}{(\leftright\, a)},\\
& \hspace*{15.7em}        \Right {\downarrow} 
                         (\caseof{b} \{\Left(\_) \to \bot;\\
& \hspace*{23.9em}     \Right(\_) \to \Nil\}))\}
\end{align*}
   of type $\tau(\E \subseteq \ConSD) =  \ftyp{\bool\times\bool}{\Am(\tri)}$~\footnote{Extracted literally, the inner case term is $\caseof{\leftright\,b} \{\Left(\_) \to \bot; \Right(c) \to c\}$, which simplifies to 
$\caseof{b} \{\Left(\_) \to \bot; \Right(\_) \to \Nil\}$.}   .
Explanation of this program:
$a$ is $\Left$ or $\Right$ depending on whether $x \leq 0$ or $x \geq 0$
but may also be $\bot$ if $x = 0$.
$b$ is $\Left$ or $\Right$ depending on whether 
$|x| \geq 1/2$ or $|x| \leq 1/2$ 
but may also be $\bot$ if $|x| = 1/2$.  Since $x = 0$ and $|x| = 1/2$ do not
happen simultaneously, by evaluating $a$ and $b$ concurrently, we obtain one of them
from which we can determine one of the cases $x \leq 0$, $x \geq 0$, or $|x| \leq 1/2$.
\end{example}

\subsection{Soundness of extracted data}
\label{subsec:soundness}
In $\CFP$, we have a second Soundness Theorem 
which ensures the correctness
of all computations, i.e.\ fair reduction paths, 
of an extracted program $M$. 
More precisely, we show in Theorem~\ref{thm-faithfulness} that 
if  $M$ realizes %
 an admissible (definition below)  $\CFP$ formula $A$ 
without free predicate variables, 
then all $d\in\ddata(M)$ realize the formula $A^-$ 
(recall from the introduction to Section~\ref{sec-pe} that $A^-$ is 
obtained from $A$ by deleting all concurrency operators $\Set$ 
and replacing every restriction by the corresponding implication).
Hence, by Computational Adequacy, all computations of $M$ 
converge to a realizer of $A^-$ (Theorem~\ref{thm-pe}).
Since $A^-$ is an $\IFP$ formula, this shows that 
the realizability interpretation of $\CFP$ is faithful to
that of
$\IFP$.
The proof of 
this result
requires some preparation.

Since the domain $D$ is algebraic (i.e.\ every element of $D$ is the 
directed supremum of compact elements), 
subdomains of $D$ can be characterized as follows. 
\begin{lem}
\label{lem-subdom}
A subset $X\subseteq D$ is a subdomain of $D$ iff the following two conditions hold:
\begin{enumerate}
\item[(i)] 
For all $a\in D$, $a$ is in $X$ iff all compact approximations of $a$ are in $X$.
\item[(ii)] If $a_0,b_0\in X$ are compact and consistent (in $D$), 
then $a_0\sqcup b_0\in X$.
\end{enumerate}
\end{lem}
\begin{proof}
The easy proof is omitted.
\end{proof}
We note that Lemma~\ref{lem-subdom}, as well as Lemmas~\ref{lem-regular},
\ref{lem-data-nonempty},
and \ref{lem-data-compact} below
can be formalized in $\RCFP$.

We define coinductively two subsets of $D$. The first, $\Data$, disallows
the constructor $\Amb$ altogether, the second, $\regD$, disallows immediate
nestings of $\Amb$ and can be seen as a semantic counterpart to the syntactic
regularity property of types. 
We call elements of $\regD$ \emph{regular}.
\begin{eqnarray*}
\Data &\eqnu& 
  (
   \{\Nil\} \cup 
   (\Data\times \Data) \cup 
   (\Data+\Data) \cup 
   (\ftyp{D}{D}))_{\bot} \\
\regD &\eqnu& 
  (
   \{\Nil\} \cup 
   (\regD\times \regD) \cup 
   (\regD+\regD) \cup 
   (\ftyp{D}{D}) \cup\\
 &&\hspace{0.5em}  \Amb(\regD\setminus\Amb(D,D),\regD\setminus\Amb(D,D))
  )_{\bot} 
\end{eqnarray*}
Clearly, $\Data\subseteq\regD\subseteq D$.

\begin{lem}
\label{lem-regular}
\begin{enumerate}
\item\label{lem-regular-sub}
$\Data$ and $\regD$ are subdomains of $D$. %
\item\label{lem-regular-data}
$\rho\subseteq\Data$ for every regular type $\rho$ without $\Am$
under the assumption that $\alpha\subseteq\Data$ for every 
free type variable of $\rho$.\footnote{Note that this is an informal rendering of the statement ``For every regular type $\rho$ without $\Am$, the formula 
$\forall \vec\alpha\,.\,\vec\alpha\subseteq \Data \to \rho\subseteq\Data$
is provable in $\RCFP$,
where $\vec\alpha$ is a list of the free type variables of $\rho$,
and $\rho\subseteq\Data$ is shorthand for $\forall d\,.\,d:\rho\to \Data(d)$.''
Other statements of this kind should be read in a similar way.
}

\item\label{lem-regular-reg}
$\rho\subseteq\regD$ for every regular type $\rho$
under the assumption that $\alpha\subseteq\regD$ for every 
free type variable of $\rho$.
\item\label{lem-regular-rea}
If all free predicate variables of the %
$\CFP$ predicate $P$ occur strictly positively, then 
$\rea(P)\subseteq\adummy{\regD}$, 
under the assumption that 
$\reali{X}\subseteq\adummy{\regD}$ and 
$\alpha_X\subseteq\regD$ for every 
free predicate variable $X$ of $P$.
\end{enumerate}
\end{lem}

\begin{proof}
The proof of 
Part~(\ref{lem-regular-sub}) 
is easy using the characterization of subdomains 
in Lemma~\ref{lem-subdom}. We skip details.

Part~(\ref{lem-regular-data}) 
is proven by structural induction on $\rho$.
For a type variable, this holds by assumption. 
For $\one, \rho\times\sigma,\rho+\sigma, \ftyp{\rho}{\sigma}$ it suffices to observe that
$\bot$, $\Nil$ and $\Fun(f)$ are data (for arbitrary $f\in[D\to D]$),
and that the constructors $\Pair,\Left,\Right$ preserve the property $\Data$.
To show that $\tfix{\alpha}{\rho}\subseteq\Data$ it suffices to show that
$\rho\subseteq \Data$ under the assumption that $\alpha\equiv\Data$. But this holds by 
the structural induction hypothesis.

Part~(\ref{lem-regular-reg}) 
is proven again by structural induction on $\rho$,
using the fact that for a determined type $\rho$ the non-bottom elements of 
$\tval{\rho}{\zeta}$ are not of the form  $\Amb(a, b)$.

We prove 
Part~(\ref{lem-regular-rea}) 
first with a strengthening of the assumption $\alpha_X\subseteq\regD$
to $\alpha_X\equiv\regD$ 
(keeping the assumption $\reali{X}\subseteq\adummy{\regD}$ unchanged).  
Under this strengthened assumption, $\reali{X}\subseteq\adummy{\alpha_X}$ and therefore
$\rea(P)\subseteq\adummy{\tau(P)}$, by Lemma~\ref{lem-realizers-typed}.
But, since $P$ is well-formed,  $\tau(P)\subseteq \regD$, by 
Part~(\ref{lem-regular-reg}). 
To prove 
Part~(\ref{lem-regular-rea}) 
in general, we observe that, since all free predicate variables
are s.p. in $P$,  $\rea(P)$ depends monotonically on $\alpha_X$.
Therefore, if $\alpha_X\subseteq\regD$, 
the inclusion $\rea(P)\subseteq \adummy{\regD}$ continues to hold.
This proves 
Part~(\ref{lem-regular-rea}) 
provided the assumption $\alpha_X\equiv\regD$ is consistent, 
that is, a type $\alpha_X$ satisfying it actually does exist.
But, by 
Part~(\ref{lem-regular-sub}), 
$\regD$ is a subdomain of $D$ and hence a possible
value for $\alpha_X$.
\end{proof}

\begin{lem}\label{lem-data-nonempty}
\begin{enumerate}
\item\label{lem-data-nonempty-ne} 
$\ddata(a)$ is nonempty for all $a\in D$.
\item\label{lem-data-nonempty-sing} 
If $a\in\Data$, then $\ddata(a) = \{a\}$.
\item\label{lem-data-nonempty-data} 
If $a\in\regD$, then $\ddata(a) \subseteq \Data$.
\end{enumerate}
\end{lem}
\begin{proof}

(\ref{lem-data-nonempty-ne}) 
For every nonempty finite increasing sequence 
$\vec a,a_n = a_0,\ldots,a_{n-1},a_n$ of compact
elements $a_i\in D$ ($i\le n$) we define $\cd{\vec a,a_n}\in D$  
($\cd{}$ for ``choose data'') by recursion on the rank of $a_n$: 
\begin{itemize}
\item[] If $a_n\in\{\bot,\Nil,\Fun(f)\}$, then $\cd{\vec a,a_n} \eqdef a_n$.
\item[] If $a_n=\Amb(b_n,c_n)$, then let 
$\vec b = b_0,\ldots,b_{n-1}$ and $\vec c = c_0,\ldots,c_{n-1}$
such that $a_i = \Amb(b_i,c_i)$ if $a_i\neq\bot$, and $b_i=c_i=\bot$ %
 otherwise.
Now set 
\[ \cd{\vec a,a_n} \eqdef \left\{
\begin{array}{ll}
   \bot            & \hbox{if }b_n=c_n=\bot\\ 
    \cd{\vec b,b_n} & \hbox{if $b_n\neq\bot$ and, for all $i<n$, $b_i=\bot$ implies $c_i=\bot$}\\
    
   \cd{\vec c,c_n} & \hbox{otherwise}
\end{array}
\right.
\]
\item[] If $a_n=\Left(b_n)$, then let 
$\vec b = b_0,\ldots,b_{n-1}$ such that $a_i = \Left(b_i)$ if $a_i\neq\bot$, 
and $b_i=\bot$, otherwise.
Set $\cd{\vec a,a_n} \eqdef \Left(\cd{\vec b,b_n})$.
The cases that $a_n$ begins with $\Right$ or $\Pair$ are similar.
\end{itemize}
It is easy to see that 
if $a_n\dle a_{n+1}$, then $\cd{\vec a,a_n}\dle\cd{\vec a,a_n,a_{n+1}}$.

Now let $a\in D$. Since $D$ is algebraic with a countable base, 
there exists an increasing sequence of compact elements, $(a_n)_{n\in\NN}$,  
that has $a$ as its supremum. 
We show that for
$d \eqdef \bigsqcup\{\cd{a_0,...a_{n-1},a_n}\mid n\in\NN\}$, 
we have $d\in\ddata(a)$. 
To this end, we define the relation $P(a,d)$ %
as `$a$ is the supremum of an increasing sequence 
$(a_n)_{n\in\NN}$ of compact elements such that 
$d = \bigsqcup\{\cd{a_0,\ldots,a_{n-1},a_n}\mid n\in\NN\}$'
and prove:
\paragraph{Claim} 
If $P(a,d)$ then $d\in\ddata(a)$, for all $d,a\in D$.

We prove the Claim by coinduction: Assume 
$P(a,d)$. %
We have to show that the right-hand side of the definition of 
$\ddata$ in Section \ref{sub-denot}
holds if $\ddata$ is replaced by $P$.
Let $a$ be the supremum of the increasing sequence $(a_n)_{n\in\NN}$ 
of compact elements.

\emph{Case $a=\Amb(b,c)$}.
Then, for some $m$, $a_k=\bot$ for all $k<m$ and $a_k=\Amb(b_k,c_k)$
for all $k\ge m$. Set $b_k=c_k=\bot$ for $k<m$.
If $b=c=\bot$ then $\cd{a_0,\ldots,a_{k-1},a_k}=\bot$
 for all $k$ 
and therefore $d=\bot$ which is correct.
If $b$ and $c$ are not both $\bot$, then there is a least $n\ge m$ such that 
$a_n =\Amb(b_n,c_n)$ and
$b_n,c_n$ are not both $\bot$. 
If $b_n\neq\bot$, then 
$\cd{a_0,\ldots,a_{k-1},a_k}= \cd{b_0,\ldots,b_{k-1},b_k}$ for all $k\ge n$.
Since for $k<n$, $\cd{a_0,\ldots,a_{k-1},a_k}= \bot=\cd{b_0,\ldots,b_{k-1},b_k}$,
$P(b, d)$ holds
and we are done. 
If $b_n=\bot$, then $c_n\neq\bot$ and, with a similar argument,  
$\cd{a_0,\ldots,a_{k-1},a_k}= \cd{c_0,\ldots,c_{k-1},c_k}$ for all $k \geq n$
which implies that $P(c,d)$ holds and we are done again.

The other cases are easy.

\medskip

(\ref{lem-data-nonempty-sing})
To show that $a\in\Data$ implies $\{a\}\subseteq\ddata(a)$, 
one sets $P(a,b) \eqdef a=b\in\Data$ and shows 
$P\subseteq\ddata$ by coinduction.

To show that  $a\in\Data$ implies $\ddata(a)\subseteq \{a\}$, 
one proves that for all compact $a_0\in D$,
if $a\in\Data$ and $d\in\ddata(a)$, then $a_0\dle a$ iff $a_0\dle d$. 
The proof can be done by induction
on $\rk(a_0)$. 

\medskip

(\ref{lem-data-nonempty-data})
One shows $\forall d\,(\exists a\in\regD\,(d\in\ddata(a)) \to d\in\Data)$
by coinduction.
\end{proof}
\begin{rem}
The proof of (\ref{lem-data-nonempty-ne})  
is constructive since for every compact approximating sequence 
of $a$, a compact approximating sequence of some element in $\ddata(a)$ is constructed. 
In particular, if $a$ is computable, then $\ddata(a)$ contains a computable element. 
However, there can be no computable function $f:D\to D$ such that $f(a)\in\ddata(a)$ 
for all $a\in D$, since such a function cannot even be monotone: 
We would necessarily have $f(\Amb(0,\bot))=0$ and $f(\Amb(\bot,1))=1$, hence
$f$ would map the consistent inputs $\Amb(0,\bot)$ and $\Amb(\bot,1)$ 
(`consistent' meaning `having a supremum') to the inconsistent outputs $0$ and $1$,
which is impossible for a monotone function.
\end{rem}
\begin{lem}
\label{lem-data-compact}
\begin{enumerate}
\item\label{lem-data-compact-elem}
If $a_0$ is compact (in $D$), 
then all elements of $\ddata(a_0)$ are compact.
\item\label{lem-data-compact-approx}
If $a_0$ is a compact approximation of $a$,
then for every $d_0\in\ddata(a_0)$ there exists 
some $d\in\ddata(a)$ such that $d_0\dle d$.
\item\label{lem-data-compact-reg}
If $a$ is regular and $d\in\ddata(a)$, 
then for every compact approximation $d_0$ of 
$d$ there exists a compact approximation $a_0$ of $a$ 
such that $d_0\in\ddata(a_0)$.
\item\label{lem-data-compact-sup}
If $c = a\sqcup b$ where $a,b$ are compact (hence $c$ is compact),  
then for every $w\in\ddata(c)$, 
there are $u \dle u' \in \ddata(a)$ and $v \dle v' \in \ddata(b)$ 
with $w = u \sqcup v$.
%
%then for every $w\in\ddata(c)$, $w\in\ddata(a)\cup\ddata(b)$ or 
%$w = u \sqcup v$ for some $u\in\ddata(a)$ and $v\in\ddata(b)$.
\end{enumerate}
\end{lem}

\begin{proof}
(\ref{lem-data-compact-elem})
Easy induction on $\rk(a_0)$.
\medskip

(\ref{lem-data-compact-approx})
Induction on $\rk(a_0)$.
If $a_0\in\{\bot,\Amb(\bot,\bot)\}$, then $d_0=\bot$ 
and we can take $d$ to be any element of 
$\ddata(a)$ (which is nonempty, as shown in 
Lemma~\ref{lem-data-nonempty}~(\ref{lem-data-nonempty-ne})).
If $a_0=\Amb(b_0,c_0)$ where, w.l.o.g., $b_0\neq\bot$ and $d_0\in\ddata(b_0)$, then,
since $a=\Amb(b,c)$ with $b_0\dle b$, by i.h.\ there is some $d\in\ddata(b)$ 
(hence also $d\in\ddata(a)$) with $d_0\dle d$.
The other cases are easy.
\medskip

(\ref{lem-data-compact-reg})
Induction on $\rk(d_0)$. 
The most complicated case is $a = \Amb(b,c)$ where $b\neq\bot$ and $d\in\ddata(b)$. 
Since $a$ is regular, $b$ starts with one of the 
constructors $\Nil$, $\Fun$, $\Left$, $\Right$, $\Pair$.
W.l.o.g.\ assume $b=\Pair(b^1,b^2)$. Then $d=\Pair(d^1,d^2)$ with $d^i\in\ddata(b^i)$
and $d_0=\Pair(d^1_0,d^2_0)$ with $d^i_0\dle d^i$. 
By i.h.\ there are compact approximations
$b^i_0$ of $b^i$ such that $d^i_0\in\ddata(b^i_0)$. 
Hence $\Pair(d^1_0,d^2_0)\in\ddata(\Pair(b^1_0,b^2_0))$. 
Hence, $a_0 \eqdef \Amb(\Pair(b^1_0,b^2_0),\bot)$ is a compact approximation of $a$
with $d_0\in\ddata(a_0)$.
\medskip

(\ref{lem-data-compact-sup})
Induction on $\rk(c)$.
We assume that $a,b,c,w$ are all different from $\bot$, 
otherwise, the solution is easy.

Case $c=\Amb(c^1,c^2)$. 
Then $a=\Amb(a^1,a^2)$ and $b=\Amb(b^1,b^2)$ with $c^i = a^i\sqcup b^i$,
and, w.l.o.g., $c^1\neq\bot$ and $w\in\ddata(c^1)$.
By the i.h., there are $u \dle u'\in\ddata(a^1)$ and $v \dle v'\in\ddata(b^1)$ 
with $w = u \sqcup v$. 
If $a^1\neq\bot$ and $b^1\neq\bot$, then $\ddata(a^1) \subseteq \ddata(a)$ and
$\ddata(b^1) \subseteq \ddata(b)$, hence we are done.
If, say, $a^1=\bot$, then $u'=\bot$, hence $u=\bot$. 
Furthermore, $b^1\neq\bot$, since $c^1\neq\bot$, hence $v' \in \ddata(b)$.
Using Lemma~\ref{lem-data-nonempty}~(\ref{lem-data-nonempty-ne}), let
$u_0'\in\ddata(a)$. Hence 
$u \dle u_0'\in\ddata(a)$ and $v \dle v'\in\ddata(b)$, as required.

Case  $c=\Pair(c^1,c^2)$. 
Then $a=\Pair(a^1,a^2)$ and $b=\Pair(b^1,b^2)$ with $c^i = a^i\sqcup b^i$,
and $w = \Pair (w^1,w^2)$ with $w^i\in\ddata(c^i)$.
By the i.h., there are $u^i \dle (u^i)'\in\ddata(a^i)$ and 
$v^i \dle (v^i)'\in\ddata(b^i)$ with $w^i = u^i \sqcup v^i$. 
Set $u = \Pair(u^1,u^2)$, $u' = \Pair((u^1)',(u^2)')$, 
$v = \Pair(v^1,v^2)$, $v' = \Pair((v^1)',(v^2)')$.
Then $u \dle u' \in\ddata(a)$, $v \dle v' \in\ddata(b)$ with 
$u \sqcup v = w$, as required.

The same argument as for $\Pair$ works for the other data constructors.

%% By i.h., there are two cases: 
%% 1. $w\in\ddata(a^1)\cup\ddata(b^1)$, say $w\in\ddata(a^1)$. 
%% Then $a^1\neq\bot$ and $w\in\ddata(a)$.
%% 2. $w = u \sqcup v$ for some $u\in\ddata(a^1)$ and $v\in\ddata(b^1)$.
%% If $a^1=\bot$ or $b^1=\bot$, say $a^1=\bot$, then $u=\bot$ and $w=v$ and $w\in\ddata(b)$.
%% If $a^1$ and $b^1$ are both different from $\bot$, then $u\in\ddata(a)$ and $v\in\ddata(b)$ 
%% and we are done as well.

%% The other cases are easy.
\end{proof}

We now provide the technical definition needed for 
the Faithfulness Theorem~\ref{thm-faithfulness}:

A CFP formula is a \emph{functional implication} if it is of the form 
$A \to B$ where $A$ and $B$ are both non-Harrop. 

A CFP expression is \emph{quasi-closed} if it contains no free predicate variables.

A $\CFP$ expression is \emph{admissible} if
\begin{enumerate}
\item[(i)] it contains free predicate variables, the concurrency operator,
and restriction at s.p.~positions only,
%\item[(ii)] every occurring restriction $\rt{A}{B}$ has a Harrop premise $A$,
\item[(ii)] for every subformula of the form $\Set(B)$ or $\rt{A}{B}$, 
$A$ is Harrop and $B$ is non-Harrop,
\item[(iii)] every occurrence of a functional implication
is part of a quasi-closed subexpression 
without $\Set$ and restriction.\footnote{The definition of admissibility in~\cite{CFPesop} is similar but disallows restrictions.}
\end{enumerate}

For example, consider the expressions
\begin{enumerate}
\item[]$\Set(\mu(\lambda X\,\lambda x\,(x=0 \lor \forall y\,(\NN(y)\to X(f(x,y))))))$
\item[] $\mu(\lambda X\,\lambda x\,(\Set(x=0 \lor \forall y\,(\NN(y)\to X(f(x,y))))))$
\end{enumerate}
both of which contain the functional implication $\NN(y)\to X(f(x,y))$.
The first is admissible, but the second isn't since the minimal
quasi-closed subexpression containing that functional implication is
$\lambda X\,\lambda x\,\Set(x=0\lor\forall y\,(\NN(y)\to X(f(x,y))))$
which does contain $\Set$.

Further examples of admissible expressions are the predicate
$\ConSD$ from Example~\ref{example-ConsSD}, 
$\lambda x\,\Set((x\leq 0 \lor x\geq 0) \lor |x| \leq 1/2)$,
as well as the coinductive predicate $\myC_2$ from Section~\ref{sec-gray},
$\nu(\lambda X\,\lambda x\,(|x|\le 1 \land \Set(\exists\, d (\SD(d) \land X(2x-d)))))$. 

Clearly, an admissible Harrop formula is quasi-closed and contains neither
the concurrency operator nor restrictions.

\begin{lem}
\label{lem-minus-non-harrop}
If $P$ is an admissible $\CFP$ expression, 
then $P^-$ is non-Harrop iff $P$ is, and furthermore $\tau(P^-) = \tau(P)^-$.
%If $P$ is an admissible and non-Harrop $\CFP$ expression, 
%then $P^-$ is a non-Harrop $\IFP$ expression. 
\end{lem}
\begin{proof}
The statement that the non-Harropness of $P^-$ implies that of $P$ is obvious
and does not require $P$ to be admissible (clearly, if $P$ is Harrop, so is 
$P^-$ since the operation $\cdot^{-}$ does not change strict positivity and
only removes potential reasons for non-Harropness). 
The rest of the Lemma holds even for `quasi-admissible' expressions 
by which we mean expressions that satisfy (i) and (ii) of the definition of
admissibility. Clearly, quasi-admissible $\CFP$ expressions are closed under 
subexpressions. Hence, one can easily show,  
by induction on the length of expressions, 
that if $P$ is quasi-admissible and non-Harrop, then $P^-$ is non-Harrop.
Similarly, one shows by a length induction that for quasi-admissible 
expressions $P$ the equation $\tau(P^-) = \tau(P)^-$ holds.
By the first part of the Lemma, it suffices to consider 
restriction and concurrency.
%
%If $P$ is Harrop, then $P^-$ is Harrop, 
%hence $\tau(P^-) = \one = \one^- = \tau(P)^-$. 
%Otherwise:
$\tau((\rt{A}{B})^-) =\tau(A\to B^-)=\tau(B^-)=\tau(B)^-=\tau(\rt{A}{B})^-$
using the induction hypothesis and the fact that, by condition (ii), $A$
is Harrop.
$\tau(\Set(B)^-)=\tau(B^-)=\tau(B)^-=(\Am(\tau(B)))^-=(\tau(\Set(B)))^-$
using the induction hypothesis.
\end{proof}

We call a type \emph{admissible}
if it contains type variables 
and $\Am$ only at s.p.\ positions and every function type is
part of a closed subtype without $\Am$.
\begin{lem}
\label{lem-admissible-type}
If $P$ is an admissible $\CFP$ expression, then $\tau(P)$ is an admissible type.
\end{lem}
\begin{proof}
We show the statement by structural induction.

If $P$ is Harrop, then $\tau(P) = \one$, which is admissible.

If $P$ is quasi-closed and does not contain $\Set$,
then, clearly, $\tau(P)$ is closed and does not contain $\Am$, and is therefore
trivially admissible.

If $P$ is a predicate variable, then $\tau(P)$ is a type variable,
which is admissible.

Now, assume none of the cases above applies.

If $P$ is an implication, $A \to B$, then $B$ is non-Harrop.
But then $A$ must be Harrop, since otherwise, we would have a 
functional implication,  which is included in the second case above. 
Hence $\tau(P) = \tau(B)$, which is admissible by the induction 
hypothesis (clearly, $B$ is admissible).

If $P$ is $\rt{A}{B}$, then $\tau(P) = \tau(B)$, 
which is admissible by the induction hypothesis.

If $P$ is $\Set(B)$, then, by admissibility of $P$, $B$ is non-Harrop.
Hence $\tau(P) = \Am(\tau(B))$. By the induction hypothesis, 
$\tau(B)$ is admissible, hence $\Am(\tau(B))$ is.

If $P$ is $\mu\,(\lambda X\,.\,Q)$ or $\nu\,(\lambda X\,.\,Q)$,
then $Q$ is admissible since we are not in the second case considered above.
Therefore $\tau(Q)$ is admissible, by the induction hypothesis.
It follows that $\tau(P)$, which is $\tfix{\alpha_X}{\tau(Q)}$, is admissible.

In all other cases the induction hypothesis applies immediately.
\end{proof}
%% \begin{proof}
%% Call a $\CFP$ expression \emph{admissible in $\vec X$}, 
%% where $\vec X$ is a tuple of predicate variables, 
%% if it satisfies (i) and (ii) of the definition of admissibility 
%% and the following generalization of (iii):
%% \begin{enumerate}
%% \item[(iii')] Every occurrence of a functional implication
%% %
%% is part of a subexpression without $\Set$ and restriction, 
%% and with free predicate variables 
%% among $\vec X$.
%% \end{enumerate}
%% Similarly, call a type \emph{admissible in $\vec\alpha$}, 
%% where $\alpha$ is a tuple of type variables,
%% if it contains type variables 
%% and $\Am$ only at s.p.\ positions and every function type is
%% part of a subtype without $\Am$ and with free type variables 
%% among $\vec\alpha$.

%% One can now show easily, by structural induction, 
%% that if $P$ is admissible in $X_1,\ldots,X_n$, 
%% then $\tau(P)$ is admissible in $\alpha_{X_1},\ldots,\alpha_{X_n}$.
%% \end{proof}

%
\begin{thm}[Faithfulness]
\label{thm-faithfulness}
If $a\in D$ realizes a quasi-closed admissible 
$\CFP$ formula $A$, then all $d\in\ddata(a)$ realize $A^-$.

Formally, $\RCFP$ proves 
$\forall\vec x\,\forall a\,(\rea(A)(a) \to \ddata(a)\subseteq \rea(A^-))$
where $\vec x$ are the free object variables of $A$.
\end{thm}
\begin{proof}
The proof is accomplished through a series of definitions and claims.

%% \emph{An admissible type}
%% %
%% is a type that contains type variables 
%% and $\Am$ only at s.p.\ positions and every function type is
%% part of a closed subtype without $\Am$.

%% Clearly, if $P$ is
%% admissible, 
%% %
%% then $\tau(P)$ is
%% (easy structural induction on $P$).

We call an $\RCFP$ predicate $P$, whose last argument place 
ranges over $D$, \emph{regular}
if $P(\vec x,a)$ implies that $a$ is regular.
More precisely, ``$P$ is regular'' stands for the formula 
$\forall\vec x\forall a\,(P(\vec x,a) \to \regD(a))$.

For any $\RCFP$ predicate $P$ whose last argument place ranges over $D$
we define a regular predicate $P'$ of the same arity by
\[P'(\vec x,a) \eqdef                   a\in\regD \land 
                  \forall d\in\ddata(a)\, P(\vec x,d)\,.\]

We also define  
\[\ddata(P)(\vec x,d) \eqdef \exists a\in D\,(P(\vec x,a) \land d\in\ddata(a)).\]
Clearly, if $P$ is regular, then $P\subseteq Q'$ iff $\ddata(P)\subseteq Q$.

For the special case of a subdomain $\alpha$, considered as the unary predicate 
$\lambda a\,.\,a:\alpha$, 
$\ddata(\alpha) \equiv \bigcup\{\ddata(a)\mid a:\alpha\}$, 
and if $\alpha$ is regular, then
$\alpha\subseteq\beta'$ iff $\ddata(\alpha)\subseteq\beta$.

\paragraph{Claim 1.}
If $\alpha$ is a subdomain of $D$, 
then so is $\alpha'$.

\paragraph{Proof of Claim~1.}
Let $\alpha$ be a subdomain of $D$. 
We show that $\alpha'$ satisfies the 
characterizing properties (i) and (ii) of subdomains from Lemma~\ref{lem-subdom}.

(i) Let $a\in\alpha'$ and $a_0$ a compact approximation of $a$. We have to show that
$a_0\in\alpha'$. $a_0$ is regular, since $a$ is. Hence assume $d_0\in\ddata(a_0)$. 
Then, by 
Lemma~\ref{lem-data-compact}~(\ref{lem-data-compact-approx}), 
there exists some $d\in\ddata(a)$ such that $d_0\dle d$.
Since $a\in\alpha'$, $d\in\alpha$. Hence $d_0\in\alpha$, since, being a subdomain, 
$\alpha$ is downward closed.
Conversely, assume all compact approximations of $a$ are in $\alpha'$ 
and let $d\in\ddata(a)$.
We have to show that $d\in\alpha$. Since $\alpha$ is a subdomain, 
it suffices to show that all
compact approximations of $d$ are in $\alpha$. 
Hence let $d_0$ be a compact approximation of $d$.
Then, by 
Lemma~\ref{lem-data-compact}~(\ref{lem-data-compact-reg}), 
(by 
Lemma~\ref{lem-regular}~(\ref{lem-regular-sub}), 
$a$ is regular since all its compact approximations are), 
there is  some compact approximation $a_0$ of $a$ 
such that $d_0\in\ddata(a_0)$. Since, by assumption, $a_0\in\alpha'$, 
it follows $d_0\in\alpha$.

(ii) Let $a_0,b_0$ be compact and consistent elements of $\alpha'$. 
We have to show that $c_0 \eqdef a_0\sqcup b_0 \in \alpha'$. 
$c_0$ is regular since $a_0$ and $b_0$ are and $\regD$ is a 
subdomain of $D$, by Lemma~\ref{lem-regular}~(\ref{lem-regular-sub}).
Let $d_0\in\ddata(c_0)$. By 
Lemma~\ref{lem-data-compact}~(\ref{lem-data-compact-sup}), 
there are $u\dle u'\in\ddata(a_0)\subseteq\alpha$ and 
$v\dle v'\in\ddata(b_0)\subseteq\alpha$ 
with $d_0 = u \sqcup v$. Since $\alpha$ is a subdomain of $D$ it follows
$d_0\in\alpha$.
%
%% $d_0\in\ddata(a_0)\cup\ddata(b_0)$, 
%% in which case $d_0\in\alpha$ since $a_0,b_0\in\alpha'$,
%% or there are $a_1\in\ddata(a_0)$ and 
%% $b_1\in\ddata(b_0)$ such that $d_0 = a_1 \sqcup b_1$, in which case, 
%% since $a_1,b_1\in\alpha$ and $\alpha$ is a subdomain, $d_0\in\alpha$ as well.
%
This completes the proof of Claim~1.

For any type $\rho$ and $i=1,2$ let $\variant{\rho}{i}$ be obtained 
from $\rho$ by replacing 
each free type variable $\alpha$ by the fresh type variable $\alpha_i$.
Further, let $\rho^-$ be obtained from $\rho$ by deleting all occurrences of $\Am$.

\paragraph{Claim 2.}
For every regular
admissible
type $\rho$,
$\variant{\rho}{1}\subseteq \variant{\rho^-}{2}'$ 
under the assumption that 
$\alpha_i\subseteq\regD$ ($i=1,2$) 
and
$\alpha_1\subseteq\alpha_2'$
for every free type variable $\alpha$ of $\rho$.

\paragraph{Proof of Claim 2.}
Note that $\variant{\rho}{1} \subseteq\regD$, 
by Lemma~\ref{lem-regular}~(\ref{lem-regular-reg}).
Therefore, it suffices to show 
$\ddata(\variant{\rho}{1}) \subseteq \variant{\rho^-}{2}$,
which we do by induction on $\rho$. 

If $\rho$ is closed and does not contain $\Am$ 
(by admissibility, this includes the case that $\rho$ is a function type),
then $\variant{\rho}{1} = \variant{\rho^-}{2} = \rho$ and, 
by 
Lemma~\ref{lem-regular}~(\ref{lem-regular-data})
$\rho\subseteq\Data$.
Hence 
$\ddata(\rho)\subseteq\rho$, by 
Lemma~\ref{lem-data-nonempty}~(\ref{lem-data-nonempty-sing}),  
and we are done.

Now assume that the above case does not apply.

If $\rho$ is a type variable $\alpha$, then 
$\variant{\rho}{1} = \alpha_1 \subseteq \alpha_2' = \variant{\rho^-}{2}'$.

If $\rho = \Am(\sigma)$, then $\rho^- = \sigma^-$ and, by the i.h.,\ 
$\variant{\sigma}{1} \subseteq \variant{\sigma^-}{2}'$. 
Let $a:\variant{\rho}{1}$. Then either $a=\bot$, 
in which case $\ddata(a) = \{\bot\} \subseteq\variant{\sigma^-}{2}$,
or else $a=\Amb(b,c)$ with $b,c:\variant{\sigma}{1}$, in which case 
$\ddata(a) \subseteq \ddata(b) \cup \ddata(c) \cup\{\bot\} \subseteq \variant{\sigma^-}{2}$.

If $\rho=\tfix{\alpha}{\sigma}$, then we have, up to bound renaming, 
$\variant{\rho}{1} = \tfix{\alpha_1}{\variant{\sigma}{1}}$ and
$\variant{\rho^-}{2} = \tfix{\alpha_2}{\variant{\sigma^-}{2}}$.
Since $\variant{\rho^-}{2}$ is a subdomain and hence, by Claim~1, 
so is $\variant{\rho^-}{2}'$, we can set
\begin{enumerate}
\item[] (1) $\alpha_1 \equiv \variant{\rho^-}{2}'$ 
\end{enumerate}
and achieve our goal (of proving $\variant{\rho}{1}\subseteq \variant{\rho^-}{2}'$) 
by proving $\variant{\sigma}{1} \subseteq \alpha_1$
(to be precise, by ``we can set'' we mean that we construct the formal proof 
with the extra assumption (1)).
Setting further
\begin{enumerate}
\item[] (2) $\alpha_2 \equiv \variant{\rho^-}{2}$
\end{enumerate}
we have $\alpha_1\subseteq\alpha_2'$ and therefore, by the induction hypothesis,
($\variant{\rho^-}{2}'$ and $\variant{\rho^-}{2}$ are both regular subdomains,
and $\sigma$ is
admissible,
since otherwise we would be in the first case),
$\variant{\sigma}{1}\subseteq \variant{\sigma^-}{2}'$. 
We further have $\variant{\sigma^-}{2} \equiv \variant{\rho^-}{2}$, by (2).
Therefore,
\[ \variant{\sigma}{1}\subseteq\variant{\sigma^-}{2}' \equiv  
\variant{\rho^-}{2}' \equiv \alpha_1.\]
The cases that $\rho$ is a product or a sum are easy.
This completes the proof of Claim~2.
\medskip

\bigskip

For any $\CFP$ predicate $P$ and $i=1,2$,
let $\variant{\rea(P)}{i}$ be obtained from $\rea(P)$ 
by replacing 
each free predicate variable $\reali{X}$ by
the fresh predicate variable $\reali{X}_i$
and
each free type variable $\alpha_X$
by the fresh type variable $(\alpha_X)_i$. 

\paragraph{Claim~3.}
If $P$ is
admissible,
then $\variant{\rea(P)}{1}\subseteq \variant{\rea(P^-)}{2}'$ 
under the assumption that 
$(\alpha_X)_i\subseteq\regD$ 
($i=1,2$),
and furthermore $(\alpha_X)_1\subseteq(\alpha_X)_2'$ 
and $\reali{X}_1 \subseteq \reali{X}_2'$, 
for every free predicate variable $X$ of $P$.

Note that the Faithfulness Theorem is a special case of 
Claim~3:
If $a\in D$ realizes the  quasi-closed admissible formula $A$, 
then, by 
Claim~3, 
$\rea(A^-)'(a)$ and therefore
all $d\in\ddata(a)$ realize $A^-$.
Hence, the proof of Claim~3 completes the whole proof of the theorem.

We also note that if $\variant{\rea(P)}{1}(\vec x,a)$, 
then $a$ is regular, by Lemma~\ref{lem-regular}~(\ref{lem-regular-rea})
(the premises of that Lemma follow from the assumptions of Claim~3). 
Hence, when proving the inclusion
$\variant{\rea(P)}{1}\subseteq \variant{\rea(P^-)}{2}'$,
the regularity part in the definition of the operation $(\cdot)'$
may be ignored.

%lem-strict-reg  tau(A) is regular
%lem-regular-reg  rho regular and alpha subset Reg -> rho subset Reg  [5.6(3)]
%lem-regular-rea  X~, alpha_X subset Reg -> R(P) subset reg [5.6(4)]
%lem-realizers-typed R(A) subset tau(A)

%% \UB{Since the arguments of $\CFP$ predicates play no active 
%% role in the proof, we suppress them, i.e. $\CFP$ predicates
%% are assumed to have arity $0$.}

\paragraph{Proof of Claim~3.}
Structural induction on
admissible
expressions $P$.

If $P$ is a Harrop predicate, then, by
admissibility,
it contains neither free predicate variables,
nor $\Set$, nor restriction, and therefore 
\[\variant{\rea(P)}{1}(\vec x,a)\equiv
     \variant{\rea(P^-)}{2}(\vec x,a) \equiv 
      \reah(P)(\vec x) \land a=\Nil.\]
Hence, it suffices to show
$\forall \vec x,a\,((a =\Nil\land\reah(P)(\vec x) 
  \to \forall d\in\ddata(a)\,(d=\Nil \land\reah(P)(\vec x)))$, 
which is a triviality since $\ddata(\Nil)=\{\Nil\}$.

From now on we assume that $P$ is non-Harrop. 

If $P$ is quasi-closed and contains neither $\Set$ nor restriction 
(because $P$ is
admissible,
this includes the case that $P$ 
is a functional implication), then $\rea(P)$ is quasi-closed and $P = P^-$.
Therefore, $\variant{\rea(P)}{1}= \variant{\rea(P^-)}{2} = \rea(P)$.
By 
Lemma~\ref{lem-regular}~(\ref{lem-regular-data}) 
and 
Lemma~\ref{lem-realizers-typed}, 
$\rea(P)\subseteq\adummy{\Data}$. Hence $\ddata(\rea(P)) = \rea(P)$ 
by 
Lemma~\ref{lem-data-nonempty}~(\ref{lem-data-nonempty-sing}),
and we are done.

If $P$ is a predicate variable $X$, then we have to prove
$\reali{X}_1\subseteq\reali{X}_2'$ which holds by the assumption.

From now on we assume that none of the above cases applies.

If $P$ is 
$A \to B$, then $A$ has  no free predicate variables
and $B$ is non-Harrop.
Since the case of a functional implication is excluded, $A$ must be Harrop.
Furthermore, $(A \to B)^-$ is $A \to B^-$, and this formula is non-Harrop,
by Lemma~\ref{lem-minus-non-harrop}.
Assume $\variant{\rea(A\to B)}{1}(b)$, that is, $b:\variant{\tau(B)}{1}$ and 
$\reah(A) \to \variant{\rea(B)}{1}(b)$.
We have to show $\variant{\rea(A\to B^-)}{2}'(b)$, that is,  
$b\in\regD$ and 
$\forall d\in\ddata(b)\,\variant{\rea(A\to B^-)}{2}(d)$, that is, for all
$d\in\ddata(b)$, $d : \variant{\tau(B^-)}{2}$ and
$\reah(A) \to \variant{\rea(B^-)}{2}(d)$.
%
%% $b\in\regD$ holds since $\tau(B^-)=\tau(B)^-$ and, 
%% by Lemma~\ref{lem-admissible-type}, $\tau(B)$ is admissible 
%% and therefore, by Claim~2, 
%% $\variant{\tau(B)}{1}\subseteq\variant{\tau(B)^-}{2}'\subseteq\regD$. 
%% %
%% Let $d\in\ddata(b)$.
$d : \variant{\tau(B^-)}{2}$ holds,
since, by 
Claim~2,  
$b:\variant{\tau(B)^-}{2}'$ 
and hence $b:\variant{\tau(B^-)}{2}'$ since, 
by Lemma~\ref{lem-minus-non-harrop}, $\tau(B)^-=\tau(B^-)$ 
(to apply Claim~2 we need the assumptions $\alpha_1\subseteq\alpha_2'$).
Assume $\reah(A)$. Then $\variant{\rea(B)}{1}(b)$ and, 
by the structural induction hypothesis (clearly, $B$ is again admissible),
$\variant{\rea(B^-)}{2}(d)$.

If $P$ is $\rt{A}{B}$ where $A$ is Harrop and $B$ is non-Harrop, 
then $P^-$ is $A \to B^-$. 
Assume $\variant{\rea(P)}{1}(b)$. 
Then, clearly, $\variant{\rea(A\to B)}{1}(b)$. 
Hence, the rest of the proof is exactly as in the previous case.
%(even for the case that $B$ is Harrop since, as one easily checks, 
%the above proof for $A\to B$ is also valid if $B$ is Harrop).

If $P$ is $\Set(B)$, then 
$\variant{\rea(P^-)}{2} = \variant{\rea(B^-)}{2}$ and,  
since $B$ is again
admissible,
by the induction hypothesis,
$\variant{\rea(B)}{1} \subseteq \variant{\rea(B^-)}{2}'$.
Assume $\variant{\rea(P)}{1}(c)$ and $d\in\ddata(c)$.
Then $c = \Amb(a,b)$ with $a,b : \variant{\tau(B)}{1}$ and 
\[(a \neq\bot \land d\in\ddata(a) \land \variant{\rea(B)}{1}(a))\lor 
  (b \neq\bot \land d\in\ddata(b) \land \variant{\rea(B)}{1}(b)).\]
In either case, it follows $\variant{\rea(B^-)}{2}(d)$.

If $P$ is $\mu\,(\lambda X\,.\,Q)$, 
then $\rea(P)$ is 
$\mu(\lambda\reali{X}\,.\,\rea(Q)[\tfix{\alpha_X}{\tau(Q)}/\alpha_X])$
whereas $\rea(P^-)$ is 
$\mu(\lambda\reali{X}\,.\,\rea(Q^-)[\tfix{\alpha_X}{\tau(Q^-)}/\alpha_X])$.
Hence, setting 
(see Lemma~\ref{lem-respect}~(\ref{lem-respect-hyp}) for a justification)
\begin{itemize}
\item[] $(\alpha_X)_1 \equiv\variant{(\tfix{\alpha_X}{\tau(Q)})}{1}$
\item[] $(\alpha_X)_2 \equiv\variant{(\tfix{\alpha_X}{\tau(Q^-)})}{2}$
\end{itemize} 
we have
$\variant{\rea(P)}{1} \equiv \mu(\lambda\reali{X}_1\,.\,\variant{\rea(Q)}{1})$,
and 
$\variant{\rea(P^-)}{2} \equiv \mu(\lambda\reali{X}_2\,.\,\variant{\rea(Q^-)}{2})$. 
Therefore, we can prove the assertion of Claim~3
(which is $\variant{\rea(P)}{1} \subseteq \variant{\rea(P^-)}{2}'$)
by s.p.\ induction (second rule for $\mu$ in Table~\ref{table-proof-ifp}) i.e.,
it suffices to show $\variant{\rea(Q)}{1} \subseteq \variant{\rea(P^-)}{2}'$
under the assumption 
\begin{itemize}
\item[] $\reali{X}_1 \equiv \variant{\rea(P^-)}{2}'$.
\end{itemize}
Setting further 
\begin{itemize}
\item[] $\reali{X}_2 \equiv \variant{\rea(P^-)}{2}$
\end{itemize} 
we have $\reali{X}_1 \subseteq (\reali{X}_2)'$, trivially, and furthermore,  
by Claim~2,  
$(\alpha_X)_1 \subseteq ((\alpha_X)_2)'$ since 
$(\tfix{\alpha_X}{\tau(Q)})^-=\tfix{\alpha_X}{\tau(Q^-)}$,
by the Lemmas~\ref{lem-strict}~(\ref{lem-strict-reg}),
\ref{lem-admissible-type} 
and ~\ref{lem-minus-non-harrop}.
Therefore, by the structural induction hypothesis
($Q$ is admissible since otherwise we would be in the second case),
$\variant{\rea(Q)}{1} \subseteq \variant{\rea(Q^-)}{2}'$. 
By the closure axiom, 
$\variant{\rea(Q^-)}{2}\subseteq \reali{X}_2$. 
Hence, $\variant{\rea(Q^-)}{2}' \subseteq (\reali{X}_2)'$
by the monotonicity of the operation $\cdot'$.
But $(\reali{X}_2)' \equiv \variant{\rea(P^-)}{2}'$.

If $P$ is $\nu\,(\lambda X\,.\,Q)$, then we work with $\ddata(\cdot)$ 
instead of $\cdot'$
using the earlier mentioned fact that 
if $Y$ is regular (i.e.~$Y \subseteq \adummy{\regD}$), then
$Y \subseteq Z'$ iff $\ddata(Y) \subseteq Z$.
Since the condition $\reali{Y}_1 \subseteq\reali{Y}_2'$ implies that 
$\reali{Y}_1\subseteq\adummy{\regD}$ for all free predicate variables $Y$ of $P$,
and these variables are all s.p., 
Lemma~\ref{lem-regular}~(\ref{lem-regular-rea}) 
yields that 
$\variant{\rea(P)}{1} \subseteq\adummy{\regD}$. 
Therefore,
the assertion to be proven is equivalent to 
$\ddata(\variant{\rea(P)}{1}) \subseteq \variant{\rea(P^-)}{2}$.
$\rea(P)$ is $\nu(\lambda\reali{X}\,.\,\rea(Q)[\tfix{\alpha_X}{\tau(Q)}/\alpha_X])$
and $\rea(P^-)$ is 
$\nu(\lambda\reali{X}\,.\,\rea(Q^-)[\tfix{\alpha_X}{\tau(Q^-)}/\alpha_X])$.
This means that, setting $(\alpha_X)_i$ as before, we have that
$\variant{\rea(P)}{1} \equiv \nu(\lambda\reali{X}_1\,.\,\variant{\rea(Q)}{1})$,
and 
$\variant{\rea(P^-)}{2} \equiv \nu(\lambda\reali{X}_2\,.\,\variant{\rea(Q^-)}{2})$. 
Therefore, 
$\ddata(\variant{\rea(P)}{1}) \subseteq \variant{\rea(P^-)}{2}$ can be proven
by 
s.p.\ coinduction, that is, we show 
$\ddata(\variant{\rea(P)}{1}) \subseteq \variant{\rea(Q^-)}{2}$
under the assumption 
\begin{itemize}
\item[] $\reali{X}_2 \equiv \ddata(\variant{\rea(P)}{1})$.
\end{itemize}
Setting 
\begin{itemize}
\item[] $\reali{X}_1 \equiv \variant{\rea(P)}{1}$
\end{itemize} 
we have 
$\ddata(\reali{X}_1) \subseteq \reali{X_2}$, 
hence $\reali{X}_1 \subseteq \reali{X_2}'$, 
since $\variant{\rea(P)}{1} \subseteq\adummy{\regD}$.
Therefore, by the structural induction hypothesis,
$\variant{\rea(Q)}{1} \subseteq \variant{\rea(Q^-)}{2}'$,  i.e.\ 
$\ddata(\variant{\rea(Q)}{1}) \subseteq \variant{\rea(Q^-)}{2}$.
Finally, by the coclosure axiom, 
$\reali{X}_1 \subseteq \variant{\rea(Q)}{1}$ and therefore,
by the monotonicity of the operation $\ddata(\cdot)$, we get
$\ddata(\variant{\rea(P)}{1}) \subseteq \ddata(\variant{\rea(Q)}{1})$.

In all other cases (conjunction, disjunction, quantifiers, 
predicate application and abstraction), the 
induction hypothesis applies in a straightforward way.

\end{proof}
\bigskip

Theorems~\ref{thm-soundnessI} and~\ref{thm-faithfulness} imply:
 \begin{thm}[Soundness Theorem II]
 \label{thm-soundnessII}
 From a $\CFP$ proof of a %well-formed
 quasi-closed
admissible formula $A$
one can extract a program $M$ such that
$\vdash M :\tau(A)$ and
$\RCFP$ proves
that all $d\in\ddata(\val{M})$ realize $A^-$,
that is, $\RCFP$ proves the formula 
$\forall d\in\ddata(M)\,\ire{d}{A^-}$.
\end{thm}

Theorems~\ref{thm-soundnessII} and~\ref{thm:data}, together with 
classical soundness yield:
\begin{thm}[Program Extraction]
  \label{thm-pe}
 From a $\CFP$ proof of a %well-formed 
 quasi-closed
admissible 
formula $A$
one can extract a program $M$ such that $\vdash M:\tau(A)$ and for any 
computation 
$M =  M_0 \newprintp M_1 \newprintp  \ldots$,
the limit, $\sqcup_{i \in \NN} (M_i)_{D}$, realizes  $A^-$ in every  
model of $\IFP$.
\end{thm}

\begin{rem}
The theorems above can be generalized by dropping in the definition of 
admissibility the condition that restrictions must have Harrop premises.
The definition of $A^-$ must then be modified by replacing $\rt{A}{B}$, where
$A$ is non-Harrop, by $\neg\neg A \to B$ (instead of $A\to B$). 
Since, by the rules (Rest-antimon) and
(Rest-stab), the formulas $\rt{A}{B}$ and $\rt{\neg\neg A}{B}$ are equivalent 
and, moreover, have the same realizers, 
the proof of Theorem~\ref{thm-faithfulness} requires
only minimal changes and the Theorems~\ref{thm-soundnessII} and~\ref{thm-pe}
are unchanged.

\end{rem}

\section{Application}
\label{sec-gray}
As our main case study, we extract a concurrent conversion program between
two representations of real numbers 
in [-1, 1],
the signed digit representation and infinite Gray code.

In the following, we also write $d:p$ for $\Pair(d,p)$.

The signed digit representation is an extension of the usual binary expansion
that uses the set $\SD \eqdef \{-1, 0, 1\}$ of \emph{signed digits}. 
The following predicate $\myC(x)$ expresses coinductively
that $x$ has a signed digit representation.
\begin{eqnarray*}
  \myC(x) &\eqnu& |x|\le 1 \land \exists\, d \in \SD\, \myC(2x-d)\,,
\end{eqnarray*}
with $\SD(d) \eqdef (d = -1 \lor d = 1) \lor d = 0$.
The type of $\myC$ is $\tau(\myC) = \stream{\tri}$
where $\tri \eqdef (\one + \one) + \one$ and 
$\stream{\rho} \eqdef  \tfix{\alpha}{\rho \times \alpha}$,
and its realizability interpretation is
\begin{eqnarray*}
  \ire{p}{\myC(x)} &\eqnu& |x|\le 1\land\exists\, d\in\SD\,\exists p'\ (
  p = d: p'\ \land 
  \ire{p'}{\myC(2x-d)})\, 
\end{eqnarray*}
which expresses indeed that $p$ is a signed digit representation of $x$, that is,
$p = d_0:d_1:\ldots$ with
$d_i\in\SD$ and
$x = \sum_{i}d_i2^{-(i+1)}$.
Here, we identified the three digits $d = -1, 1, 0$ with their realizers
$\Left(\Left),    \Left(\Right), \Right$.

Infinite Gray code (\cite{Gianantonio99,Tsuiki02}) is an almost redundancy free
representation of real numbers in [-1, 1] using the partial digits 
$\{-1,1, \bot\}$. A stream 
$p =d_0:d_1:\ldots$ 
of such digits
is an infinite Gray code of $x$ iff 
$d_i = \sgb(\tent^{i}(x))$
 where $\tent$ is the tent function 
$\tent(x) = 1-|2x|$ and $\sgb$ is 
a multi-valued version of the sign function for which 
$\sgb(0)$ is any element of $\{-1, 1,\bot\}$
  (see also Example~\ref{example-ConsSD}).
Clearly,
$d_i = \bot$ 
for at most one $i$,
after which the sequence continues with $1:(-1)^\omega$.
Therefore, 
this coding has 
little 
redundancy in that the
code is uniquely determined and total except for at most one digit,
which need not be defined.
Hence, infinite Gray code 
is accessible through concurrent computation with two threads,
since of any two of its digits at least one is defined.
The coinductive predicate 
\begin{eqnarray*}
\myG(x) &\eqnu& |x|\le 1 \land \D(x) \land \myG(\tent(x))\,,
\end{eqnarray*}
where $\D$ is the predicate 
$
\D(x) \eqdef  x\neq 0 \to (x\leq 0 \lor x\geq 0)\,
$  
from Example~\ref{ex-d},
expresses that $x$ has an infinite Gray code. 
Indeed,
$\tau(\myG) = \bool^{\omega}$ and
\[
\ire{p}{\myG(x)} \eqnu |x|\le 1\land \,\exists d,p'(p = d:p' \land
  (x\neq 0 \to \ire{d}{(x \le 0 \lor x \ge 0)})\ \land 
  \ire{p'}{\myG(\tent(x))})\,.
\]
Identifying $\Left, \Right, \bot$ (the range of the variable $d$) with
$-1, 1, \bot$, one sees that the realizer $p$ is indeed a Gray code of $x$.

In \cite{IFP}, the inclusion $\myC \subseteq \myG$ was proved in IFP 
and a sequential conversion
function from signed digit representation to infinite Gray code extracted.
On the other hand, a program producing a signed digit representation from an
infinite Gray code cannot access its input sequentially 
from left to right since it will diverge when it accesses $\bot$.
Therefore, the program needs to evaluate 
two consecutive digits concurrently to obtain at least one
of them. 
With this idea in mind, we define
a concurrent version of $\myC$ as
\begin{eqnarray*}
\myC_2(x) &\eqnu& |x|\le 1 \land \Set(\exists\, d \in \SD\, \myC_2(2x-d)),
\end{eqnarray*}
so that $\tau(\myC_2) = \tfix{\alpha}{\Am(\tri \times \alpha)}$, and
prove $\myG \subseteq \myC_2$ in CFP (Theorem~\ref{thm-g-mc}).
Then we can extract from the proof a concurrent algorithm that 
converts infinite Gray code to signed digit representation.  
Note that, while the formula $\myG \subseteq \myC_2$ is \emph{not} admissible,
the formula $\myC_2(x)$ \emph{is}.
Therefore, if for some real number $x$ we can prove $\myG(x)$,
the proof of $\myG \subseteq \myC_2$ will give us a proof of $\myC_2(x)$ to which
Theorem~\ref{thm-pe} applies. Since $\myC_2(x)^-$ is $\myC(x)$, this means that 
we have a nondeterministic program all of whose fair computation paths 
will result in a (deterministic) signed digit representation of $x$.

Now we carry out the proof of $\myG \subseteq \myC_2$.
For simplicity, we use 
pattern matching on constructor expressions for defining functions.  
For example, we write
${\mathsf f}\ (a:t) \eqdef M$
for
${\mathsf f} \eqdef \lambda c.\, \caseof{c}\,\{\Pair(a, t) \to M \}.$

The crucial step in the proof is accomplished by
Example~\ref{example-ConsSD}, since it yields nondeterministic
information about the first digit of the signed digit representation
of $x$, as expressed by the predicate
\[
\ConSD(x)  \eqdef  \Set((x\leq 0 \lor x\geq 0) \lor |x| \leq 1/2).
\]

\begin{lem}
\label{lem-gsd}
$\myG \subseteq \ConSD$.
\end{lem}
\begin{proof}
$\myG(x)$ implies $\D(x)$ and $\D(\tent(x))$, 
and hence $\ConSD$, by Example~\ref{example-ConsSD}.
\end{proof}
The extracted program $\gscomp:\ftyp{\stream{\bool}}{\Am(\tri)}$  
uses the program $\conSD$ defined in Example~\ref{example-ConsSD}:
$$
\gscomp\ (a:b:p)  \eqdef \conSD\  (\Pair(a, b))\,.
$$

We also need the following closure properties of $\myG$:
\begin{lem}
\label{lem-gclosure}
Assume $\myG(x)$. Then:
\begin{itemize}
\item[(1)] $\myG(\tent(x))$, $\myG(|x|)$, and $\myG(-x)$;
\item[(2)] if $x \ge 0$, then $\myG(2x-1)$ and $\myG(1-x)$;
\item[(3)] if $|x|\le 1/2$, then $\myG(2x)$.
\end{itemize}
\end{lem}
\begin{proof}
  This follows directly from the definition of $\myG$ and 
  elementary properties
  of the tent function (recall $\tent(x)=1-|2x|$).
  The extracted programs consist of simple manipulations of 
  the given digit stream realizing $\myG(x)$, concerning only its tail and  
  first two digits. No nondeterminism is involved.

  We only use the fixed point property of $\myG$, namely that
  for all $y$,  $\myG(y)$ is equivalent to 
  $|y|\le 1 \land \D(y) \land \myG(\tent(y))$.
  
  This equivalence has computational content which will show up in the programs
  extracted from the proofs below: If it is used from left to right, a 
  stream (the realizer of $\myG(y)$) is split into its head (realizer of $\D(y)$)
  and tail (realizer of $\myG(\tent(y))$). Using it
  from right to left corresponds to the converse operation of adding 
  a digit to a stream.
  
  Proofs of (1-3):
  
  (1) follows directly from the equivalence above and the fact that
  $\tent(x)=\tent(-x)$. 
  The three extracted programs are ${\sf f}_1(d:p) = p$, 
  ${\sf f}_2(d:p) = 1:p$, and 
  ${\sf f}_3(p) = \nh\ p$ where $\nh\ (d:p) = \nnot\ d: p$.
  
  (2) Assume in addition $x\ge 0$. Then $\tent(x) = 1 - 2x$. 
  Since $\myG(\tent(x))$ and, by (1), $\myG$ is closed under negation, we have
  $\myG(2x-1)$. 
  Furthermore, since $0 \le 1-x\le 1$, we have $|1-x|\le 1$ and $\D(1-x)$.
  Therefore, to establish $\myG(1-x)$, it suffices to show $\myG(\tent(1-x))$.
  But $\tent(1-x) = 1-2(1-x)= 2x-1$ and we have shown $\myG(2x-1)$ already.
  The extracted programs are ${\sf f}_4(d:p) = \nh\ p$ and 
  ${\sf f}_5(d:p) = 1:\nh\ p$.

  (3) Now assume $|x|\le 1/2$. Then $1-|2x| \ge 0$ and we have 
  $\myG(1-|2x|)$ (since $\myG(x)$). Therefore, by (2), $\myG(|2x|)$. 
  Hence $\myG(\tent(|2x|))$ and therefore also $\myG(\tent(2x))$ since
  $\tent(2x)=\tent(|2x|)$. Since $|x|\le 1/2$ implies $|2x|\le 1$ and 
  $\myG(x)$ implies $\D(x)$ and hence $\D(2x)$, it follows $\myG(2x)$.
  The extracted program is ${\sf f}_6(d:e:q) = d:\nh\ q$.
\end{proof} 
\begin{thm}
\label{thm-g-mc}
$\myG\subseteq\myC_2$.
\end{thm}
\begin{proof}
By coinduction. Setting 
$A(x) \eqdef \exists d \in \SD\, \myG(2x-d)$, 
we have to show 
\begin{equation}
\label{eq-step}
\myG(x) \to |x| \leq 1 \land \Set(A(x))\,.
\end{equation}
Assume $\myG(x)$. 
Then $\ConSD(x)$, by Lemma \ref{lem-gsd}.
Therefore, it suffices to show 
\begin{equation}
\label{eq-consd}
\ConSD(x)  \to \Set(A(x))\,
\end{equation} 
which, with the help of the rule (Conc-mp), can be reduced to
\begin{equation}
\label{eq-onedigit}
(x \leq 0 \lor x \geq 0 \lor |x| \leq 1/2)   \to A(x).
\end{equation}
(\ref{eq-onedigit}) can be easily shown 
using Lemma~\ref{lem-gclosure}:
If $x\le 0$, then $\tent(x) = 2x+1$. Since $\myG(\tent(x))$, we have 
$\myG(2x-d)$ for $d=-1$.
If $x \ge 0$, then $\myG(2x-d)$ for $d=1$ by (2).
If $|x|\le 1/2$, then $\myG(2x-d)$ for $d=0$ by (3).
\end{proof}

The program 
$\onedigit : \ftyp{\stream{\bool}}{\ftyp{\tri}{\tri\times\stream{\bool}}}$
extracted from the proof of (\ref{eq-onedigit}) 
from the assumption $\myG(x)$ is
\begin{align*}
\onedigit\ (a:b:p)\ c \eqdef 
&  \caseof{c} \{\\
& \hspace{1em} \Left(d) \to 
                       \caseof{d} \{\\
&\hspace{7em} \Left(\_) \to \Pair(-1, b:p); \\
&\hspace{7em}               \Right(\_) \to \Pair(1,(\nnot\ b):p) \};\\
&\hspace{1em}\Right(\_) \to    \Pair(0, a: (\nh\ p))\}\\
\nnot\ a \eqdef& \caseof{a} \{\Left(\_) \to \Right; \\
&\hspace*{4.7em}\Right(\_) \to \Left\}\\
\nh\ (a: p) \eqdef& (\nnot\ a): p
\end{align*}
This is lifted to a proof of (\ref{eq-consd}) using $\mapamb$ (the realizer of
(Conc-mp)). Hence the extracted realizer 
${\mathsf s} : \ftyp{\stream{\bool}}{\Am(\tri \times \stream{\bool})}$
of~(\ref{eq-step}) is 
\[
{\mathsf s}\ p \eqdef  \mapamb\ (\onedigit\ p)\  (\gscomp\ p)
\]
The main program extracted  from the proof of  Theorem~\ref{thm-g-mc} is
obtained from the step function $\mathsf{s}$ by a special form of recursion, 
commonly known as \emph{coiteration}. Formally, we use the realizer of the
coinduction rule $\COIND(\Phi_{\myC_2},\myG)$ where $\Phi_{\myC_2}$ is the operator 
used to define $\myC_2$ as largest fixed point, i.e.\ 
\begin{align*}
\Phi_{\myC_2} &\eqdef \lambda X\,\lambda x\, |x| \leq 1 \land \Set(\exists d \in \SD\, 
       X(2x-d)).
\end{align*}
The realizer of coinduction (whose correctness is shown in~\cite{IFP}) 
also uses a program 
$\mon : \ftyp{(\ftyp{\alpha_X}{\alpha_Y})}{\Am(\ftyp{\tri \times \alpha_X)}{\Am(\tri \times \alpha_Y)}}$ 
extracted from the canonical proof of 
the monotonicity of $\Phi_{\myC_2}$:
\begin{align*}
&\mon\ f\ p = \mapamb\ (\mon'\ f)\  p\\ 
&\qquad \hbox{where}\quad
\mon'\ f\ (a:t) = a : f\ t
\end{align*}
Putting everything together, we obtain the 
\emph{infinite Gray code to signed digit representation conversion program} 
\[\gtos : \ftyp{\stream{\bool}}{\tfix{\alpha}{\Am(\tri \times \alpha)}}, 
\quad\gtos \eqrec (\mon\ \gtos) \circ \mathsf{s}\]
\begin{picture}(80,75)(-80,5)                                      
\put(20,60){$\stream{\bool}$}
\put(65,61){$\Longrightarrow$}
\put(64,70){$\gtos$}
\put(120,60){$\tfix{\alpha}{\Am(\tri \times \alpha)}$}
\put(24,40){$\Big\Downarrow$}
\put(15,40){$\mathsf{s}$}
\put(130,40){$\rotatebox{90}{=}$ %
}
\put(0,13){$\Am(\tri \times \stream{\bool})$}
\put(58,25){$\mon\ \gtos$}
\put(65,13){$\Longrightarrow$}
\put(110,13){$\Am(\tri \times (\tfix{\alpha}{\Am(\tri \times \alpha)}))$}
\end{picture} 

Using the equational theory of RCFP,
one can simplify $\gtos$ to %the following program:
\begin{align*}
  \gtos\ &(a:b:t) =\Amb(\\
  &\ \ \caseof{a}      \{\Left(\_) \to  -1: \gtos\ (b:t);\\
  &\ \ \hspace*{1.8cm  } \Right(\_) \to  1: \gtos((\nnot\ b):t)\},\\
  &\ \ \caseof{b}      \{\Left(\_) \to  \bot;\\
  &\ \ \hspace*{1.8cm}   \Right(\_) \to  0: \gtos(a:(\nh\ t))\}).
  \end{align*}
The soundness of the $\RCFP$ axioms with respect to the denotational 
semantics and the adequacy property of our language guarantee that the 
two programs are equivalent.

In \cite{Tsuiki05}, a Gray-code to signed digit conversion program was written
with the locally angelic $\Amb$ operator
that evaluates the first two cells $a$ and $b$  
in parallel and continues the computation based
on the value obtained first.  
In that program,  if the value of $b$ is first obtained and 
it is $\Left$, then $a$ has to be evaluated again.   
With globally angelic choice, as the above program shows, 
one can simply discard that value and use the value of the other thread.  
Globally angelic choice also has the possibility to speed up the 
computation  if the two threads of $\Amb$ are computed in parallel and 
the {\em whole} computation 
terminates first if it uses the value of the thread 
that terminates second.

%% Globally angelic choice also has the possibility to speed up the 
%% computation  if the two threads of $\Amb$ are computed in parallel and 
%% the {\em whole} computation based on the \UB{value obtained second}
%% terminates first.

\section{Implementation}
\label{sec-experiments}
We describe the Haskell implementation of extracted programs in general,
and the particular extracted program for infinite Gray code conversion in particular.

Since our programming language can be viewed as a fragment of Haskell,  
we can execute the extracted program in Haskell by
implementing the Amb operator with the Haskell concurrency module.
We comment on the essential points of the implementation.
The full code is listed in the appendix.

First, we define the domain $D$ as a Haskell data type: %
\begin{verbatim}
data D = Nil | Le D | Ri D | Pair(D, D) | Fun(D -> D) | Amb(D, D)
\end{verbatim}
The $\ssp$-reduction, which preserves the Phase I denotational semantics and reduces a program to a w.h.n.f. 
with the leftmost outermost reduction strategy,  coincides with reduction in Haskell. 
Thus, we can identify extracted programs with programs of type D that compute 
that phase.
The $\newprintc$ reduction that concurrently calculates the arguments of $\Amb$
can be implemented  with the Haskell concurrency module.   
In \cite{BoisPLT02}, the (locally angelic) amb operator 
was implemented in Glasgow Distributed Haskell (GDH).
Here, we implemented it 
with  the Haskell libraries \verb|Control.Concurrent| and \verb|Control.Exception|
as 
a simple function 
\verb|ambL :: [b] -> IO b|
that concurrently evaluates the elements of a list and writes the 
result first obtained in a mutable variable.

Finally, the function 
\verb|ed :: D -> IO D|  produces an element of  $\ddata(a)$  from $a \in D$
by activating \verb|ambL| for the case of $\Amb(a, b)$.
It corresponds to $\newprintp$-reduction though it computes arguments of 
a pair 
sequentially.
This function is nondeterministic since the result of executing 
\verb|ed (Amb (a,b))|
depends on which of the arguments \verb|a,b| delivers a result first.
The set of all possible results of \verb|ed a| corresponds to the set $\ddata(a)$.

We executed the program extracted in Section \ref{sec-gray} with
\verb|ed|. 
As we have noted, the number $0$ has three Gray-codes
(i.e., realizers of $\myG(0)$):  
$a = \bot\!:\!1\!:\!(-1)^\omega$,
$b = 1\!:\!1\!:\!(-1)^\omega$, and 
$c = -1\!:\!1\!:\!(-1)^\omega$.  On the other hand, 
the set of signed digit representations of $0$ is
$A \cup B \cup C$ where
$A = \{0^\omega\}$,
$B = \{0^k\!:\!1\!:\!(-1)^\omega \mid k \geq 0\}$, 
and
$C = \{0^k\!:\!(-1)\!:\!1^\omega \mid k \geq 0\}$, 
i.e.,
$A \cup B \cup C$ is the set of realizers of $\myC(0)$.
One can calculate
$$\gtos (a) = \Amb(\bot, 0\!:\! \Amb(\bot,  0\!:\! \ldots))$$
and $\ddata(\gtos (a)) = A$. Thus 
$\gtos (a)$ is reduced uniquely to $0 \!:\! 0 \!:\! \ldots$
by the operational semantics.
On the other hand,  one can calculate $\ddata(\gtos (b)) = A \cup B$ and
$\ddata(\gtos (c)) = A \cup C$.  They are subsets of 
the set of realizers of $\myC(0)$ as Theorem \ref{thm-soundnessII} says, and
$\gtos (b)$ is reduced to an element of $A \cup B$ 
as Theorem \ref{thm-pe} says\footnote{To be precise, 
Theorems~\ref{thm-soundnessII} and~\ref{thm-pe} are applied to proofs of 
$\myC_2(0)$ obtained by combining the proof of $\myG\subseteq\myC_2$~(Theorem~\ref{thm-g-mc}) with different proofs of $\myG(0)$ that yield the realizers $b$ and $c$.}.

We wrote a program that produces a $\{-1, 1, \bot\}$-sequence 
with the option to control the speed of computation of each digit ($-1$ and $1$). 
Applying this to \texttt{gtos} and then to \texttt{ed} yields the expected results.

\section{Conclusion}
\label{sec-conclusion}

We introduced the logical system $\CFP$ by extending $\IFP$ \cite{IFP}
with two propositional operators $\rt{A}{B}$ and $\Set(A)$,
and developed a method for
extracting nondeterministic and concurrent 
programs that are provably total and  satisfy
their specifications. 

While $\IFP$ already imports classical logic through nc-axioms that
need only be true classically, in $\CFP$ the access to classical logic
is considerably widened through the rule (Conc-lem) which,
when interpreting $\rt{A}{B}$ as $A \to B$ and identifying $\Set(A)$ with $A$,
is constructively 
invalid but has nontrivial 
nondeterministic
computational content.

We applied our system to extract a 
concurrent
translation from infinite Gray 
code to the signed digit representation, thus demonstrating that this approach
is not merely about program extraction `in principle', but can be used to
solve nontrivial 
concurrent
computation problems through 
program extraction.

After an overview of related work, 
we conclude with a brief discussion of further applications of our approach.

\subsection{Related work}
\label{sub-related}
The CSL 2016 paper \cite{BergerCSL16} 
is an early attempt to capture the concurrent execution of partial programs
via program extraction and can be seen as the starting point of our work.
Our main advance, compared to that paper, is the specification of 
globally angelic choice with the help of the new logical connective $B|_A$ 
which allows us to express bounded nondeterminism with complete
control of the number of threads.
In contrast, \cite{BergerCSL16} 
modelled nondeterminism with countably infinite branching,
which is either unsuitable or overkill for most applications. 
Furthermore, our approach has a typing discipline,
and a sound and complete small-step reduction that permits switching 
between global and local nondeterminism (see Section~\ref{sub-local} below).

As for the study of angelic nondeterminism, it is not easy to develop 
a denotational semantics as we noted in Section \ref{sec-ang}, 
and it has been mainly studied from the operational point of view,  
for example through notions of equivalence or refinement of processes
and associated proof methods,
which are all fundamental for correctness and termination
\cite{LassenMoran99,MoranSandsCarlsson2003,Lassen2006,sabel_schmidt-schauss_2008,CarayolHirschkoffSangiori2005,Levy07}.
Regarding imperative languages, Hoare logic and its extensions 
have been applied to nondeterminism and proving totality from the very 
beginning (\cite{Apt2019FiftyYO} is a good survey on this subject).
\cite{Mamouras15} studies angelic nondeterminism with 
an extension of Hoare Logic.
In \cite{Geoffroy18} it is shown that nondeterminism can be modelled 
within Krivine's classical realizability~\cite{Krivine03}.

There are many logical approaches to concurrency.
For example, Reynolds' separation logic~\cite{Reynolds:2002} has been extended  
to the concurrent and higher-order setting~\cite{OHearn07,Brookes07,Jungetal18}, 
and there are logics for session types
and process calculi ~\cite{Wadler14,CairesPfenningToninho16,Kouzapasetal16}
as well as logics based on Avron's hypersequent calculus \cite{AscieriCiabattoniGenco20}, 
which are oriented more towards the 
formulae-as-types/proofs-as-programs~\cite{Howard80,Wadler14a} or 
proofs-as-processes paradigm~\cite{Abramsky94}.
Another line of research are process 
algebras~\cite{Milner80,Hoare85,MilnerParrowWalker92} with associated temporal and modal 
logics~\cite{Stirling91},
as well as axiomatic approaches to process algebra~\cite{BergstraKlop89}.
In~\cite{SchmidtSchau2020}
a translation of Milner's $\pi$-calculus in a formal model of concurrent Haskell
is provided and proven to be correct. Based on this, we hope to give a similar
correctness proof of the current implementation of our calculus in concurrent Haskell.

The work cited above introduces highly specialized logics and expressive languages that model and reason about concurrent programs with complex communication patterns. While we do not model communication, our approach focuses on a form of choice that can be seen as the formal representation of a race condition, where two threads of a process attempt to modify the same memory cell simultaneously (in our case, the same cell of a stream of digits representing a real number). Race conditions generally lead to unpredictable behaviour, so they are typically \emph{avoided} through proper memory access management. In contrast, our work provides a logical framework for the \emph{safe use} of racy programs, an area that has not been fully explored. Additionally, while the cited work is concerned with formal systems for \emph{processes}, our work focuses on the formalization of \emph{abstract mathematics} (e.g., the axiomatic description of real numbers), which connects to concurrency through realizability and program extraction. The result is programs (with a limited form of concurrency) that compute provably correct witnesses for statements in abstract mathematics.

\subsection{Further applications in computable analysis}
\label{sub-analysis}
Our formalization of infinite Gray code in $\CFP$ has been extended 
from real numbers to compact sets of reals~\cite{BergerSpreen23},
building upon a general theory of digital representation 
in compact metric spaces~\cite{BergerSpreen16,Spreen20}.
In~\cite{BergerSpreen23} further rules for $\rt{A}{B}$ and $\Set(B)$
are introduced whose realizers are expressed using a generalized case-construct
with overlapping clauses.

Another application in computable analysis is
Gaussian elimination which involves the task of finding
a non-zero entry in a non-singular matrix. 
As shown in~\cite{BergerSeisenbergerSpreenTsuiki22}, 
our approach makes it possible to search for such `pivot elements'
in a concurrent way.

\subsection{Modelling locally angelic choice} 
\label{sub-local}
We remarked earlier that our interpretation of $\Amb$ corresponds to
\emph{globally} angelic choice. Surprisingly, \emph{locally} angelic choice
can be modelled by a slight modification of the restriction and 
the total concurrency
operators: Simply replace $A$ by the logically equivalent formula
$A \lor \False$, more precisely, set
$\rtp{A}{B} \eqdef \rt{A}{(B\lor\False)}$ and
$\Set'(A) \eqdef \Set(A\lor\False)$.
Then the proof rules in 
Section~\ref{sec-cfp}
with $\rt{}{}$ and $\Set$ replaced by $\rtp{}{}$ and
$\Set'$, respectively but without the 
strictness
condition, are theorems of $\CFP$. 
To see that the operator $\Set'$ indeed corresponds to locally angelic choice
it is best to compare the realizers of the rule (Conc-mp) for $\Set$ and $\Set'$.
Assume $A$, $B$ are non-Harrop and $f$ is a realizer of $A \to B$.
Then, 
if $\Amb(a,b)$ realizes $\Set(A)$, 
then $\Amb(\strictapp{f}{a}, \strictapp{f}{b})$ realizes $\Set(B)$.
This means that to choose, say, the left argument of $\Amb$ as a result,
$a$ must terminate and so must the ambient (global) computation 
$\strictapp{f}{a}$.
On the other hand, the program extracted from the proof of (Conc-mp) for
$\Set'$ takes a realizer $\Amb(a,b)$ of $\Set'(A)$ and returns 
$\Amb(\strictapp{(\aup \circ f \circ \adown)}{a}, \strictapp{(\aup \circ f \circ \adown)}{b})$
as realizer of $\Set'(B)$, where $\aup$ and $\adown$ are the realizers of $B \to (B\lor\False)$ and 
$(A \lor\False) \to A$, 
namely,
$\aup \eqdef \lambda a.\, \Left(a)$ and
$\adown \eqdef \lambda c.\, \caseof{c}\{\Left(a) \to a\}$.
Now, to choose the left argument of $\Amb$, 
it is enough for $a$ to terminate since the non-strict operation $\aup$
will immediately produce a w.h.n.f. without invoking the ambient computation.
By redefining realizers of $\rt{A}{B}$ and $\Set(A)$ as realizers of 
$\rtp{A}{B}$ and $\Set'(A)$ and 
the realizers of the rules of $\CFP$ as those extracted from the proofs 
of the corresponding rules for $\rtp{}{}$ and $\Set'$,
we have another realizability interpretation of CFP that models
locally angelic choice.  

\subsection{Markov's principle with restriction}
\label{sub-markov}
So far, \restintro{}\ is the only rule that derives a restriction in a 
non-trivial way. 
However, there are other such rules,
for example
\begin{center}
\AxiomC{$\forall x \in \NN (P(x) \lor \neg P(x))$}
\RightLabel{Rest-Markov}
\UnaryInfC{$\rt{\exists x \in\NN\,P(x)}{\exists x \in\NN\,P(x)}$}
            \DisplayProof \ \ \ \ 
          \end{center}
If $P(x)$ is Harrop, then (Rest-Markov) 
is realized by minimization.
More precisely, if $f $ realizes $\forall x \in \NN (P(x) \lor \neg P(x))$,
then $\min(f)$ realizes the formula  
$\rt{\exists x \in\NN\, P(x)}{\exists x \in\NN\,P(x)}$,
where $\min(f)$ computes the least $k \in \NN$ such that $f\, k = \Left$
if such $k$ exists, and does not terminate, otherwise. 
One might expect as conclusion of (Rest-Markov) the formula
$\rt{(\neg\neg\exists x \in\NN\,P(x))}{\exists x \in\NN\,P(x)}$.  
However, because of (Rest-stab) (which is realized by the identity), 
this wouldn't make a difference.
The rule (Rest-Markov) can be used, for example, to prove that
Harrop predicates that are recursively enumerable (re) and 
have re complements are decidable. 
From the proof one can extract a program 
that 
concurrently searches for evidence of membership in the predicate and
its complement.

\subsubsection*{Acknowledgements}
This work was supported by IRSES Nr.~612638 CORCON and Nr.~294962 COMPUTAL of
  the European Commission, the JSPS Core-to-Core Program, A. Advanced
  research Networks and JSPS KAKENHI  15K00015 and 23H03346 as well as 
  the Marie Curie RISE project CID (H2020-MSCA-RISE-2016-731143).
\bigskip

We thank the three anonymous reviewers for their thorough reviews and constructive
criticism of our work. They uncovered numerous errors, and made many excellent suggestions, 
a large number of which we have incorporated. The paper has improved substantially as a result of their work.

\bibliographystyle{alphaurl}
\bibliography{../../bibandmac/refs}

\vfill

\newpage

\section*{Appendix}

\appendix

\section{Implementation}
\label{Sec:appendix-program}
We explain the program and experiments of Section~7 in more detail.  The source code (GraySD.hs) is available from the archive \cite{githubUB}.

\subsection{Nondeterminism}
\label{sub-program-nondet}

Using the primitives of the Haskell libraries
\verb|Concurrent| and \verb|Exception|
we can implement nondeterministic choice through a program
\verb|ambL| that picks from a list nondeterministically
a terminating element (if exists). Although in our application 
we need only binary choice, we implement arbitrary finite choice
since it is technically more convenient and permits more applications,
e.g.\ Gaussian elimination (Section~\ref{sub-analysis}).

\begin{verbatim}
import Control.Concurrent
import Control.Exception

ambL :: [a] -> IO a
ambL xs = 
  do { m <- newEmptyMVar ; 
       acts <- sequence 
                 [ forkIO (do { y <- evaluate x ; putMVar m y }) 
                    | x <- xs ] ;
       z <- takeMVar m ;        
       x <- sequence_ (map killThread acts) ;
       seq x (return z)
     }
\end{verbatim}
Comments:

\begin{itemize}
\item \verb|newEmptyMVar| creates an empty mutable variable, 
\item \verb|forkIO| creates a thread,
\item \verb|evaluate| evaluates its argument to weak head normal form,  
\item \verb|putMVar m y| writes \verb|y| into the mutable variable \verb|m| 
       provided \verb|m| is empty,
\item the line \verb|seq x (return z)| makes sure that the threads are killed 
      before the final result \verb|z| is returned. 
\end{itemize}

\subsection{Extracting data}
\label{sub-program-data}

We define the domain $D$ (Section~2) and a program \verb|ed| on $D$
(`extract data') that, using \verb|ambL|,  nondeterministically selects 
a terminating argument of the constructor Amb.

\begin{verbatim}
data D = Nil | Le D | Ri D | Pair (D,D) | Fun(D -> D) | Amb (D,D)

ed :: D -> IO D
ed (Le d) = do { d' <- ed d ; return (Le d') }
ed (Ri d) = do { d' <- ed d ; return (Ri d') }
ed (Pair (d,e)) = do { d' <- ed d ; e' <- ed e ; return (Pair (d',e')) }
ed (Amb (a,b)) = do { c <- ambL [a,b] ; ed c } 
ed d = return d
\end{verbatim}

\verb|ed| can be seen as an implementation of the operational semantics 
in Section~3.

\subsection{Gray code to Signed Digit Representation conversion}
\label{sub-program-gray}

We read off the programs extracted in  Sections~5 and~6
to obtain the desired conversion function.
Note that this is nothing but a copy of the programs in those sections
  with type annotations for readability. The programs work without
  type annotation because Haskell infers their types.
The Haskell types contain only one type $D$.
Their types as CFP programs are shown as comments in the code below. 

\paragraph{From Section~5.}

\begin{verbatim}
mapamb :: (D -> D)  -> D -> D  -- (B -> C) -> A(B) -> A(C)
         -- (A(B) is the type of Amb(a,b) where a,b are of type B)

leftright :: D -> D   -- B + C -> 1 + 1
leftright = \b ->  case b of {Le _ -> Le Nil; Ri _ -> Ri Nil}

mapamb = \f -> \c -> case c of {Amb(a,b) -> Amb(f $! a, f $! b)}

conSD :: D -> D    -- 2 x 2 -> A(3) 
         -- (2 = 1+1, etc. where 1 is the unit type)
conSD = \c -> case c of {Pair(a, b) ->
      Amb(Le $! (leftright a), 
          Ri $! (case b of {Le _ -> bot; Ri _ -> Nil}))}
bot = bot
\end{verbatim}

\paragraph{From Section~6.}

\begin{verbatim}
gscomp :: D -> D  -- [2] -> A(3) where [2] = fix a . 2 x a
gscomp (Pair(a, Pair(b, p))) = conSD (Pair(a, b))

onedigit :: D -> D -> D   -- [2] -> 3 -> 3 x [2]
onedigit  (Pair(a, Pair (b, p))) c = case c of {
       Le d -> case d of {
              Le _ -> Pair(Le(Le Nil), Pair(b,p));
              Ri _  -> Pair(Le(Ri Nil), Pair(notD b,p))
        };
        Ri _  -> Pair(Ri Nil, Pair(a, nhD p))}

notD :: D -> D   -- 2 -> 2
notD a = case a of {Le _ -> Ri Nil; Ri _ -> Le Nil}

nhD :: D -> D  -- [2] -> [2]
nhD (Pair (a, p)) = Pair (notD a, p)

s :: D -> D    -- [2] -> A(3 x [2])
s p = mapamb (onedigit p) (gscomp p)

mon :: (D -> D) -> D -> D   -- (B -> C) -> A(3 x B) -> A(3 x C)
mon f p = mapamb (mond f) p  
   where  mond f (Pair(a,t)) = Pair(a, f t)

gtos :: D -> D    -- [2]  -> fix a . A(3 x a) 
gtos = (mon gtos) . s
\end{verbatim}

\subsection{Gray code generation with delayed digits}
\label{sub-program-delay}

Recall that Gray code has the digits $1$ and $ -1$, modelled
as \verb|Ri Nil| and \verb|Le Nil|. A digit may as well be undefined ($\bot$) 
in which case it is modelled by a nonterminating computation (such as \verb|bot|).

To exhibit the nondeterminism in our programs we generate digits with different
computation times. For example, \verb|graydigitToD 5| denotes the digit $1$ computed in
$500000$ steps, while \verb|graydigitToD 0| does not terminate and therefore 
denotes $\bot$.

\begin{verbatim}
delay :: Integer -> D
delay n  | n > 1     = delay (n-1)
         | n == 1    = Ri Nil
         | n == 0    = bot
         | n == (-1) = Le Nil
         | n < (-1)  = delay (n+1)

graydigitToD :: Integer -> D  
graydigitToD a | a == (-1) = Le Nil
               | a == 1    = Ri Nil
               | True      = delay (a*100000)
\end{verbatim}

The function \verb|grayToD| lifts this to Gray codes, that is, 
infinite sequences of partial Gray digits 
represented as elements of $D$:

\begin{verbatim}
-- list to Pairs
ltop :: [D] -> D
ltop = foldr (\x -> \y -> Pair(x,y)) Nil

grayToD :: [Integer] -> D
grayToD = ltop . (map graydigitToD)
\end{verbatim}

For example, \verb|grayToD (0:5:(-3):[-1,-1..])| denotes the Gray code
$\bot:1:-1:-1,-1,\ldots$ where the first digit does not terminate, the second digit (1)
takes 500000 steps to compute and the third digit (-1) takes 300000 steps.
The remaining digits (all $-1$) take one step each.

\subsection{Truncating the output and printing the result}
\label{sub-io}

The program $\gtos$ transforms Gray code into signed digit representation, 
so both input and output are infinite. 
To observe the computation, 
we truncate the output to some finite approximation.
%% we truncate the input to some finite approximation which $\gtos$ will
%% map to some finite approximation of the output. 
%
This finite output
is a nondeterministic element of $D$ 
(i.e.~it may contain the constructor \verb|Amb|) 
from which we then can extract nondeterministically a deterministic 
data using the function \verb|ed| which can be printed.

In the following we define the truncation and the printing 
of deterministic finite data.

\paragraph{Truncating $d\in D$ at depth $n$.}

\begin{verbatim}
takeD :: Int -> D -> D
takeD n d | n > 0 = 
  case d of
    {
      Nil        -> Nil ;
      Le a       -> Le (takeD (n-1) a) ;
      Ri a       -> Ri (takeD (n-1) a) ;
      Pair(a, b) -> Pair (takeD (n-1) a, takeD (n-1) b) ;
      Amb(a,b)   -> Amb(takeD (n-1) a, takeD (n-1) b) ;
      Fun _      -> error "takeD _ (Fun _)" ;
    }
            | otherwise = Nil
\end{verbatim}

\paragraph{Showing a partial signed digit.}

\begin{verbatim}
dtosd :: D -> String
dtosd (Le (Ri Nil)) = " 1"
dtosd (Le (Le Nil)) = "-1"
dtosd (Ri Nil)      = " 0"
dtosd _             = " bot"
\end{verbatim}

\paragraph{Printing an element of $D$ that represents a finite deterministic 
signed digit stream.}

\begin{verbatim}
prints :: D -> IO ()
prints (Pair (d,e)) = putStr (dtosd d) >> prints e
prints Nil          = putStrLn ""
prints _ = error "prints: not a partial signed digit stream"
\end{verbatim}

\subsection{Experiments}
\label{sub-experiments}

As explained in Section~7,
there are three Gray codes of $0$:
\begin{eqnarray*}
a &=& \ \ \bot:1:-1,-1,-1,\ldots\\
b &=& \ \ \, 1:1:-1,-1,-1,\ldots\\
c &=& -1:1:-1,-1,-1,\ldots
\end{eqnarray*}

and the set of signed digit representations of $0$ is
$A \cup B \cup C$ where

\begin{eqnarray*}
A &=& \{0^\omega\}\\
B &=& \{0^k\!:\!1\!:\!(-1)^\omega \mid k \geq 0\}\\
C &=& \{0^k\!:\!(-1)\!:\!1^\omega \mid k \geq 0\}. 
\end{eqnarray*}

Our \verb|gtos| program 
nondeterministically produces 
an element of $A$ for input $a$,
an element of $A \cup B$ for input $b$, 
and an element of $A \cup C$ for input $c$.
As the following results show, 
the obtained value depends on 
the speed of computation of the individual Gray-digits.

\bigskip

Input $b$:
\begin{verbatim}
*GraySD> ed (takeD 50 (gtos (grayToD (1:1:[-1,-1..])))) >>= prints
 1-1-1-1-1-1-1-1-1-1-1-1-1-1-1-1-1-1-1-1-1-1-1-1 bot
\end{verbatim}

Input $c$:
\begin{verbatim}
*GraySD> ed (takeD 50 (gtos (grayToD (-1:1:[-1,-1..])))) >>= prints
-1 1 1 1 1 1 1 1 1 1 1 1 1 1 1 1 1 1 1 1 1 1 1 1 bot
\end{verbatim}

\bigskip

Input $a$ (demonstrating that the program can cope with an undefined digit):
\begin{verbatim}
*GraySD> ed (takeD 50 (gtos (grayToD (0:1:[-1,-1..])))) >>= prints
 0 0 0 0 0 0 0 0 0 0 0 0 0 0 0 0 0 0 0 0 0 0 0 0 bot
\end{verbatim}

\bigskip

Input $b$ with delayed first digit:
\begin{verbatim}
*GraySD> ed (takeD 50 (gtos (grayToD (2:1:[-1,-1..])))) >>= prints
 0 0 1-1-1-1-1-1-1-1-1-1-1-1-1-1-1-1-1-1-1-1-1-1 bot
\end{verbatim}

\bigskip

Same, but with more delayed first digit:
\begin{verbatim}
*GraySD> ed (takeD 50 (gtos (grayToD (10:1:[-1,-1..])))) >>= prints
 0 0 0 0 0 0 0 0 0 1-1-1-1-1-1-1-1-1-1-1-1-1-1-1 bot
\end{verbatim}

\bigskip

Input $1,1,1,\ldots$ which is the Gray code of  $1/3$:
\begin{verbatim}
*GraySD> ed (takeD 50 (gtos (grayToD ([1,1..])))) >>= prints
 1-1 1-1 1-1 1-1 1-1 1-1 1 0-1-1 1-1 1-1 1-1 1-1 bot
\end{verbatim}

\bigskip

Same, but with delayed first digit:
\begin{verbatim}
*GraySD> ed (takeD 50 (gtos (grayToD (2:[1,1..])))) >>= prints
 0 1 1-1 1-1 1-1 1-1 0 1 1-1 1-1 1-1 1-1 1-1 1-1 bot
\end{verbatim}

To see that the last two results are indeed approximations of signed digit 
representations of $1/3$, one observes that in the signed digit 
representation 
\verb|0 -1| means the same as \verb|-1 1|, %($0-1/2 = -1+1/2$),
and 
\verb|1 -1| means the same as \verb|0 1|, %($0+1/2 = 1-1/2$), 
so both results are equivalent to
\begin{verbatim}
 0 1 0 1 0 1 0 1 0 1 0 1 0 1 0 1 0 1 0 1 0 1 0 1 bot
\end{verbatim}
%\begin{verbatim}
% 1-1 1-1 1-1 1-1 1-1 1-1 1-1 1-1 1-1 1-1 1-1 1-1 bot
%\end{verbatim}
which denotes an approximation to $\sum_{n>0}(1/2)^{2n} = 1/3$.
%which is an approximation to \verb|concat (repeat [1,-1])|, a signed digit representation of $1/3$.

\bigskip

Note that since our experiments use the nondeterministic program \verb|ed|,
the results obtained with a different computer may differ from the ones 
included here.

Our theoretical results ensure that, whatever the results are, 
they will be correct.

\end{document}